\documentclass[11pt]{article}

\def\showauthornotes{0}
\def\showtableofcontents{1}
\def\showkeys{0}
\def\showdraftbox{0}
\def\showcolorlinks{1}
\def\usemicrotype{1}
\def\showfixme{0}

\usepackage{etex}

\usepackage{xspace,enumerate}

\usepackage[dvipsnames]{xcolor}

\usepackage[full]{textcomp}

\usepackage[american]{babel}

\usepackage{mathtools}

\usepackage{bm}

\usepackage{amsthm}

\newtheorem{theorem}{Theorem}[section]
\newtheorem*{theorem*}{Theorem}

\newtheorem{claim}[theorem]{Claim}
\newtheorem*{claim*}{Claim}

\newtheorem*{proposition*}{Proposition}
\newtheorem{lemma}[theorem]{Lemma}
\newtheorem*{lemma*}{Lemma}
\newtheorem{corollary}[theorem]{Corollary}

\newtheorem*{conjecture*}{Conjecture}

\newtheorem{fact}[theorem]{Fact}
\newtheorem*{fact*}{Fact}

\newtheorem*{hypothesis*}{Hypothesis}

\theoremstyle{definition}
\newtheorem{definition}[theorem]{Definition}

\newtheorem{remark}[theorem]{Remark}

\usepackage[letterpaper,
top=1in,
bottom=1in,
left=1in,
right=1in]{geometry}

\usepackage[varg]{txfonts} 
\renewcommand{\mathbb}{\varmathbb}

\ifnum\showkeys=1
\usepackage[color]{showkeys}
\fi

\ifnum\showcolorlinks=1
\usepackage[
pagebackref,
colorlinks=true,
urlcolor=blue,
linkcolor=blue,
citecolor=OliveGreen,
]{hyperref}
\fi

\ifnum\showcolorlinks=0
\usepackage[
colorlinks=false,
pdfborder={0 0 0}
]{hyperref}
\fi

\usepackage{prettyref}

\newcommand{\savehyperref}[2]{\texorpdfstring{\hyperref[#1]{#2}}{#2}}

\newrefformat{eq}{\savehyperref{#1}{\textup{(\ref*{#1})}}}
\newrefformat{lem}{\savehyperref{#1}{Lemma~\ref*{#1}}}
\newrefformat{def}{\savehyperref{#1}{Definition~\ref*{#1}}}
\newrefformat{thm}{\savehyperref{#1}{Theorem~\ref*{#1}}}
\newrefformat{cor}{\savehyperref{#1}{Corollary~\ref*{#1}}}
\newrefformat{cha}{\savehyperref{#1}{Chapter~\ref*{#1}}}
\newrefformat{sec}{\savehyperref{#1}{Section~\ref*{#1}}}
\newrefformat{app}{\savehyperref{#1}{Appendix~\ref*{#1}}}
\newrefformat{tab}{\savehyperref{#1}{Table~\ref*{#1}}}
\newrefformat{fig}{\savehyperref{#1}{Figure~\ref*{#1}}}
\newrefformat{hyp}{\savehyperref{#1}{Hypothesis~\ref*{#1}}}
\newrefformat{alg}{\savehyperref{#1}{Algorithm~\ref*{#1}}}
\newrefformat{rem}{\savehyperref{#1}{Remark~\ref*{#1}}}
\newrefformat{item}{\savehyperref{#1}{Item~\ref*{#1}}}
\newrefformat{step}{\savehyperref{#1}{step~\ref*{#1}}}
\newrefformat{conj}{\savehyperref{#1}{Conjecture~\ref*{#1}}}
\newrefformat{fact}{\savehyperref{#1}{Fact~\ref*{#1}}}
\newrefformat{prop}{\savehyperref{#1}{Proposition~\ref*{#1}}}
\newrefformat{prob}{\savehyperref{#1}{Problem~\ref*{#1}}}
\newrefformat{claim}{\savehyperref{#1}{Claim~\ref*{#1}}}
\newrefformat{relax}{\savehyperref{#1}{Relaxation~\ref*{#1}}}
\newrefformat{red}{\savehyperref{#1}{Reduction~\ref*{#1}}}
\newrefformat{part}{\savehyperref{#1}{Part~\ref*{#1}}}

\newcommand{\Sref}[1]{\hyperref[#1]{\S\ref*{#1}}}

\usepackage{nicefrac}

\newcommand{\half}{\nicefrac12}

\ifnum\usemicrotype=1
\usepackage{microtype}
\fi

\ifnum\showauthornotes=1
\newcommand{\Authornote}[2]{{\sffamily\small\color{red}{[#1: #2]}}}
\newcommand{\Authorcomment}[2]{{\sffamily\small\color{gray}{[#1: #2]}}}
\newcommand{\Authorstartcomment}[1]{\sffamily\small\color{gray}[#1: }

\newcommand{\Authorfnote}[2]{\footnote{\color{red}{#1: #2}}}
\newcommand{\Authorfixme}[1]{\Authornote{#1}{\textbf{??}}}
\newcommand{\Authormarginmark}[1]{\marginpar{\textcolor{red}{\fbox{\Large #1:!}}}}
\else
\newcommand{\Authornote}[2]{}
\newcommand{\Authorcomment}[2]{}
\newcommand{\Authorstartcomment}[1]{}

\newcommand{\Authorfnote}[2]{}
\newcommand{\Authorfixme}[1]{}
\newcommand{\Authormarginmark}[1]{}
\fi

\newcommand{\Bnote}{\Authornote{B}}

\ifnum\showfixme=0

\fi

\usepackage{boxedminipage}

\DeclareBoldMathCommand[heavy]\boldlangle{\left\langle}
\DeclareBoldMathCommand[bold]\boldrangle{\right\rangle}
\DeclareBoldMathCommand[bold]\boldlvert{\left\lVert}
\DeclareBoldMathCommand[heavy]\boldrvert{\right\rVert}

\newcommand{\paren}[1]{(#1)}
\newcommand{\Paren}[1]{\left(#1\right)}

\newcommand{\Bigparen}[1]{\Big(#1\Big)}

\newcommand{\abs}[1]{\lvert#1\rvert}

\newcommand{\card}[1]{\lvert#1\rvert}

\newcommand{\set}[1]{\{#1\}}
\newcommand{\Set}[1]{\left\{#1\right\}}

\newcommand{\norm}[1]{\lVert#1\rVert}

\newcommand{\Norm}[1]{\left\lVert#1\right\rVert}

\newcommand{\normt}[1]{\norm{#1}_2}

\newcommand{\snorm}[1]{\norm{#1}^2}

\newcommand{\iprod}[1]{\langle#1\rangle}

\newcommand{\Iprod}[1]{\left\langle#1\right\rangle}

\newcommand{\Esymb}{\mathbb{E}}
\newcommand{\Psymb}{\mathbb{P}}

\DeclareMathOperator*{\E}{\Esymb}

\DeclareMathOperator*{\ProbOp}{\Psymb}

\renewcommand{\Pr}{\ProbOp}

\newcommand{\ot}{\otimes}

\newcommand{\textparen}[1]{\text{(#1)}}

\ifx\because\undefined
\newcommand{\because}[1]{\textparen{because #1}}
\else
\renewcommand{\because}[1]{\textparen{because #1}}
\fi

\newcommand{\sge}{\succeq}
\newcommand{\sle}{\preceq}

\newcommand{\defeq}{\stackrel{\mathrm{def}}=}

\newcommand{\seteq}{\mathrel{\mathop:}=}

\newcommand{\from}{\colon}

\newcommand{\Mid}{\;\middle\vert\;}

\newcommand{\mper}{\,.}
\newcommand{\mcom}{\,,}

\newcommand\bdot\bullet

\ifx\mathds\undefined 

\else

\fi

\DeclareMathOperator{\Tr}{Tr}

\DeclareMathOperator{\poly}{poly}

\DeclareMathOperator{\argmax}{argmax}

\DeclareMathOperator{\supp}{supp}

\DeclareMathOperator{\sign}{sign}

\DeclareMathOperator{\rank}{rank}

\DeclareMathOperator{\tsum}{{\textstyle \sum}}

\newcommand{\N}{\mathbb N}
\newcommand{\R}{\mathbb R}
\newcommand{\C}{\mathbb C}

\newcommand{\GF}[1]{\mathbb F_{#1}}

\newcommand{\problemmacro}[1]{\texorpdfstring{\textsc{#1}}{#1}\xspace}

\newcommand{\maxcut}{\problemmacro{Max Cut}}

\newcommand{\cD}{\mathcal D}
\newcommand{\cE}{\mathcal E}

\newcommand{\cI}{\mathcal I}

\newcommand{\cL}{\mathcal L}

\newcommand{\cR}{\mathcal R}

\newcommand{\cU}{\mathcal U}
\newcommand{\cV}{\mathcal V}

\newcommand{\cX}{\mathcal X}

\newcommand{\bbR}{\mathbb R}

\renewcommand{\leq}{\leqslant}
\renewcommand{\le}{\leqslant}
\renewcommand{\geq}{\geqslant}
\renewcommand{\ge}{\geqslant}

\ifnum\showdraftbox=1
\newcommand{\draftbox}{\begin{center}
  \fbox{%
    \begin{minipage}{2in}%
      \begin{center}%
          \Large\textsc{Working Draft}\\%
        Please do not distribute%
      \end{center}%
    \end{minipage}%
  }%
\end{center}
\vspace{0.2cm}}
\else
\newcommand{\draftbox}{}
\fi

\let\epsilon=\varepsilon

\numberwithin{equation}{section}

\newcommand{\MYstore}[2]{%
  \global\expandafter \def \csname MYMEMORY #1 \endcsname{#2}%
}

\newcommand{\MYload}[1]{%
  \csname MYMEMORY #1 \endcsname%
}

\newcommand{\MYnewlabel}[1]{%
  \newcommand\MYcurrentlabel{#1}%
  \MYoldlabel{#1}%
}

\newcommand{\MYdummylabel}[1]{}

\newcommand{\torestate}[1]{%
  \let\MYoldlabel\label%
  \let\label\MYnewlabel%
  #1%
  \MYstore{\MYcurrentlabel}{#1}%
  \let\label\MYoldlabel%
}

\newcommand{\restatetheorem}[1]{%
  \let\MYoldlabel\label
  \let\label\MYdummylabel
  \begin{theorem*}[Restatement of \prettyref{#1}]
    \MYload{#1}
  \end{theorem*}
  \let\label\MYoldlabel
}

\newcommand{\restatelemma}[1]{%
  \let\MYoldlabel\label
  \let\label\MYdummylabel
  \begin{lemma*}[Restatement of \prettyref{#1}]
    \MYload{#1}
  \end{lemma*}
  \let\label\MYoldlabel
}

\newcommand{\restateprop}[1]{%
  \let\MYoldlabel\label
  \let\label\MYdummylabel
  \begin{proposition*}[Restatement of \prettyref{#1}]
    \MYload{#1}
  \end{proposition*}
  \let\label\MYoldlabel
}

\newcommand{\restatefact}[1]{%
  \let\MYoldlabel\label
  \let\label\MYdummylabel
  \begin{fact*}[Restatement of \prettyref{#1}]
    \MYload{#1}
  \end{fact*}
  \let\label\MYoldlabel
}

\newcommand{\restate}[1]{%
  \let\MYoldlabel\label
  \let\label\MYdummylabel
  \MYload{#1}
  \let\label\MYoldlabel
}

\newcommand{\addreferencesection}{
  \phantomsection
  \addcontentsline{toc}{section}{References}
}

\newcommand{\sse}{\subseteq}
\newcommand{\ra}{\rightarrow}
\newcommand{\e}{\epsilon}

\let\origparagraph\paragraph
\renewcommand{\paragraph}[1]{\origparagraph{#1.}}

\allowdisplaybreaks

\def\be{\begin{equation}}
\def\ee{\end{equation}}
\def\ba#1\ea{\begin{align}#1\end{align}}
\def\ban{\begin{align*}}
\def\ean{\end{align*}}

\def\benum{\begin{enumerate}}
\def\eenum{\end{enumerate}}
\def\bit{\begin{itemize}}
\def\eit{\end{itemize}}

\newcommand{\ifexp}[1]{}  

\let\pref=\prettyref

\renewcommand{\L}[1]{L_{#1}}

\newcommand{\spectralnorm}[1]{\norm{#1}_{\mathrm{spectral}}}
\newcommand{\sos}{\mathrm{sos}}
\newcommand{\dhell}{d_{H}}

\DeclareMathOperator*{\argmin}{arg\,min}
\newcommand{\ptoqnorm}[3]{\|#1\|_{#2\rightarrow #3}}
\newcommand{\tfnorm}[1]{\ptoqnorm{#1}{2}{4}}

\newcommand{\cp}{\textup{\textsf{cp}}}

\newcommand{\pef}{\mbox{p.e.f.}\xspace}
\newcommand{\pd}{\mbox{p.d.}\xspace}

\DeclareMathOperator*{\pE}{\Tilde{\Esymb}}

\newcommand{\bd}{\Phi}
\newcommand{\Holder}{H\"{o}lder}

\newcommand{\ASVP}{\textup{\textsf{ASVP}}\xspace}

\newcommand{\tsp}{\textup{\textsf{TSP}}\xspace}

\newcommand{\sval}{\nu}
\newcommand{\rval}{\Tilde{\nu}}

\newcommand{\univ}{\cU}
\newcommand{\elem}{\omega}

\newcommand{\topf}{\overline{f}}
\newcommand{\botf}{\underline{f}}

\newcommand{\spectral}{\mathrm{spectral}}

\newcommand{\Span}{\mathrm{span}}

\newcommand{\Sbar}{\overline{S}}

\newcommand{\Brandao}{Brand\~ao\xspace}

\newcommand{\otratio}[1]{\norm{#1}_{2:1}}

\title{Rounding Sum-of-Squares Relaxations}

\author{Boaz Barak\thanks{Microsoft Research.}\and Jonathan Kelner\thanks{Department of Mathematics, Massachusetts Institute of Technology.} \and David Steurer\thanks{Department of Computer Science, Cornell University.}}

\date{\today}

\begin{document}

\maketitle

\draftbox

\begin{abstract}
We present a general approach to rounding semidefinite programming relaxations obtained by the Sum-of-Squares method (Lasserre hierarchy).
Our approach is based on using the connection between these relaxations and the Sum-of-Squares proof system to transform
a \emph{combining algorithm}---an algorithm that maps a distribution over solutions into a (possibly weaker) solution---into
a \emph{rounding algorithm} that maps a solution of the relaxation to a solution of the original problem.

Using this approach, we obtain algorithms that yield improved results for natural variants of three well-known problems:

\begin{enumerate}

\item  We give a quasipolynomial-time algorithm that approximates $\max_{\norm{x}_2=1} P(x)$ within an additive factor
of $\e\norm{P}_{spectral}$ additive approximation, where $\e>0$ is a constant,
$P$ is a  degree $d=O(1)$, $n$-variate polynomial with nonnegative coefficients,
and $\norm{P}_{spectral}$ is the spectral norm of a matrix corresponding to $P$'s coefficients.
Beyond being of interest in its own right, obtaining such an approximation for general polynomials (with possibly negative coefficients) is a long-standing open question in quantum information theory,
and our techniques have already led to improved results in this area (Brand\~{a}o  and Harrow, STOC '13).

\item  We give a polynomial-time algorithm that,
given a subspace $V \subseteq \R^n$ of dimension $d$ that (almost) contains the characteristic function of a set of size $n/k$,
finds a vector $v\in V$ that satisfies $\E_i v_i^4 \geq \Omega(d^{-1/3} k(\E_i v_i^2)^2)$.
This is a natural analytical relaxation of the problem of finding the sparsest element in a subspace, and is also motivated by a connection to the Small Set Expansion problem shown by Barak et al. (STOC 2012).
In particular our results yield an improvement of the previous best known algorithms for small set expansion in a certain range of parameters.

\item We use this notion of $L_4$ vs. $L_2$ sparsity to obtain a polynomial-time algorithm with substantially improved guarantees for recovering a planted sparse vector $v$ in a random $d$-dimensional subspace of $\R^n$.
If $v$ has $\mu n$ nonzero coordinates, we can recover it with high probability whenever $\mu\leq O(\min(1,n/d^2))$.  In particular, when $d\leq \sqrt{n}$, this recovers a planted vector with up to $\Omega(n)$ nonzero coordinates.
When $d\leq n^{2/3}$,
our algorithm improves upon existing methods based on comparing the $L_1$ and $L_\infty$ norms, which intrinsically require $\mu \leq O\left(1/\sqrt{d}\right)$.

\end{enumerate}

\end{abstract}

\thispagestyle{empty}

\clearpage

\ifnum\showtableofcontents=1
{
\setcounter{tocdepth}{2}
\tableofcontents
 }
\fi
\thispagestyle{empty}
\setcounter{page}{0}
\clearpage

\section{Introduction}

\label{sec:intro}

Convex programming is the algorithmic workhorse behind many applications in computer science and other fields.
But its power is far from understood, especially in the case of \emph{hierarchies} of linear programming (LP) and semidefinite programming (SDP) relaxations.
These are systematic approaches to make a convex relaxation tighter by adding to it more constraints.
Various such hierarchies have been proposed independently by researchers from several communities~\cite{Shor87,SheraliA90,LovaszS91,Nesterov00,Parrilo00,Lasserre01}.
In general, these hierarchies are parameterized by a number $\ell$ called their \emph{level}.
For problems on $n$ variables, the hierarchy of the $\ell^{th}$ level can be optimized in $n^{O(\ell)}$ time, where for the typical domains used in CS (such as $\{0,1\}^n$ or the $n$-dimensional unit sphere), $n$ rounds correspond to the exact (or near exact) solution by brute force exponential-time enumeration.

There are several strong \emph{lower bounds} (also known as \emph{integrality gaps}) for these hierarchies,
in particular showing that $\omega(1)$ levels (and often even $n^{\Omega(1)}$ or $\Omega(n)$ levels) of many such hierarchies can't improve  by much on the known polynomial-time approximation guarantees for many NP hard problems,
including \textsf{SAT}, \textsf{Independent-Set}, \textsf{Max-Cut} and more~\cite{Grigoriev01,Grigoriev01b,AroraBLT06,delaVegaK07,Schoenebeck08,Tulsiani09,CharikarMM09,BenabbasGMT12,BhaskaraCVGZ12}.
Unfortunately, there are many fewer \emph{positive} results, and several of them only show that these hierarchies can match the performance of previously known (and often more efficient) algorithms,
rather than using hierarchies to get genuinely new algorithmic results.\footnote{
The book chapter \cite{ChlamtacT10} is a good source for several of the known upper and lower bounds, though it does not contain some of the more recent ones.}
For example, Karlin, Mathieu and Nguyen~\cite{KarlinMN11}
showed that  $\ell$ levels of the Sum of Squares hierarchy can approximate the \textsf{Knapsack} problem up to a factor of $1+1/\ell$, thus approaching the performance of the standard dynamic program.
Guruswami and Sinop~\cite{GuruswamiS11}  and (independently) Barak, Raghavendra, and Steurer~\cite{BarakRS11} showed that some SDP hierarchies can match the performance of the \cite{AroraBS10} algorithm for \textsf{Small Set Expansion} and \textsf{Unique Games},
and their techniques also gave improved results for some other problems (see also \cite{RaghavendraT12,AroraG11,AroraGS13}).
Chlamtac and Singh~\cite{ChlamtacS08} (building on~\cite{Chlamtac07}) used hierarchies to obtain some new approximation guarantees for the independent set problem in $3$-uniform hypergraphs.
Bhaskara, Charikar, Chlamtac, Feige, and Vijayaraghavan~\cite{BhaskaraCCFV10} gave an LP-hierarchy based approximation algorithm for the $k$-densest subgraph problem, although they also showed a purely combinatorial algorithm with the same performance.
The famous algorithm of Arora, Rao and Vazirani~\cite{AroraRV04} for \textsf{Sparsest Cut} can be viewed (in retrospect) as using a constant number of rounds of an SDP hierarchy to improve upon the performance of the basic LP for this problem.
Perhaps the most impressive use of super-constant levels of a  hierarchy to solve a new problem was the work of Brand\~{a}o, Christandl and Yard~\cite{BrandaoCY11}
who used an SDP hierarchy (first proposed by~\cite{DohertyPS04}) to give a quasipolynomial time algorithm for a variant of the \emph{quantum separability problem}
of testing whether a given density matrix corresponds to a separable (i.e., non-entangled) quantum state or is $\e$-far from all such states (see Section~\ref{sec:nonneg:intro}).


One of the reasons for this paucity of positive results is that we have relatively few tools to \emph{round} such convex hierarchies.
A \emph{rounding algorithm} maps a solution to the relaxation to a solution to the original program.\footnote{
While the name derives from the prototypical case of relaxing an integer program to a linear program by allowing the variables to take non-integer values, we use ``rounding algorithm'' for any mapping from relaxation solutions to actual solutions, even in cases where the actual solutions are themselves non-integer.}
In the case of a hierarchy, the relaxation solution satisfies more constraints, but we do not always know how to take advantage of this when rounding.
For example, \cite{AroraRV04}  used a very sophisticated analysis to get better rounding when the solution to a \textsf{Sparsest Cut} relaxation satisfies  a constraint known as triangle inequalities, but we have no general tools to use the additional constraints that come from higher levels of the hierarchies, nor do we know if these can help in rounding or not.
This lack of rounding techniques is particularly true for the \emph{Sum of Squares} (SOS, also known as \emph{Lasserre}) hierarchy~\cite{Parrilo00,Lasserre01}.\footnote{
  While it is common in the TCS community to use \emph{Lasserre} to describe the primal version of this SDP, and \emph{Sum of Squares (SOS)} to describe the dual, in this paper we use the more descriptive SOS name for both programs.
  We note that in all the applications we consider, strong duality holds, and so these programs are equivalent. }
This is the strongest variant of the canonical semidefinite programming hierarchies, and has recently shown promise to achieve tasks beyond the reach of weaker hierarchies~\cite{BarakBHKSZ12}.
But there are essentially no general rounding tools that take full advantage of its power.\footnote{
  The closest general tool we are aware of is the repeated conditioning methods of  \cite{BarakRS11,GuruswamiS11}, though these can be implemented in weaker hierarchies too and so do not seem  to use the full power of the SOS hierarchy.
  However, this technique does play a role in this work as well.}

\medskip 
In this work we propose a general approach to rounding SOS hierarchies, and instantiate this approach in two cases, giving new algorithms making progress on natural variants of two longstanding problems.
Our approach is based on the intimate connection between the SOS hierarchy and the ``Positivstellensatz''/``Sum of Squares'' proof  system.
This connection was used in previous work for either negative results~\cite{Grigoriev01,Grigoriev01b,Schoenebeck08}, or positive results for  specific instances~\cite{BarakBHKSZ12,ODonnellZ13,KauersOTZ14},
translating proofs of a bound on the actual value of these instances into proofs of bounds on the relaxation value.
In contrast, we use this connection to give explicit rounding algorithms for general instances of certain computational problems.

\subsection{The Sum of Squares hierarchy} \label{sec:sos:intro}

Our work uses the \emph{Sum of Squares (SOS)} semidefinite programming hierarchy and in particular its relationship with the Sum of Squares (or Positivstellensatz) proof system.
We now briefly review both the hierarchy and proof system.
See the introduction of~\cite{ODonnellZ13} and the monograph~\cite{Laurent09} for a more in depth discussion of these concepts and their history.
Underlying both the SDP and proof system is the natural approach to prove that a real polynomial $P$ is nonnegative via showing that it equals a \emph{sum of squares}: $P = \sum_{i=1}^k Q_i^2$ for some polynomials $Q_1,\ldots,Q_k$.
The question of when a nonnegative polynomial has such a ``certificate of non-negativity'' was studied by Hilbert  who realized this doesn't always hold and asked (as his $17{th}$ problem) whether a nonnegative polynomial is always a sum of squares of \emph{rational} functions.
This was proven to be the case by Artin, and also follows from the more general \emph{Positivstellensatz} (or ``Positive Locus Theorem'')~\cite{Krivine64,Stengle74}.

The Positivstellensatz/SOS proof system of Grigoriev and Vorobjov~\cite{GrigorievV01} is based on the Positivstellensatz as a way to refute the assertion that a certain set of polynomial equations
\begin{equation}
  P_1(x_1,\ldots,x_n)=\ldots=P_k(x_1,\ldots,x_n)=0 \label{eq:poly-eq}
\end{equation}
 can be satisfied by showing that there exists some polynomials $Q_1,\ldots,Q_k$ and a sum of squares polynomial $S$ such that
\begin{equation}
  \sum P_i Q_i = 1 + S \mper \label{eq:refutation}
\end{equation}
(\cite{GrigorievV01} considered inequalities as well, although in our context one can always restrict to equalities without loss of generality.)
One natural measure for the complexity of such proof is the \emph{degree} of the polynomials $P_1Q_1,\ldots,P_kQ_k$ and $S$.

The sum of squares semidefinite program was proposed independently by several authors~\cite{Shor87,Parrilo00,Nesterov00, Lasserre01}
One way to describe it is as follows.
If the set of equalities (\ref{eq:poly-eq}) is satisfiable then in particular there exists some random variable $X$ over $\R^n$ such that
\begin{equation}
  \E P_1(X_1,\ldots,X_n)^2 =\ldots= \E P_k(X_1,\ldots,X_n)^2=0 \mper \label{eq:poly-eq-ex}
\end{equation}
That is, $X$ is some distribution over the non-empty set of solutions to (\ref{eq:poly-eq}).

For every degree $\ell$, we can consider the linear operator $\cL = \cL_{\ell}$ that maps a polynomial $P$ of degree at most $\ell$ into the number $\E P(X_1,\ldots,X_n)$.
Note that by choosing the monomial basis, this operator can be described by a vector of length $n^{\ell}$, or equivalently, by an $n^{\ell/2} \times n^{\ell/2}$ matrix.
This operator satisfies the following conditions:
\begin{description}
\item[Normalization] If $P$ is the constant polynomial $1$ then $\cL P = 1$
\item[Linearity] $\cL (P+Q) = \cL P + \cL Q$ for every $P,Q$ of degree $\leq \ell$.
\item[Positivity] $\cL P^2 \geq 0$ for every $P$ of degree $\leq \ell/2$.
\end{description}
Following~\cite{BarakBHKSZ12}, we call a linear operator satisfying the above conditions \emph{a level $\ell$ pseudoexpectation function}, or $\ell$-\pef, and use the suggestive notation $\pE P(X)$  to denote $\cL P$.
Correspondingly we will sometimes talk about a \emph{level $\ell$ pseudodistribution} (or $\ell$-\pd ) $X$, by which we mean that there is an associated level $\ell$ pseudoexpectation operator.
Given the representation of $\cL$ as an $n^{\ell}$ dimension vector it is possible to efficiently check that it satisfies the above conditions efficiently,
and in particular the positivity condition corresponds to the fact that, when viewed as a matrix, $\cL$ is positive semidefinite.
Thus it is also possible to optimize over the set of operators satisfying these conditions in time $n^{O(\ell)}$,
and this optimization procedure is known as the SOS SDP hierarchy.
Clearly, as $\ell$ grows, the conditions become stricter.
In Appendix~\ref{app:toolkit} we collect some useful properties of these pseudoexpectations.
In particular one can show (see Corollary~\ref{cor:pef-constraint}) that if $\pE P^2(X)=0$ then  $\pE P(X)Q(X) = 0$ for every polynomial $Q$ (as long as $Q,P$ have degrees at most $\ell/2$).
Thus, if there is a refutation to (\ref{eq:poly-eq}) of the form (\ref{eq:refutation})  where all polynomials involved have degree at most $\ell$ then there would not exist a level $2\ell$ pseudoexpectation operator satisfying (\ref{eq:poly-eq-ex}).
This connection goes both ways,  establishing an equivalence between the degree of Positivstellensatz proofs and the level of the corresponding SOS relaxation.

Until recently, this relation was mostly used for \emph{negative} results, translating proof complexity lower bounds into integrality gap results for the SOS hierarchy~\cite{BarakBHKSZ12,ODonnellZ13,KauersOTZ14}.
However, in 2012 Barak, Brand\~{a}o, Harrow, Kelner, Steurer and Zhou~\cite{BarakBHKSZ12} used this relation for \emph{positive} results, showing that the SOS hierarchy can in fact solve some interesting instances of the \textsf{Unique Games} maximization problem that fool weaker hierarchies.
Their idea was to use the analysis of the previous works that proved these integrality gaps for weaker hierarchies.
Such proofs work by showing that (a) the weaker hierarchy outputs a large value on this particular instance but (b) the true value is actually small.
\cite{BarakBHKSZ12}'s insight was that oftentimes the proof of (b) only uses arguments that can be captured by the SOS/Positivstellensatz  proof system, and hence inadvertently shows that the SOS SDP value is actually small as well.
Some follow up works~\cite{ODonnellZ13,KauersOTZ14}  extended this to other instances, but all these results held for very specific instances which have been proven before to have small objective value.

\medskip In this work we use this relation to get some guarantees on the performance of the SOS SDP on \emph{general} instances.
We give a more detailed overview of our approach in Section~\ref{sec:overview}, but the high level idea is as follows.
For particular optimization problems, we design a ``rounding algorithm'' that on input the moment matrix of a distribution on \emph{actual solutions} achieving a certain value $\sval$,
outputs a solution with some value $\rval$ which is a function of $\sval$.
We call such an algorithm a \emph{combining algorithm}, since it ``combines'' a distribution over solutions into a single one.
(Note that the solution output by the combining algorithm need not be in the support of the distribution, and generally, when $\rval \neq \sval$, it won't be.)
We then ``lift'' the analysis of this combining algorithm into the SOS framework, by showing that all the arguments can be captured in this proof system.
This in turns implies that the algorithm would still achieve the value $\rval$ even if it is only given a \emph{pseudoexpectation} of the distribution of sufficiently high level $\ell$, and hence in fact this combining algorithm is a rounding algorithm for the level $\ell$ SOS hierarchy.
We apply this idea to obtain new results for two applications--- optimizing polynomials with nonnegative coefficients over the unit sphere, and finding ``analytically sparse'' vectors inside a subspace.

\begin{remark}[Relation to the Unique Games Conjecture.]
While the SOS hierarchy is relevant to many algorithmic applications, some recent work focused on its relation to Khot's Unique Games Conjecture (UGC)~\cite{Khot02}.
On a high level, the UGC implies that the basic semidefinite program is an optimal efficient algorithm for many problems, and hence in particular using additional constant or polylogarithmic levels of the SOS hierarchy will not help.
More concretely, as discussed in Section~\ref{sec:asvp:intro} below, the UGC is closely related to the question of how hard it is to find sparse (or ``analytically sparse'') vectors in a given subspace.
Our work shows how the SOS hierarchy can be useful in general, and in particular gives strong average-case results and
nontrivial worst-case results for finding sparse vectors in subspaces.
Therefore, it can be considered as giving some (far from conclusive) evidence that the UGC might be false.
\end{remark}




\subsection{Optimizing polynomials with nonnegative coefficients bounded spectral norm}
\label{sec:nonneg:intro}

Our first result yields an \emph{additive} approximation to this optimization problem for polynomials with nonnegative coefficients, when the value is scaled by the spectral norm of an associated matrix.
If $P$ is an $n$-variate degree-$t$ homogeneous polynomial with nonnegative coefficient, then it can be represented by a tensor $M \in \R^{n^t}$ such that $P(x) = M \cdot x^{\otimes t}$ for every $x\in\R^n$.
It is convenient to state our result in terms of this tensor representation:


\begin{theorem} \label{thm:nonneg:intro} There is an algorithm $A$, based on $O(t \log n/\e^2)$ levels of the SOS hierarchy, such that for every even\footnote{
The algorithm easily generalizes to polynomials of odd degree $t$ and to non-homogenous polynomials, see Remark~\ref{rem:odd-degree}.}
$t$ and nonnegative $M\in \R^{n^t}$,
\[
\max_{\norm{x}=1} M \cdot x^{\otimes t} \leq A(M) \leq  \max_{\norm{x}=1} M \cdot x^{\otimes t}  + \e\norm{M}_{spectral} \mcom
\]
where $\cdot$ denotes the standard dot product,  and $\norm{M}_{spectral}$ denotes the spectral norm of $M$, when considered as an $n^{t/2}\times n^{t/2}$ matrix.
\end{theorem}

Note that the algorithm of Theorem~\ref{thm:nonneg:intro} only uses a logarithmic number of levels, and thus it shows that this fairly natural polynomial optimization problem can be solved in quasipolynomial time, as opposed to the exponential time needed for optimizing over general polynomials of degree $>2$.
Indeed, previous work on the convergence of the Lasserre hierarchy for general polynomials~\cite{doherty2012convergence} can be described in our language here as trying to isolate a solution in the support of the distribution, and this generally requires a linear number of levels.
Obtaining the logarithmic bound here relies crucially on constructing a ``combined'' solution that is not necessarily in the support.
The algorithm is also relatively simple, and so serves as a good demonstration of our general approach.

\medskip 

\noindent \emph{Relation to quantum information theory.} An equivalent way to state this result is that we get an $\e$ additive approximation in the case that $\norm{M}_{spectral} \leq 1$, in which case the value $\max_{\norm{x}=1} M \cdot x^{\otimes t}$ is in the interval $[0,1]$.
This phrasing is particularly natural in the context of quantum information theory.
A general (potentially mixed) quantum state on $2\ell$-qubits is represented by a an $n^2 \times n^2$ \emph{density matrix} $\rho$ for $n=2^{\ell}$; $\rho$ is a positive semidefinite matrix and has trace $1$.
If $\rho$ is \emph{separable}, which means that there is no entanglement between the first $\ell$ qubits and the second $\ell$ qubits,
then $\rho = \E xx^* \otimes yy^*$ for some distribution over $x,y \in \C^n$, where $v^*$ denotes the complex adjoint operation.
If we further restrict the amplitudes of the states to be real, and enforce symmetry on the two halves, then this would mean that $\rho = \E x^{\otimes 4}$.
(All our results should generalize to states without those restrictions to symmetry and real amplitudes, which we make just to simplify the statement of the problem and the algorithm.)
A quantum \emph{measurement operator} over this space is an $n^2\times n^2$ matrix $M$ of spectral norm $\leq 1$.
The probability that the measurement accepts a state $\rho$ is $\Tr(M\rho)$.
Finding an algorithm that, given a measurement $M$, finds the separable state $\rho$ that maximizes this probability is an important question in quantum information theory
which amounts to finding a classical upper bound for the complexity class \textbf{QMA}$(2)$ of Quantum Merlin Arthur proofs with two independent provers~\cite{HarrowM13}.
Note that if we consider symmetric real states then this is the same as finding $\argmax_{\norm{x}=1} M \cdot x^{\otimes 4}$, and hence dropping the non-negativity constraint in our result would resolve this longstanding open problem.
There is a closely related dual form of this question, known as the \emph{quantum separability problem}, where one is given a quantum state $\rho$ and wants to find the test $M$ that maximizes
\begin{equation}
\Tr(M\rho) - \max_{\rho' \text{ separable }} \Tr(M\rho') \label{eq:quantum-seperability}
\end{equation}
or to simply distinguish between the case that this quantity is at least $\e$ and the case that $\rho$ is separable.
The best result known in this area is the paper \cite{BrandaoCY11} mentioned above, which solved the distinguishing variant of quantum separability problem in the case that measurements are restricted to so-called \emph{Local Operations and one-way classical communication} (one-way LOCC) operators.
However, they did not have an rounding algorithm, and in particular did not solve the problem of actually finding a separable state that maximizes the probability of acceptance of a given one-way LOCC  measurement.
The techniques of this work were used by Brand\~{a}o and Harrow~\cite{BrandaoH13} to solve the latter problem, and also greatly simplify the proof of \cite{BrandaoCY11}'s result,
which originally involved relations between several measures of entanglement proved in several papers.\footnote{The paper~\cite{BrandaoH13} was based on a previous version of this work~\cite{BarakKS12}
that contained only the results for nonnegative tensors. }

For completeness, in Appendix~\ref{sec:locc} we give a short proof of this result, specialized to the case of real vectors and polynomials of degree four (corresponding to quantum states of two systems, or two prover QMA proofs).
We also show in Appendix~\ref{sec:lowrank} that in the case the measurement satisfies the stronger condition of having its $\ell_2$ (i.e., Frobenius) norm be at most $1$,
there is a simpler and more efficient algorithm for estimating the maximum probability the measurement accepts a separable state,
giving an $\e$ additive approximation in $\poly(n)\exp(\poly(1/\e))$ time.
In contrast, \cite{BrandaoCY11}'s algorithm took quasipolynomial time even in this case.

\medskip

\noindent \emph{Relation to small set expansion.} Nonnegative tensors also arise naturally in some applications, and in particular in the setting of small set expansion for Cayley graphs over the cube, which was our original motivation to study them.
In particular, one corollary of our result is:

\begin{corollary}[Informally stated] \label{cor:cayley-sse}
There is an algorithm $A$, based on $\poly(K(G))\log n$ levels of the SOS hierarchy, that solves the \textsf{Small Set Expansion} problem on Cayley graphs $G$ over $\GF2^{\ell}$ (where $\ell = \log n$)
where $K(G)$ is a parameter bounding the spectral norm of an operator related to $G$'s top eigenspace.

\end{corollary}

We discuss the derivation and the meaning of this corollary in Section~\ref{sec:SSE} but note that the condition of having small value $K(G)$ seems reasonable.
Having $K(G)=O(1)$ implies that the graph is a small set expander, and in particular the known natural examples of Cayley graphs that are
small set expanders,  such as the noisy Boolean hypercube and the ``short code'' graph of~\cite{BarakGHMRS12} have $K(G)=O(1)$.
Thus a priori one might have thought that a graph that is hard to distinguish from small set expanders would have a small value of $K(G)$.


\subsection{Optimizing hypercontractive norms and finding analytically sparse vectors}
\label{sec:asvp:intro}

Finding a sparse nonzero vector inside a $d$ dimensional linear subspace $V\subseteq \R^n$ is a natural task arising in many applications in machine learning and optimization (e.g., see \cite{DemanetH13} and the references therein).
Related problems are known under many names including the  ``sparse null space'', ``dictionary learning'', ``blind source separation'', ``min unsatisfy'', and ``certifying restricted isometry property''  problems.
(These problems all have the same general flavor but differ on various details such as worst-case vs. average case, affine vs. linear subspaces, finding a single vector vs. a basis, and more.)
Problems of this type are often NP-hard, with some hardness of approximation results known, and conjectured average-case hardness (e.g., see~\cite{AroraBSS97,KoiranZ12,GottliebN10} and the references therein).

We consider a natural relaxation of this problem, which we call the \emph{analytically sparse vector}  problem ($\ASVP$), which assumes the input subspace (almost) contains an actually sparse $0/1$ vector,
but allows the algorithm to find a vector $v\in V$ that is only ``analytically sparse'' in the sense that $\norm{v}_4/\norm{v}_2$ is large.
More formally, for $q>p$ and $\mu>0$, we say that a vector $v$ is \emph{$\mu$ $L_q/L_p$-sparse} if $(\E_i v_i^q)^{1/q}/(E_i v_i^p)^{1/p} \geq \mu^{1/q-1/p}$.
That is, a vector is $\mu$ $L_q/L_p$-sparse if it has the same $q$-norm vs $p$-norm ratio as a $0/1$ vector of measure at most $\mu$.

This is a natural relaxation, and similar conditions have been considered in the past.
For example, Spielman, Wang, and Wright~\cite{SpielmanWW12} used in their work on dictionary learning a subroutine finds a vector $v$ in a subspace that maximizes the ratio $\norm{v}_{\infty}/\norm{v}_1$ (which can be done efficiently via $n$ linear programs).
However, because any subspace of dimension $d$ contains an $O(1/\sqrt{d})$ $L_{\infty}/L_1$-sparse vector, this relaxation can only detect the existence of vectors that are supported on less than $O(n/\sqrt{d})$ coordinates.
Some works have observed that the  $L_2/L_1$ ratio is a much better proxy for sparsity~\cite{ZibulevskyP01,DemanetH13}, but computing it is a non-convex optimization problem  for which no efficient algorithm is known.
Similarly, the $L_4/L_2$ ratio is a good proxy for sparsity for subspaces of small dimension (say $d =O(\sqrt{n})$) but it is non-convex, and it is not known how to efficiently optimize it.\footnote{
It seems that what makes our relaxation different from the original problem is not so much the qualitative issue of considering analytically sparse vectors as opposed to actually sparse vectors,
but the particular choice of the $L_4/L_2$ ratio, which on one hand seems easier (even if not truly easy) to optimize over than the $L_2/L_1$ ratio, but provides better guarantees than the $L_{\infty}/L_1$ ratio.
However, this choice does force us to restrict our attention to subspaces of low dimension, while in some applications such as certifying the restricted isometry property, the subspace in question
is often the kernel of a ``short and fat'' matrix, and hence is almost full dimensional. Nonetheless, we believe it should be possible to extend our results to handle subspaces of higher dimension, perhaps at the some mild cost
in the number of rounds.}

Nevertheless, because $\norm{v}_4^4$ is a degree $4$ polynomial, the problem of maximizing it for $v\in V$ of unit norm amounts to a polynomial maximization problem over the sphere, that has a natural SOS program.
Indeed, \cite{BarakBHKSZ12} showed that this program does in fact yield a good approximation of this ratio for random subspaces.
As we show in Section~\ref{sec:planted}, we can use this to improve upon the results of \cite{DemanetH13} and find planted sparse vectors in random subspaces that are of not too large a dimension:

\begin{theorem}\label{thm:planted:intro} There is a constant $c >0$ and an algorithm $A$,
based on $O(1)$-rounds of the SOS program, such that for every vector $v_0 \in \R^n$ supported on at most $cn\min(1,n/d^2)$ coordinates, if $v_1,\ldots,v_d$ are chosen independently at random from the Gaussian distribution
on $\R^n$, then given any basis for $V = \Span\{ v_0,\ldots, v_d\}$ as input, $A$ outputs an $\e$-approximation of $v_0$ in $\poly(n,\log(1/\e))$ time.
\end{theorem}
In particular, we note that  this recovers a planted vector with up to $\Omega(n)$ nonzero coordinates when $d\leq \sqrt{n}$, and it
can recover vectors with more than the $O(n/\sqrt{d})$ nonzero coordinates that are necessary for existing techniques whenever $d\ll n^{2/3}$.

\medskip\noindent   Perhaps more significantly, we prove the following nontrivial \emph{worst-case} bound for this problem:


\begin{theorem}\label{thm:ASVP:intro} There is a polynomial-time algorithm $A$, based on $O(1)$ levels of the SOS hierarchy, that on input a $d$-dimensional subspace $V\subseteq \R^n$ such that there is
a $0/1$-vector $v\in V$ with at most $\mu n$ nonzero coordinates, $A(V)$ outputs an $O(\mu d^{1/3})$ $L_4/L_2$-sparse vector in $V$.

Moreover, this holds even if $v$ is not completely inside $V$ but only satisfies $\norm{\Pi_V v}_2^2 \geq (1-\e)\norm{v}_2^2$, for some absolute constant $\e>0$, where $\Pi_V$ is the projector to $V$.
\end{theorem}


The condition that the vector is $0/1$ can be significantly relaxed, see Remark~\ref{rem:non-boolean}.
Theorem~\ref{thm:ASVP} is also motivated by the \textsf{Small Set Expansion} problem.
The current best known algorithms for \textsf{Small Set Expansion} and \textsf{Unique Games}~\cite{AroraBS10}
reduce these problems into the task of finding a sparse vector in a subspace, and then find this vector using brute force enumeration.
This enumeration is the main bottleneck in improving the algorithms' performance.\footnote{
This is the only step that takes super-polynomial time in \cite{AroraBS10}'s algorithm for \textsf{Small Set Expansion}.
Their algorithm for \textsf{Unique Games} has an additional divide and conquer step that takes subexponential time, but, in our opinion, seems less inherently necessary.
Thus we conjecture that if the sparse-vector finding step could be sped up then it would be possible to speed up the algorithm for both problems.
}
\cite{BarakBHKSZ12} showed that, at least for the \textsf{Small Set Expansion} question, finding an $L_4/L_2$ \emph{analytically sparse} vector would be good enough.
Using their work we obtain the following corollary of Theorem~\ref{thm:ASVP:intro}:

\begin{corollary}[Informally stated] \label{cor:sse-approx}
There is an algorithm that given an $n$-vertex graph $G$ that contains a set $S$ of size $o(n/d^{1/3})$ with expansion at most $\e$, outputs a set $S'$ of measure $\delta=o(1)$ with expansion bounded away from $1$, i.e., $\Phi(S)\le 1-\Omega(1)$,
where $d$ is the dimension of the eigenspace of $G$'s random walk matrix corresponding to eigenvalues larger than $1-O(\e)$.
\end{corollary}

The derivation and meaning of this result is discussed in Section~\ref{sec:SSE}.
We note that this is the first result that gives an approximation of this type to the small set expansion in terms of the dimension of the top eigenspace, as opposed to an approximation that is polynomial in the number of vertices.

\subsection{Related work} \label{sec:relwork}

Our paper follows the work of \cite{BarakBHKSZ12}, that used the language of pseudoexpectation to argue that the SOS hierarchy can solve specific interesting instances of \textsf{Unique Games},
and perhaps more importantly, how it is often possible to almost mechanically ``lift'' arguments about actual distributions to the more general setting of pseudodistribution.
In this work we show how the same general approach be used to obtain positive results for general instances.

The fact that LP/SDP solutions can be viewed as expectations of distributions is well known,
and several rounding algorithms can be considered as trying to ``reverse engineer'' a relaxation solution to get a good distribution over actual solutions.

Techniques such as randomized rounding, the hyperplane rounding of~\cite{GoemansW95}, and the rounding for TSP~\cite{GharanSS11,AnKS12} can all be viewed in this way.
One way to summarize the conceptual difference between our techniques and those approaches is that these previous algorithms often considered the relaxation solution as giving moments of an \emph{actual} distribution on \emph{``fake''} solutions.
For example, in \cite{GoemansW95}'s \maxcut algorithm, where actual solutions are modeled as vectors in $\{\pm 1\}^n$,
the SDP solution is treated as the moment matrix of a Gaussian distribution over real vectors that are not necessarily $\pm 1$-valued.
Similarly in the \tsp setting one often considers the LP solution to yield moments of a distribution over spanning trees that are not necessarily \tsp tours.
In contrast, in our setting we view the solution as providing moments of a \emph{``fake''} distribution on  \emph{actual} solutions.

Treating solutions explicitly as ``fake distributions'' is prevalent in the literature on \emph{negative results} (i.e., integrality gaps) for LP/SDP hierarchies.
For hierarchies weaker than SOS, the notion of ``fake'' is different, and means that there is a collection of local distributions,
one for every small subset of the  variables, that are consistent with one another but do not necessarily correspond to any global distribution.
Fake distributions are also used in some positive results for hierarchies, such as~\cite{BarakRS11,GuruswamiS11},
but we make this more explicit, and, crucially, make much heavier use of the tools afforded by the Sum of Squares relaxation.

The notion of a ``combining algorithm'' is related to the notion of \emph{polymorphisms}~\cite{BulatovJK05} in the study of constraint satisfaction problems.
A polymorphism is a way to combine a number of satisfying assignments of a CSP into a different satisfying assignments, and some relations between
polymorphism, their generalization to approximation problems, rounding SDP's are known (e.g., see the talk~\cite{Raghavendra10}).
The main difference is polymorphisms operate on each bit of the assignment independently, while we consider here combining algorithms that can be very global.

In a follow up (yet unpublished) work, we used the techniques of this paper to obtain improved results for the \emph{sparse dictionary learning} problem,  recovering a set of vectors $x_1,\ldots,x_m \in \R^n$ from random samples of $\mu$-sparse linear combinations of them for any $\mu = o(1)$, improving upon previous results that required $\mu \ll 1/\sqrt{n}$~\cite{SpielmanWW12,AroraGM13,AgarwalA0NT13}.


\subsection{Organization of this paper}

In Section~\ref{sec:overview} we give a high level overview of our general approach, as well as proof sketches for (special cases of) our  main results.
Section~\ref{sec:non-neg} contains the proof of Theorem~\ref{thm:nonneg:intro}--- a quasipolynomial time algorithm to optimize polynomials with  nonnegative coefficients over the sphere.
Section~\ref{sec:asvp} contains the proof of Theorem~\ref{thm:ASVP:intro}--- a polynomial time algorithm for an $O(d^{1/3})$-approximation of the ``analytical sparsest vector in a subspace'' problem.
In Section~\ref{sec:planted} we show how to use the notion of analytical sparsity to solve the question of finding a ``planted'' sparse vector in a random subspace.
Section~\ref{sec:SSE} contains the proofs of Corollaries~\ref{cor:cayley-sse} and~\ref{cor:sse-approx} of our results to the small set expansion problem.
Appendix~\ref{app:toolkit} contains certain technical lemmas showing that pseudoexpectation operators obey certain inequalities that are true for actual expectations.
Appendix~\ref{sec:locc} contains a short proof (written in classical notation, and specialized to the real symmetric setting) of \cite{BrandaoCY11,BrandaoH13}'s result that the SOS hierarchy yields a good
approximation to the acceptance probability of QMA(2) verifiers / measurement operators that have bounded one-way LOCC norm. 
Appendix~\ref{sec:lowrank} shows a simpler algorithm for the case that the verifier satisfies the stronger condition of a bounded $L2$ (Frobenius) norm.
For the sake of completeness, Appendix~\ref{app:sse-vs-norm} reproduces the proof from \cite{BarakBHKSZ12} of the relation between hypercontractive norms and small set expansion.
Our papers raises many more questions than it answers, and some discussion of those appears in Section~\ref{sec:discussion}.

\subsection{Notation}

\begin{small} 
\paragraph{Norms and inner products} We will use linear subspaces of the form $V=R^\univ$ where $\univ$ is a finite set with an associated measure $\mu:\univ\rightarrow [0,\infty]$.
The $p$-norm of a vector $v\in V$ is defined as $\norm{v}_p = \left( \sum_{\elem\in\univ} \mu(\elem) |v_\elem|^p \right)^{1/p}$.
Similarly, the inner product of $v,w \in V$ is defined as $\iprod{u,v} = \sum_{\elem \in \univ}  \mu(\elem) u_{\elem}v_{\elem}$.
We will only use two measures in this work: the \emph{counting measure}, where $\mu(\elem)=1$ for every $\elem \in \univ$, and the \emph{uniform measure}, where $\mu(\elem)=1/|\univ|$ for all $\elem\in\univ$.
(The norms corresponding to this measure are often known as the \emph{expectation} norms.)

We will  use vector notation (i.e., letters such as $u,v$, and indexing of the form $u_i$) for elements of subspaces with the counting measure, and function notation (i.e., letters such as $f,g$ and indexing of the form $f(x)$) for elements of subspaces with the uniform measure. The dot product notation $u\cdot v$ will be used exclusively for the inner product with the counting measure.

\paragraph{Pseudoexpectations} We use the notion of \emph{pseudoexpectations} from~\cite{BarakBHKSZ12}.
A \emph{level $\ell$ pseudoexpectation function} ($\ell$-\pef) $\pE_{\cX}$ is an operator mapping a polynomial $P$ of degree at most $\ell$ into a number denoted by $\pE_{x\sim\cX} P(x)$ and satisfying the linearity, normalization, and positivity conditions as stated in Section~\ref{sec:sos:intro}.
We sometimes refer to $\cX$ as a \emph{level $\ell$ pseudodistribution} ($\ell$-\pd ) by which we mean that there exists an associated pseudoexpectation operator.\footnote{
In the paper~\cite{BarakBHKSZ12} we used the name \emph{level $\ell$ fictitious random variable} for $\cX$, but we think the name pseudodistribution is better as it is more analogous to the name pseudoexpectation. The name ``pseudo random variable'' would of course be much too confusing.}
If $P,Q$ are polynomials of degree at most $\ell/2$, and $\pE_{\cX}$ is an $\ell$-\pef,
we say that $\pE_{\cX}$ is \emph{consistent} with the constraint $P(x) \equiv 0$ if it satisfies $\pE_{x\sim\cX} P(x)^2 = 0$.
We say that it is consistent with the constraint $Q(x) \geq 0$, if it consistent with the constraint $Q(x) - S(x) \equiv 0$ for some polynomial $S$ of degree $\leq \ell/2$ which is a \emph{sum of squares}. 

(In the context of optimization, to enforce the inequality constraint $Q(x) \geq 0$, it is always possible to add an auxiliary variable $y$ and then enforce the equality constraint $Q(x)-y^2 \equiv 0$.)
Appendix~\ref{app:toolkit} contains several useful facts about pseudoexpectations.

\end{small}

\section{Overview of our techniques} \label{sec:overview}


Traditionally to design a mathematical-programming based approximation algorithm for some optimization problem $O$, one first decides what the relaxation is---
i.e., whether it is a linear program, semidefinite program, or some other convex program, and what constraints to put in.
Then, to demonstrate that the value of the program is not too far from the actual value,
one designs a \emph{rounding algorithm} that maps a solution of the convex program into a solution of the original problem of approximately the same value.
Our approach is conceptually different--- we design the rounding algorithm first, analyze it, and only then come up with the relaxation.

Initially, this does not seem to make much sense--- how can you design an algorithm to round solutions of a relaxation when you don't know what the relaxation is?
We do this by considering an idealized version of a rounding algorithm which we call a \emph{combining algorithm}.
Below we discuss this in more detail but roughly speaking, a combining algorithm maps a distribution over \emph{actual solutions} of $O$ into a single solution (that may or may not be part of the support of this distribution).
This is a potentially much easier task than rounding relaxation solutions, and every rounding algorithm yields a combining algorithm.
In the other direction, every combining algorithm yields a rounding algorithm for \emph{some} convex programming relaxation, but in general that relaxation could be of exponential size.
Nevertheless, we show that in several interesting cases, it is possible to transform a combining algorithm into a rounding algorithm for a not too large relaxation that we can efficiently optimize over,
thus obtaining a feasible approximation algorithm.
The main tool we use for that is the \emph{Sum of Squares} proof system, which allows to lift certain arguments from the realm of combining algorithms to the realm of rounding algorithms.

We now explain more precisely the general approach, and then give an overview of how we use this approach for our two applications---
finding ``analytically sparse'' vectors in subspaces, and optimizing polynomials with nonnegative coefficients over the sphere.

\medskip

Consider a general optimization problem of minimizing some objective function in some set $S$, such as the $n$ dimensional Boolean hypercube or the unit sphere.
A \emph{convex relaxation} for this problem consists of an embedding that maps elements in $S$ into elements in some convex domain,
and a suitable way to generalize the objective function to  a convex function on this domain.
For example, in linear programming relaxations we typically embed $\{0,1\}^n$ into the set $[0,1]^n$,
while in semidefinite programming relaxations we might embed $\{0,1\}^n$ into the set of $n \times n$ positive semidefinite matrices using the map $x \mapsto X$ where $X_{i,j}=x_ix_k$.
Given this embedding, we can use convex programming to find the element in the convex domain that maximizes the objective,
and then use a \emph{rounding algorithm} to map this element back into the domain $S$ in a way that approximately preserves the objective value.

A \emph{combining algorithm} $C$ takes as input a \emph{distribution} $\cX$ over solutions in $S$ and maps it into a single element $C(\cX)$ of $S$,
such that the objective value of $C(\cX)$ is approximately close to the expected objective value of a random element in $\cX$.
Every rounding algorithm $R$  yields a combining algorithm $C$.
The reason is that if there is some embedding $f$ mapping elements in $S$ into some convex domain $T$, then for every distribution $\cX$ over $S$,
we can define  $y_{\cX}$ to be $\E_{x\in\cX} f(x)$.
By convexity, $y_{\cX}$ will be in $T$ and its objective value will be at most the average objective value of an element in $\cX$.
Thus if we define $C(\cX)$ to output $R(y_{\cX})$ then $C$ will be a combining algorithm with approximation guarantees at least as good as $R$'s.

In the other direction, because the set of distributions over $S$ is convex and can be optimized over by an $O(|S|)$-sized linear program, every combining algorithm can be viewed as a rounding algorithm for this program.
However, $|S|$ is typically exponential in the bit description of the input, and hence this is not a very useful program.
In general, we cannot improve upon this, because there is always a trivially lossless combining algorithm that ``combines'' a distribution $\cX$ into a single solution $x$ of the same expected value
by simply sampling $x$ from $\cX$ at random.
Thus for problems where getting an exact value is exponentially hard, this combining algorithm cannot be turned into a rounding algorithm for a subexponential-sized efficiently-optimizable convex program.
However it turns out that at least in some cases, \emph{nontrivial} combining algorithms can be turned into a rounding algorithm for an \emph{efficient} convex program.
A nontrivial combining algorithm $C$ has the form $C(\cX) = C'(M(\cX))$  where $C'$ is an efficient (say polynomial or quasipolynomial time) algorithm
and $M(\cX)$ is a short (say polynomial or quasipolynomial size)  \emph{digest} of the distribution $\cX$.
In all the cases we consider, $M(\cX)$ will consist of all the moments up to some level $\ell$ of the random variable $\cX$, or some simple functions of it.
That is, typically $M(\cX)$ is a vector in $\R^{m^{\ell}}$ such that for every $i_1,\ldots,i_{\ell} \in [m]$, $M_{i_1,\ldots,i_{\ell}} = \E_{x\sim\cX} x_{i_1}\cdots x_{i_{\ell}}$.
We do not have a general theorem showing that any nontrivial combining algorithm can be transformed into a rounding algorithm for an efficient relaxation.
However, we do have a fairly general ``recipe'' to use  the \emph{analysis} of nontrivial combining algorithms to transform them into rounding algorithms.
The key insight is that many of the tools used in such analyses, such as the Cauchy--Schwarz and \Holder\  inequalities, and other properties of distributions,
fall under the ``Sum of Squares'' proof framework, and hence can be shown to hold even when the algorithm is applied not to actual moments but to so-called ``pseudoexpectations''
that arise from the SOS semidefinite programming hierarchy.


\medskip

We now turn to giving a high level overview of our  results.
For the sake of presentations, we focus on certain special cases of these two applications, and even for these cases omit many of the proof details and only provide rough sketches of the proofs.
The full details can be found in Sections~\ref{sec:planted}, \ref{sec:asvp} and~\ref{sec:non-neg}.

\subsection{Finding a planted sparse vector in a random low-dimensional subspace}

We consider the following natural problem, which was also studied by Demanet and Hand~\cite{DemanetH13}.
Let $f_0 \in \R^\univ$  be a sparse function over some universe $\univ$ of size $n$.
That is, $f_0$ is supported on at most $\mu n$ coordinates for some $\mu = o(1)$.
Let $V$ be the subspace spanned by $f_0$ and $d$ random (say Gaussian) functions $f_1,\ldots,f_d \in \R^\univ$.
Can we recover $f_0$ from any basis for $V$?

Demanet and Hand showed that if $\mu$ is very small, specifically $\mu \ll 1/\sqrt{d}$, then $f_0$ would be the most $L_{\infty}/L_1$-sparse function in $V$, and hence (as mentioned above) can be recovered
efficiently by running $n$ linear programs.
The SOS framework yields a natural and easy to describe algorithm for recovering $f_0$ as long as $\mu$ is a sufficiently small constant and the dimension $d$ is at most $O(\sqrt{n})$.
The algorithm uses the SOS program for finding the most $L_4/L_2$-sparse function in $V$, which, as mentioned above, is simply the polynomial optimization problem of maximizing $\norm{f}_4^4$ over $f$
in the intersection of $V$ and  the unit Euclidean sphere.

Since $f_0$ itself is in particular $\mu$ $L_4/L_2$-sparse , the optimum for the program is at least $\mu$.
Thus a combining algorithm would get as input a distribution $\cD$ over functions $f\in V$ satisfying $\norm{f}_2=1$ and $\norm{f}_4^4 \geq 1/\mu$, and need to output a vector closely correlated with $f_0$.\footnote{
Such a closely correlated vector can be corrected to output $f_0$ exactly, see Section~\ref{sec:planted}.}
(We use here the \emph{expectation} norms, namely $\norm{f}_p^p = \E_\elem |f(\elem)|^p$.)
For simplicity, assume that the $f_i$'s are orthogonal to $f_0$ (they are nearly orthogonal, and so everything we say below will still hold up to a sufficiently good approximation, see Section~\ref{sec:planted}).
In this case, we can write every $f$ in the support of $\cD$ as  $f = \iprod{f_0,f} f_0 + f'$ where  $f' \in V' = \Span\{ f_1,\ldots, f_d \}$.
It is not hard to show using standard concentration of measure results (see e.g., \cite[Theorem 7.1]{BarakBHKSZ12}) that if $d=O(\sqrt{n})$ then every $f'\in V'$ satisfies
\begin{equation}
\norm{f'}_4 \leq C \norm{f'}_2 \mcom \label{eq:2to4random}
\end{equation}
for some constant $C$.
Therefore using triangle inequality, and using the fact that $\norm{f'}_2 \leq \norm{f}_2=1$, it must hold that
\begin{equation}
\mu^{-1/4} \leq \norm{f}_4 \leq \iprod{f,f_0}\mu^{-1/4} + C \label{eq:planted-random1}
\end{equation}
or
\begin{equation}
\iprod{f,f_0} \geq 1 - C\mu^{1/4} = 1 - o(1)  \label{eq:planted-random2}
\end{equation}
for $\mu = o(1)$.

In particular this implies that if we apply a singular value decomposition (SVD) to the second moment matrix $D$ of $\cD$ (i.e., $D = \E_{f\in\cD} f^{\otimes 2}$) then
the top eigenvector will have  $1-o(1)$ correlation with $f_0$, and hence we can simply output it as our solution.

To make this combining algorithm into a rounding algorithm we use the result of \cite{BarakBHKSZ12} that showed that (\ref{eq:2to4random}) can actually be proven
via a sum of squares argument. Namely they showed that there is a degree $4$ sum of squares polynomial $S$ such that
\begin{equation}
\norm{\Pi' f}_4^4 + S(f) = C^4\norm{f}_2^4 \mper \label{eq:sos-2to4random}
\end{equation}

(\ref{eq:sos-2to4random}) implies that even if $\cD$ is merely a \emph{pseudodistribution} then it must satisfy (\ref{eq:2to4random}).
(When the latter is raised to the fourth power to make it a polynomial inequality.)
We can then essentially follow the argument, proving a version of (\ref{eq:planted-random1}) raised to the $4^\text{th}$ power by appealing to the
fact that pseudodistributions satisfy \Holder 's inequality, (Corollary~\ref{cor:pef-holder-inside}) and hence deriving that $\cD$ will satisfy (\ref{eq:planted-random2}), with possibly slightly worse constants,
even when it is only a pseudodistribution.

In Section~\ref{sec:planted}, we make this precise and extend the argument to obtain nontrivial (but weaker) guarantees when $d\geq \sqrt{n}$.  We then show how to use an additional correction step to recover the original function $f_0$ up to arbitrary accuracy, thus boosting our approximation of $f_0$ into an essentially exact one.

\subsection{Finding ``analytically sparse'' vectors in general subspaces}
\label{sec:overview:asvp}

We now outline the ideas behind the proof of Theorem~\ref{thm:ASVP}--- finding analytically sparse vectors in \emph{general} (as opposed to random) subspaces.
This is a much more challenging setting than random subspaces, and indeed our algorithm and its analysis is more complicated (though still only uses a constant number of SOS levels),
and at the moment, the approximation guarantee we can prove is quantitatively weaker.
This is the most technically involved result in this paper,
and so the reader may want to skip ahead to Section~\ref{sec:overview:nonnegative} where we give an overview of the simpler result of optimizing over polynomials with nonnegative coefficients.

We consider the special case of Theorem~\ref{thm:ASVP} where we try to distinguish between a YES case where there is a $0/1$ valued $o(d^{-1/3})$-sparse function that is completely contained in the input subspace,
and a NO case where every function in the subspace has its four norm bounded by a constant times its two norm.
That is, we suppose that we are given some subspace $V\subseteq \R^\univ$ of dimension $d$ and a distribution $\cD$ over functions $f:\univ\rightarrow \{0,1\}$ in $V$ such that
$\Pr_{\elem\in\univ }[ f(\elem)=1] = \mu$ for every $f$ in the support of $\cD$, and $\mu = o(d^{-1/3})$.
The goal of our combining algorithm to output some function $g\in V$ such that $\norm{g}_4^4 = \E_\elem g(\elem)^4 \gg (\E_\elem g(\elem)^2 )^2 = \norm{g}_2^4$.
(Once again, we use the \emph{expectation} inner product and norms, with uniform measure over $\univ$.)

Since the $f$'s correspond to sets of measure $\mu$, we would expect the inner product $\iprod{f,f'}$ of a typical pair $f,f'$ (which equals the measure of the intersection of the corresponding sets) to be roughly $\mu^2$.
Indeed, one can show that if the average inner product $\iprod{f,f'}$ is $\omega(\mu^2)$ then it's easy to find such a desired function $g$.
Intuitively, this is because in this case the distribution $\cD$ of sets does not have an equal chance to contain all the elements in $\univ$, but rather there is some set $I$ of  $o(|\univ|)$ coordinates which is favored by $\cD$.
Roughly speaking, that would mean that a random linear combination $g$ of these functions would have most of its mass concentrated inside this small set $I$, and hence satisfy $\norm{g}_4 \gg \norm{g}_2$.
But it turns out that letting $g$ be a random gaussian function matching the first two moments of $\cD$ is equivalent to taking such a random linear combination, and so our combining algorithm can obtain this $g$ using moment information alone.

Our combining algorithm will also try all $n$ \emph{coordinate projection} functions.
That is, let $\delta_\elem$ be the function such that $\delta_\elem(\elem')$ equals $n=|\univ|$ if $\elem=\elem'$ and equals $0$ otherwise, (and hence under our expectation inner product $f(\elem) = \iprod{f,\delta_\elem}$).
The algorithms will try all functions of the form $\Pi \delta_u$ where $\Pi$ is the projector to the subspace $V$.
Fairly straightforward calculations show that $2$-norm squared of such a function is expected to be $(d/n)\norm{\delta_\elem}_2^2 = d$, and it turns out in our setting we can assume that the norm is well concentrated around this expectation (or else we'd be able to find a good solution in some other way).
Thus, if coordinate projection fails then it must hold that
\begin{equation}
  O(d^2) = O(\E_\elem \norm{\Pi \delta_\elem}_2^4) \geq \E_\elem \norm{\Pi \delta_\elem}_4^4 = \E_{\elem,\elem'} \iprod{\Pi \delta_\elem,\delta_{\elem'}}^4 \label{eq:asvp:1} \mper
\end{equation}
It turns out that (\ref{eq:asvp:1}) implies some nontrivial constraints on the distribution $\cD$.
Specifically we know that
\[
\mu = \E_{f \sim \cD} \norm{f}_4^4 = \E_{f\sim \cD, \elem\in \univ} \iprod{f,\delta_\elem}^4 \mper
\]
But since $f = \Pi f$ and $\Pi$ is symmetric, the RHS is equal to
\[
\E_{f\sim \cD, \elem\in \univ} \iprod{f,\Pi \delta_\elem}^4 = \iprod{ \E_{f \sim \cD} f^{\otimes 4} , \E_{\elem\in \univ} (\Pi \delta_\elem)^{\otimes 4}} \leq \norm{\E_{f \sim \cD} f^{\otimes 4}}_2 \norm{\E_{\elem\in\univ} (\Pi \delta_\elem)^{\otimes 4}}_2  \mcom
\]
where the last inequality uses Cauchy--Schwarz.
If we square this inequality we get that
\[
\mu^2 \leq \iprod{\E_{f \sim \cD} f^{\otimes 4},\E_{f \sim \cD} f^{\otimes 4}}\iprod{\E_{\elem\in\univ]} (\Pi \delta_\elem)^{\otimes 4},\E_{\elem\in \univ} (\Pi \delta_\elem)^{\otimes 4}} = \left(\E_{f,f' \sim \cD} \iprod{f,f'}^4\right)\left(\E_{\elem,\elem'} \iprod{\Pi \delta_\elem,\Pi \delta_{\elem'}}^4 \right) \mper
\]
But since is a projector satisfying $\Pi = \Pi^2$, we can use (\ref{eq:asvp:1}) and obtain that
\[
\Omega(\mu^2/d^2) \leq \E_{f,f' \sim \cD} \iprod{f,f'}^4 \mper
\]
Since $d = o(\mu^{-3})$ this means that
\begin{equation}
\E_{f,f' \sim \cD} \iprod{f,f'}^4 \gg \mu^8  \mper \label{eq:asvp:2}
\end{equation}

Equation~(\ref{eq:asvp:2}), contrasted with the fact that $\E_{f,f' \sim \cD} \iprod{f,f'} = O(\mu^2)$,
means that the inner product of two random functions in $\cD$ is somewhat ``surprisingly unconcentrated'', which seems to be a nontrivial piece of information about $\cD$.\footnote{
Interestingly, this part of the argument does not require $\mu$ to be $o(d^{-1/3})$, and some analogous ``non-concentration'' property of $\cD$ can be shown to hold for a hard to round $\cD$ for any $\mu=o(1)$.
However, we currently know how to take advantage of this property to obtain a combining algorithm only in the case that $\mu \ll d^{-1/3}$.}
Indeed, because the $f$'s are nonnegative functions, if we pick a random $u$ and consider the distribution $\cD_u$ where the probability of every function is reweighed proportionally to $f(u)$, then intuitively that should increase the probability of pairs with large inner products.
Indeed, as we show in Lemma~\ref{lem:pseudo-holder}, one can use \Holder 's inequality to prove that there exist $\elem_1,\ldots,\elem_4$ such that
under the distribution $\cD'$ where every element $f$ is reweighed proportionally to $f(\elem_1)\cdots f(\elem_4)$,
it holds that
\begin{equation}
  \E_{f,f' \sim \cD'} \iprod{f,f'} \geq \left( \E_{f,f' \sim \cD} \iprod{f,f'}^4 \right)^{1/4} \mper  \label{eq:asvp:3}
\end{equation}
(\ref{eq:asvp:3}) and (\ref{eq:asvp:2}) together imply that $\cE_{f,f'\sim \cD'} \iprod{f,f'} \gg \mu^2$, which, as mentioned above, means that we can find a function  $g$ satisfying $\norm{g}_4 \gg \norm{g}_2$ by taking a gaussian function matching the first two moments of $\cD'$.

Once again, this combining algorithm can be turned into an algorithm that uses $O(1)$ levels of the SOS hierarchy.
The main technical obstacle (which is still not very hard) is to prove another appropriate generalization of \Holder 's inequality for pseudoexpectations (see Lemma~\ref{lem:pseudo-holder}).
Generalizing to the setting that in the YES case the function is only approximately in the vector space is a bit more cumbersome.
We need to consider apart from $f$ the function $\topf$ that is obtained by first projecting $f$ to the subspace and then ``truncating'' it by rounding each coordinate where $f$ is too small to zero.
Because this truncation operation is not a low degree polynomial,  we include the variables corresponding to $\topf$ as part of the relaxation,
and so our pseudoexpectation operator also contains the moments of these functions as well.

\subsection{Optimizing polynomials with nonnegative coefficients}
\label{sec:overview:nonnegative}

We now consider the task of maximizing a polynomial with nonnegative coefficients over the sphere, namely proving Theorem~\ref{thm:nonneg}.
We consider the special case of Theorem~\ref{thm:nonneg}  where the polynomial is of degree $4$.
That is, we are given a parameter  $\e>0$ and an $n^2\times n^2$ nonnegative matrix $M$ with spectral norm at most $1$ and want to find an $\e$ additive approximation to the maximum of
\begin{equation}
\sum_{i,j,k,l} M_{i,j,k,l} x_ix_jx_kx_l \mcom \label{eq:nonneg:obj}
\end{equation}
over all $x\in R^n$ with $\norm{x}=1$, where in this section we let $\norm{x}$ be the standard (counting) Euclidean norm $\norm{x} = \sqrt{\sum_i x_i^2}$.

One can get some intuition for this problem by considering the case where $M$ is $0/1$ valued and $x$ is $0/k^{-1/2}$ valued for some $k$.
In this case one can think of $M$ is a $4$-uniform hypergraph on $n$ vertices and $x$ as a subset $S\subseteq [n]$ that maximizes the number of edges inside $S$ divided by $|S|^2$,
and so this problem is related to some type of a densest subgraph problem on a hypergraph.\footnote{
The condition of maximizing $|E(S)|/|S|^2$ is related to the \emph{log density} condition used by~\cite{BhaskaraCCFV10} in their work on the densest subgraph problem,
since, assuming that the set $[n]$ of all vertices is not the best solution, the set $S$ satisfies that $\log_{|S|}|E(S)| > \log_n |E|$.
However, we do not know how to use their algorithm to solve this problem.
Beyond the fact that we consider the hypergraph setting, their algorithm manages to find a set of nontrivial density under the assumption that there is a ``log dense'' subset,
but it is not guaranteed to find the ``log dense'' subset itself.
}

Let's assume that we are given a distribution $\cX$ over unit vectors that achieve some value $\sval$ in (\ref{eq:nonneg:obj}). This is a non convex problem, and so generally the average of these vectors would not
be a good solution. However, it turns out that the vector $x^*$ defined such that $x^*_i = \sqrt{\E_{x\sim \cX} x_i^2}$ can sometimes be a good solution for this problem. Specifically, we will show that if it fails to give a solution of value at least $c-\e$, then we can find a new distribution $\cX'$ obtained by reweighing elements $\cX$ that is in some sense ``simpler'' than $\cX$.
More precisely, we will define some nonnegative potential function $\Psi$ such that $\Psi(\cX) \leq \log n$ for all $\cX$ and $\Psi(\cX') \leq \Psi(\cX) - \Omega(\e^2)$ under the above conditions.
This will show that we will need to use this reweighing step at most logarithmically many times.

Indeed, suppose that
\begin{equation}
  \sum_{i,j,k,l} M_{i,j,k,l} x^*_ix^*_jx^*_kx^*_l  = ({x^*}^{\otimes 2})^T M{x^*}^{\otimes 2} \leq \sval - \e \mper \label{eq:nonneg:sketch:1}
\end{equation}

We claim that in contrast
\begin{equation}
y^T My \geq \sval \mcom \label{eq:nonneg:sketch:2}
\end{equation}
where $y$ is the $n^2$-dimensional vector defined by $y_{i,j} = \sqrt{\E_{x\sim \cX} x_i^2x_j^2}$.
Indeed, (\ref{eq:nonneg:sketch:2}) follows from the non-negativity of $M$ and the Cauchy--Schwarz inequality since
\[
\sval = \sum_{i,j,k,l} M_{i,j,k,l} \E_{x\in\cX} x_ix_jx_kx_l \leq \sum_{i,j,k,l} M_{i,j,k,l} \sqrt{\E_{x\sim \cX} x_i^2x_j^2}\sqrt{\E_{x\sim \cX} x_k^2x_l^2} = y^T My
\]

Note that since $\cX$ is a distribution over unit vectors, both $x^*$ and $y$ are unit vectors, and hence (\ref{eq:nonneg:sketch:1}) and (\ref{eq:nonneg:sketch:2})  together with the fact that $M$ has bounded spectral norm imply that
\begin{multline}
\e \leq y^T My - ({x^*}^{\otimes 2})^T M{x^*}^{\otimes 2} = (y-{x^*}^{\otimes 2})^T M(y+{x^*}^{\otimes 2}) \\
\leq \norm{y-{x^*}^{\otimes 2}}\cdot \norm{y+{x^*}^{\otimes 2}} \leq 2\norm{y-{x^*}^{\otimes 2}} \mper \label{eq:nonneg:sketch:3}
\end{multline}

However, it turns out that $\norm{y-{x^*}^{\otimes 2}}$ equals $\sqrt{2}$ times the \emph{Hellinger distance} of the two distributions $D,D^*$ over $[n]\times [n]$ defined as follows:
$\Pr[D=(i,j)] = \E x_i^2x_j^2$ while $\Pr[ D^* = (i,j) ] = (\E x_i^2)(\E x_j^2)$ (see Section~\ref{sec:non-neg}).
At this point we can use standard information theoretic inequalities to derive from (\ref{eq:nonneg:sketch:3}) that there is $\Omega(\e^2)$  \emph{mutual information} between the two parts of $D$.
Another way to say this is that the entropy of the second part of $\cD$ drops on average by $\Omega(\e^2)$ if we condition on the value of the first part.
To say the same thing mathematically, if we define $D(\cX)$ to be the distribution $(\E_{x\sim \cX} x_1^2 , \ldots, \E_{x\sim \cX} x_n^2)$ over $[n]$ and $D(\cX | i)$ to be the distribution $\tfrac{1}{\E_{x\sim\cX} x_i^2}(\E_{x\sim \cX} x_i^2x_1^2 , \ldots, \E_{x\sim \cX} x_i^2x_n^2)$ then
\[
\E_{i\sim D(x)} H(\cX | i) \leq H(\cX) - \Omega(\e^2) \mper
\]
But one can verify that $D(\cX | i) = D(\cX_i)$ where $\cX_i$ is the distribution over $x$'s such that $\Pr[\cX_i = x] = x_i^2\Pr[\cX = x ] / \E_{\cX} x_i^2$, which means that if we define $\Psi(\cX) = H(D(\cX))$ then we get that
\[
\E_{i\sim D(x)} \Psi(\cX_i) \leq \Psi(\cX) - \Omega(\e^2)
\]
and hence $\Psi$ is exactly the potential function we were looking for.

To summarize our combining algorithm will do the following for $t = O(\log n/\e^2)$ steps:
given the first moments of the distribution $\cX$, define the vector $x^*$ as above and test if it yields an objective value of at least $\sval-\e$.
Otherwise, pick $i$ with probability $\E_{x\sim \cX} x_i^2$ and move to the distribution $\cX_i$.
Note that given $d$ level moments for $\cX$, we can compute the $d-1$ level moments of $\cX_i$, and hence the whole algorithm can be carried out with only access to level $O(\log n / \e^2)$ moments  of $\cX$.
We then see that the only properties of the moments used in this proof are linearity, the fact that $\sum x_i^2$ can always be replaced with $1$ in  any expression, and the Cauchy--Schwarz  inequality used for obtaining (\ref{eq:nonneg:sketch:2}).
It turns out that all these properties hold even if we are not given access to the moments of a true distribution $\cX$ but are only given access to a level $d$ \emph{pseudoexpectation} operator $\pE$
for $d$ equalling some constant times  $\log n / \e^2$.
Such pseudoexpectations operators can be optimized over in $d$ levels of the SOS hierarchy, and hence this combining algorithm is in fact a rounding algorithm.

\section{Approximation for nonnegative tensor maximization}
\label{sec:non-neg}

In this section we prove Theorem~\ref{thm:nonneg:intro}, giving an approximation algorithm for the maximum over the sphere of a polynomial with nonnegative coefficients.
We will work in the space $R^n$ endowed with the \emph{counting} measure for norms and inner products.
We will define the \emph{spectral norm} of a degree-$2t$ homogeneous polynomial $M$ in $x=x(x_1,\ldots,x_n)$, denoted by $\spectralnorm{M}$, to be the minimum of the spectral norm of  $Q$ taken over all
quadratic forms $Q$ over $(\R^n)^{\ot t}$ such that $Q(x^{\ot t})= M(x)$ for every $x$. Note that we can compute the spectral norm of an homogeneous polynomial in polynomial time using semidefinite programming.
Thus we can restate our main theorem of this section as:

\begin{theorem}[Theorem~\ref{thm:nonneg:intro}, restated] \label{thm:nonneg}
Let $M$ be a degree-$2t$ homogeneous polynomial in $x=(x_1,\ldots,x_n)$ with nonnegative coefficients.
Then, there is an algorithm, based on $O(t^2 \log n / \e^2)$ levels of the SOS hierarchy, that finds
a unit vector $x^* \in \R^n$ such that
\[
M(x^*) \geq \max_{x\in \R^n, \norm{x}=1} M(x) - \e\spectralnorm{M} \mper
\]
\end{theorem}

To prove Theorem~\ref{thm:nonneg} we first come up with a \emph{combining algorithm},  namely an algorithm that  takes (the moment matrix of) a distribution $\cX$ over unit vectors
$x\in R^n$  such that $M(x) \geq \sval$ and find a unit vector $x^*$ such that $M(x^*) \geq \sval -\e$.
We then show that the algorithm will succeed even if $\cX$ is merely a level $O(t \log n /\e^2)$ \emph{pseudo distribution}; that  is, the moment matrix is a pseudoexpectation operator.
The combining algorithm is very simple:

\medskip

\noindent \emph{\large Combining algorithm for polynomials with nonnegative coefficients:}

\medskip

\noindent \textbf{Input:} distribution $\cX$ over unit $x\in\R^n$ such that $M(x)=\sval$.

\medskip
\noindent \textbf{Operation:} Do the following for $t^2 \log n/\e^2$ steps:

\begin{description}

\item[\it Direct rounding:] For $i\in [n]$, let $x^*_i = \sqrt{\E_{x\sim \cX} x_i^2}$. If $M(x^*) \geq \sval-4\e$ then output $x^*$ and quit.

\item[\it Conditioning:] Try to find $i_1,\ldots, i_{t-1} \in [n]$ such that the distribution $\cX_{i_1,\ldots,i_{t-1}}$ satisfies $\Psi(\cX_{i_1,\ldots,i_{t-1}}) \leq \Psi(\cX) - \e^2/t^2$, and set $\cX= \cX_{i_1,\ldots,i_{t-1}}$, where:
\begin{itemize}
\item  $\cX_{i_1,\ldots,i_{t-1}}$ is defined by letting  $\Pr[ \cX_{i_1,\ldots,i_{t-1}}=x ]$ be proportional to $\Pr[\cX=x]\cdot \prod_{j=1}^{t-1}x_{i_j}^2$ for every $x\in\R^n$.

\item $\Psi(\cX)$ is defined to be $H(A(\cX))$ where $H(\cdot)$ is the Shannon entropy function and $A(\cX)$ is the distribution over $[n]$ obtained by letting $\Pr[A(\cX)=i] = \E_{x\sim\cX} x_i^2$ for every
$i\in [n]$.
\end{itemize}
\end{description}

Clearly $\Psi(\cX)$ is always in $[0,\log n]$, and hence if we can show that we always succeed in at least one of the steps, then eventually the algorithm will output a good $x^*$.
We now show that if the direct rounding step fails, then the conditioning step must succeed. We do the proof under the assumption that $\cX$ is an actual distribution. Almost of all of this analysis holds verbatim when $\cX$ is a pseudodistribution of level at least $2t^2 \log n / \e^2$, and we note the one step where the extension requires using a nontrivial (though easy to prove) property of pseudoexpectations, namely that they satisfy the Cauchy--Schwarz inequality.

\paragraph{Some information theory facts}
We recall some standard relations between various entropy and distance measures.
Let $X$ and $Y$ be two jointly distributed random variables.
We denote the joint distribution of $X$ and $Y$ by $\set{XY}$,
and their marginal distributions by $\set{X}$ and $\set{Y}$.
We let $\set{X}\set{Y}$ denote the product of the distributions $\set{X}$ and $\set{Y}$ (corresponding to sampling $X$ and $Y$ independently from their marginal distribution).
Recall that the \emph{Shannon entropy} of $X$, denoted by $H(X)$, is defined to be $\sum_{x\in \mathrm{Support}(X)}\Pr[X=x]\log(-\Pr[X=x])$. The \emph{mutual information} of $X$ and $Y$ is defined as
$I(X,Y) \defeq H(X) - H(X\mid Y)$, where $H(X\mid Y)$ is \emph{conditional entropy} of $X$ with respect to $Y$, defined as $\E_{y\sim \set{Y}}H(X\mid Y=y)$.
The \emph{Hellinger distance} between two distributions $p$ and $q$ is defined by
$\dhell(p,q) \defeq \Paren{1-\sum_i \sqrt{p_iq_i}}^{1/2}$.
(In particular, $\dhell(p,q)$ equals $1/\sqrt{2}$ times the Euclidean distance of the unit vectors $\sqrt p$ and $\sqrt q$.) The following inequality (whose proof follows by combining standard relations
between the Hellinger distance,  Kullback–--Leibler divergence, and mutual information) would be useful for us

\begin{lemma} 
\label{lem:hellinger}
  For any two jointly-distributed random variables $X$ and $Y$,
  \begin{displaymath}
    2\dhell\Bigparen{\set{X Y},\set{X}\set{Y}}^2\le I(X,Y)
  \end{displaymath}
\end{lemma}


\subsection{Direct Rounding}
\label{sec:simple-rounding}

Given $\cX$, we define the following correlated random variables $A_1,\ldots,A_t$ over $[n]$: the probability that $(A_1,\ldots,A_t) = (i_1,\ldots,i_t)$ is equal to $\E_{x\sim \cX} x_{i_1}^2\cdots x_{i_t}^2$. Note that for every $i$, the random variable $A_i$ is distributed according to $A(\cX)$.  (Note that even if $\cX$ is only a pseudodistribution, $A_1,\ldots,A_t$ are actual random variables.)
The following lemma gives a sufficient condition for our direct rounding step to succeed:

\begin{lemma}
  \label{lem:direct-rounding}
  Let $M,\cX$ be as above.
  If $d_H(\set{A_1\cdots A_t},\set{A_1}\cdots\set{A_t})\le \e$, then the unit vector $x^*$ with $x^*_i=(\E_{x\sim\cX} x_i^2)^{1/2}$ satisfies $M(x^*)\ge \sval-4\e\spectralnorm{M}$.
  Moreover, this holds even if $\cX$ is a level $\ell \geq 2t$ pseudodistribution.
\end{lemma}

\begin{proof}
  Let $Q$ be a quadratic form with $Q(x^{\ot t})=M(x)$.
  Let $y\in (\R^{n})^{\ot t}$ be the vector $y_{i_1\cdots i_t}=\paren{\pE_{x\sim cX} x_{i_1}^2\cdots x_{i_t}^2}^{1/2}$.
  Then,
  \begin{equation}
    \label{eq:1}
    \pE M(x^*) = \iprod{\hat M,\pE {x^*}^{\ot 2t}}\le \iprod{\hat M,y \ot y} = Q(y)
  \end{equation}

  Here, the vector $\hat M\in (\R^{n})^{\ot 2t}$ contains the coefficients of $M$. In particular, $\hat M\ge 0$ entry-wise.
  The inequality in \pref{eq:1} uses Cauchy--Schwarz; namely that $\pE x^\alpha x^\beta \le \paren{\pE (x^\alpha)^2\cdot \pE (x^\beta)^2}^{1/2}=y_\alpha y_\beta$.
  The final equality in \pref{eq:1} uses that $y$ is symmetric.

  Next, we bound the difference between $Q(y)$ and $M(x^*)$
  \begin{equation}
    \label{eq:2}
    Q(y)-M(x^*) = Q(y)-Q({x^*}^{\ot t})
    =\iprod{y+{x^*}^{\ot t},Q(y-{x^*}^{\ot t})}
    \le \norm{Q} \cdot \norm{y+{x^*}^{\ot t}} \cdot \norm{y-{x^*}^{\ot t}}\mper
  \end{equation}
  (Here, $\iprod{\cdot,Q~\cdot}$ denotes the symmetric bilinear form corresponding to $Q$.)

  Since both ${x^*}^{\ot t}$ and $y$ are unit vectors, $\norm{y+{x^*}^{\ot t}}\le 2$.
  By construction, the vector~$y$ corresponds to the distribution $\set{A_1\cdots A_t}$ and ${x^*}^{\ot t}$ corresponds to the distribution $\set{A_1}\cdots\set{A_t}$.
  In particular, $\dhell(\set{A_1\cdots A_t},\set{A_1}\cdots\set{A_t})=\tfrac{1}{\sqrt{2}}\norm{y-{x^*}^{\ot t}}$.
  Together with the bounds \pref{eq:1} and \pref{eq:2},
  \begin{displaymath}
    M(x^*)\ge \pE M(x) - 4\norm{Q}\cdot \dhell(\set{A_1\cdots A_t},\set{A_1}\cdots\set{A_t})\mper\qedhere
  \end{displaymath}

To verify this carries over when $\cX$ is a pseudodistribution, we just need to use the fact that Cauchy--Schwarz holds for pseudoexpectations (Lemma~\ref{lem:pef-cauchy-schwarz}).
\end{proof}

\subsection{Making Progress}
\label{sec:making-progress}

The following lemma shows that if the sufficient condition above is violated, then on expectation we can always make progress. (Because $A_1,\ldots,A_t$ are actual
random variables, it automatically holds regardless of whether $\cX$ is an actual distribution or a pseudodistribution.)

\begin{lemma}
\label{lem:making-progress}
If $\dhell(\set{A_1\cdots A_t},\set{A_1}\cdots\set{A_t})\ge \e$, then $H(A_t\mid A_1\cdots A_{t-1})\le H(A)-2\e^2/t^2$
\end{lemma}

\begin{proof}
  The bound follows by combining a hybrid argument with Lemma~\ref{lem:hellinger}.

  Let $A'_1,\ldots,A'_t$ be independent copies of $A_1,\ldots,A_t$ so that
  \begin{displaymath}
    \set{A_1\cdots A_t\cdots A'_1\cdots A'_t} = \set{A_1\cdots A_t} \set{A_1}\cdots \set{A_t}\mper
  \end{displaymath}

  We consider the sequence of distributions $D_0,\ldots,D_t$ with
  \begin{displaymath}
    D_i = \set{A_1\cdot A_i \cdots A'_{i+1}\cdots A'_t}\mper
  \end{displaymath}

  By assumption, $\dhell(D_0,D_t)\ge \e$.
  Therefore, there exists an index $i$ such that $\dhell(D_{i-1},D_{i})\ge \e/t$.
  Let $X=A_1\cdots A_{i-1}$ and $Y=A_{i}A'_{i+1}\cdots A'_t$.
  Then, $D_i=\set{X Y}$ and $D_{i-1}=\set{X}\set{Y}$.
  By \pref{lem:hellinger},
  \begin{displaymath}
    H(Y)-H(Y\mid X)=I(X,Y)\ge 2\dhell(\set{X Y},\set{X}\set{Y})\ge 2 \e^2/t^2\mper
  \end{displaymath}
  Since $A'_{i+1},\ldots,A'_t$ are independent of $A_1,\ldots,A_{i}$,
  \begin{displaymath}
    H(Y)-H(Y\mid X) = H(A_i) - H(A_i\mid A_1\cdots A_{i-1})\mper
  \end{displaymath}
  By symmetry and the monotonicity of entropy under conditioning, we conclude
  \begin{displaymath}
    H(A_t\mid A_1\cdots A_{t-1}) \le H(A)-2\e^2/t^2\mper\qedhere
  \end{displaymath}
\end{proof}

Lemma~\ref{lem:making-progress} implies that if our direct rounding fails then the expectation of $H(A_1)$ conditioned on $A_2,\ldots,A_t$ is at most $H(A) - 2\e^2/t^2$,
but in particular  this means there exist $i_1,\ldots,i_{t-1}$ so that $H(A_t|A_1=i_1,\ldots,A_{t-1}=i_{t-1}) \leq H(A) - 2\e^2/t^2$.
The probability of $i$ under this distribution $A_t|A_1=i_1,\ldots,A_{t-1}=i_{t-1}$  is proportional to $\E_{x\sim \cX} x_i^2 \cdot \prod_{j=1}^{t-1} x_{i_j}^2$,
which means that it exactly equals the distribution $A(\cX_{i_1,\ldots,i_{t-1}})$. Thus we see that $\Psi(\cX_{i_1,\ldots,i_{t-1}}) \leq \Psi(\cX) - 2\e^2/t^2$.
This concludes the proof of Theorem~\ref{thm:nonneg}. \qed

\begin{remark}[Handling odd degrees and non homogenous polynomials] \label{rem:odd-degree}
If the polynomial is not homogenous but only has monomials of even degree, we can homogenize it by
multiplying every monomial with an appropriate power of $(\sum x_i^2)$ which is identically equal to  $1$ on the sphere.
To handle odd degree monomials we can introduce a new variable $x_0$ and set a constraint that it must be identically equal to $1/2$.
This way we can represent all odd degree monomials by even degree monomials with a blowup of $2$ in the coefficients.
Note that if the pseudoexpectation operator is consistent with this constraint then our rounding algorithm will in fact output a vector that satisfies it.
\end{remark}
\section{Finding an ``analytically sparse'' vector in a subspace}
\label{sec:asvp}

In this section we prove  Theorem~\ref{thm:ASVP:intro}.
We let $\univ$ be a universe of size $n$ and  $L_2(\univ)$ be the vector space of real-valued functions~$f\from \univ\to \bbR$.
The measure on the set $\univ$ is the uniform probability distribution and hence we will use
the inner product $\iprod{f,g}=\E_{\elem} f(\elem)g(\elem)$ and norm
$\norm{f}_p =\left( \E_{\elem} f(\elem)^p\right)^{1/p}$ for $f,g\from \univ\to \R$ and $p \geq 1$.

\begin{theorem}[Theorem~\ref{thm:ASVP:intro}, restated] \label{thm:ASVP} There is a constant $\e>0$ and a polynomial-time algorithm $A$, based on $O(1)$ levels of the SOS hierarchy,
that on input a projector operator $\Pi$ such that there exists a $\mu$-sparse Boolean function $f$ satisfying $\norm{\Pi f}_2^2 \geq (1-\e)\norm{f}_2^2$,  outputs a function $g \in \textrm{Image}(\Pi)$ such that
\[
\norm{g}_4^4 \geq \Omega\left(\tfrac{\norm{g}_2^4}{\mu(\rank \Pi)^{1/3}}\right) \mper
\]
\end{theorem}

We will prove Theorem~\ref{thm:ASVP} by first showing a combining algorithm and then transforming it into a rounding algorithm. Note that the description of the combining algorithm is
independent of the actual relaxation used, since it assumes a true distribution on the solutions, and so we first describe the algorithm before specifying the relaxation.
In our actual relaxation we will use some auxiliary variables that will make the analysis of the algorithm simpler.

\medskip

\noindent \emph{\large Combining algorithm for finding an analytically sparse vector:}

\medskip
\noindent\textbf{Input:} Distribution $\cD$ over Boolean (i.e., $0/1$ valued) functions $f \in L_2(\univ)$ that satisfy:
\begin{itemize}

\item $\mu(f) = \Pr[ f(\elem)=1] = 1/\lambda$.

\item $\norm{\Pi f}_2^2 \geq (1-\e)\norm{f}_2^2$.

\end{itemize}

\noindent\textbf{Goal:} Output $g$ such that
\begin{equation}
\norm{g}_4^4 \geq \gamma \norm{g}_2^2 \text{ where } \gamma = \Omega(1/\mu(\rank \Pi)^{1/3}) \label{eq:goal}
\end{equation}

\noindent\textbf{Operation:} Do the following:

\begin{description}

\item[\it Coordinate projection rounding:]
For $\elem\in \univ$, let $\delta_\elem\from \univ\to \R$ be the function that satisfies
$\iprod{f,\delta_\elem}=f(\elem)$ for all $f\in L_2(\univ)$.
Go over all vectors of the form $g_\elem = \Pi \delta_\elem$ for $\elem\in \univ$ and if there is one that satisfies (\ref{eq:goal}) then output it.
Note that the output of this procedure is independent of the distribution $\cD$.

\item[\it Random function rounding:]  Choose a random gaussian vector $t \in L_2(\univ)$ and output $g=\Pi t$ if it satisfies (\ref{eq:goal}).
(Note that this is also independent of the distribution $\cD$.)

\item[\it Conditioning:] Go over all choices for $\elem_1,\ldots,\elem_4 \in \univ$ and modify the distribution $\cD$
to the distribution $\cD_{\elem_1,\ldots,\elem_4}$ defined such that $\Pr_{\cD_{\elem_1,\ldots,\elem_4}}[f]$ is proportional to $\Pr_{\cD}[f]\prod_{j=1}^4 f(\elem_j)^2$ for every $f$.

\item[\it Gaussian rounding:] For every one of these choices, let $t$ to be a random Gaussian that matches the first two moments of the distribution $\cD$, and output $g=\Pi t$ if it satisfies (\ref{eq:goal}).

\end{description}


Because we will make use of this fact later, we will note when certain properties hold not just for expectations of actual probability distributions but for \emph{pseudoexpectations} as well.
The extension to pseudoexpectations is typically not deep, but can be cumbersome, and so the reader might want to initially restrict attention to the case of a combining algorithm, where we only deal with actual expectations.
We show the consequences for each of the steps failing, and then combine them together to get a contradiction.

\subsection{Random function rounding}

We start by analyzing the random function rounding step.
Let $e_1,\ldots,e_n$ be an orthonormal basis for the space of functions $L_2(\univ)$.
Let $t$ be a standard Gaussian function in $L_2(\univ)$, i.e., $t=\xi_1 e_1 + \ldots + \xi_n e_n$ for independent standard normal variable $\xi_1,\ldots,\xi_n$ (each with mean $0$ and variance $1$).
The following lemmas combined show what are the consequences if $\norm{\Pi t}_4$ is not much bigger than $\norm{\Pi t}_2$.

\begin{lemma}
  For any $f,g\from \univ\to \R$,
  \begin{displaymath}
    \E_t \iprod{f,t}\iprod{g,t}
    = \iprod{f,g}
    \mper
  \end{displaymath}
\end{lemma}

\begin{proof}
  In the $\{e_1,\ldots,e_n\}$ basis, $f=\sum_i a_i e_i$ and $g=\sum_j b_j e_j$.
  Then, $\iprod{f,g}=\sum_i a_ib_i$ and $\iprod{f,t}\iprod{g,t} = \sum_{ij} a_ib_j \xi_i\xi_j$, which has expectation $\sum_i a_ib_i$.
  Hence, the left-hand side is the same as the right-hand side.
\end{proof}

\begin{lemma} \label{lem:randfunc-4norm}
  The $4^{th}$ moment of $\norm{\Pi t}_4$ satisfies
  \begin{displaymath}
    \E_t \norm{\Pi t}_4^4
    \geq \E_\elem \norm{\Pi \delta_\elem}_2^4
    \mper
  \end{displaymath}
\end{lemma}

\begin{proof}
  By the previous lemma, the Gaussian variable $\Pi t(\elem)=\iprod{\Pi \delta_\elem, t}$ has variance $\norm{\Pi \delta_\elem}_2^2$.
  Therefore,
  \begin{align*}
    \E_t \norm{\Pi t}_4^4
    & = \E_t \E_\elem \Pi t(\elem)_4^4
    = \E_\elem \E_t \iprod{\delta_\elem,\Pi t}^4\\
    & = 3 \E_\elem \Paren{\E_t \iprod{\Pi \delta_\elem,t}^2}^4
    = 3 \E_\elem \norm{\Pi \delta_\elem}_2^4
    \mcom
  \end{align*}
  since $3 = \E_{X\sim N(0,1)} X^4$.
\end{proof}

\begin{lemma} \label{lem:randfunc-2norm}
  The $4^{th}$  moment of $\norm{\Pi t}_2$ satisfies
  \begin{displaymath}
    \E_t \norm{\Pi t}_2^4
    \le 10 \cdot  (\rank \Pi)^2
    \mper
  \end{displaymath}
\end{lemma}

\begin{proof}
  The random variable $\norm{\Pi t}_2$ has a $\chi^2$-distribution with $k=\rank \Pi$ degrees of freedom.
  The mean of this distribution is $k$ and the variance is $2k$.
  It follows that $\E_t \norm{\Pi t}_2^4 \le 10(\rank \Pi)^4$.
\end{proof}

\subsection{Coordinate projection rounding}
\label{sec:coordinate}

We now turn to showing the implications of the failure of projection rounding.
We start by noting the following technical lemma, that holds for both the expectation and counting inner products:

\begin{lemma}
  Let $x$ and $y$ be two independent, vector-valued random variables.
  Then,
  \begin{displaymath}
    \E\iprod{x,y}^4 \le \Paren{\E\iprod{x,x'}^4}^{1/2} \cdot \Paren{\E\iprod{y,y'}^4}^{1/2}\mper
  \end{displaymath}
  Moreover, this holds even if $x,y$ come from a level $\ell \geq 10$ pseudodistribution. 
\end{lemma}

\begin{proof} By Cauchy--Schwarz,
  \begin{multline*}
    \pE_{x,y} \iprod{x,y}^4
     = \iprod{ \pE\nolimits_x x^{\ot 4}, \pE\nolimits_y y^{\ot 4}}
     \\
     \le \norm{\pE\nolimits_xx^{\ot 4}}_2
     \cdot \norm{ \pE\nolimits_y y^{\ot 4}}_2
     = \Paren{\pE\nolimits_{x,x'}\iprod{x,x'}^4}^{1/2}
     \cdot \Paren{\pE\nolimits_{y,y'}\iprod{y,y'}^4}^{1/2}
     \mper
  \end{multline*}

We now consider the case of pseudodistributions.
In this case the pseudoexpectation over two independent $x$ and $x'$ is obtained using Lemma~\ref{lem:pef:independent}.
Let $X$ and $Y$ be the $n^4$-dimensional vectors $\pE x^{\ot 4}$ and  $\pE y^{\ot 4}$ respectively.

We can use the standard Cauchy--Schwarz to argue that $X\cdot Y \leq \norm{X}_2\cdot\norm{Y}_2$, and so what is left is to argue that
  $\norm{X}_2^2 = \pE_{x,x'}\iprod{x,x'}^4$, and similarly for $Y$.
This holds by linearity for the same reason this is true for actual expectations, but for the sake of completeness, we do this calculation.
We use the counting inner product for convenience.
Because the lemma's statement is scale free, this will imply it also for the expectation norm.
\[
\pE_{x,x'} \iprod{x,x'}^4 = \pE_{x,x'} \sum_{i,j,k,l} \pE x_ix_jx_kx_lx'_ix'_jx'_kx'_l = \sum_{i,j,k,l}(\pE_x x_ix_jx_kx_l)(\pE_{x'} x'_ix'_jx'_kx'_l) \mcom
\]
where the last equality holds by independence. But this is simply equal to
\[
\sum_{i,j,k,l}(\pE_x x_ix_jx_kx_l)^2 = \norm{X}_2^2
\]
\end{proof}

The following lemma shows a nontrivial consequence  for $\norm{\Pi \delta_{\elem}}_4^4$ being small:

\begin{lemma}[Coordinate projection rounding] \label{lem:coor-proj}
  For any distribution $\cD$ over $L_2(\univ)$,
  \begin{displaymath}
    \E_{f\sim \cD}\norm{\Pi f}_4^4 \leq
    \Paren{
      \E\nolimits_{f,f'\sim \cD} \iprod{f,\Pi f'}^4
    }^{1/2}
    \cdot\Paren{
      \E\nolimits_{\elem}\norm{\Pi \delta_{\elem}}^4_4
    }^{1/2}
    \mper
  \end{displaymath}
  Moreover, this holds even if $\cD$ is a level $\ell \geq 104$ pseudodistribution.
  (Note that $\elem$ is simply the uniform distribution over $\univ$, and hence the last term of the right hand side always denotes an actual expectation.)
\end{lemma}

\begin{proof} By the previous lemma,
  \begin{align*}
    \pE_{f\sim \cD}\norm{\Pi f}_4^4
    & = \E_{f\sim \cD} \E_\elem \iprod{\delta_\elem,\Pi f}^4
    \le\Paren{
      \pE\nolimits_{f,f'\sim \cD}\iprod{f,\Pi f'}^4
    }^{1/2}
    \cdot \Paren{
      \E\nolimits_{\elem,\elem'}\iprod{\delta_\elem,\Pi \delta_{\elem'}}^4
    }^{1/2}
    \\
    & =\Paren{
      \pE\nolimits_{f,f'\sim \cD} \iprod{f,\Pi f'}^4
    }^{1/2}
    \cdot\Paren{
      \E\nolimits_{\elem}\norm{\Pi \delta_{\elem}}^4_4
    }^{1/2}
    \mper\qedhere
  \end{align*}
\end{proof}

\subsection{Gaussian Rounding}
\label{sec:gaussian}

In this subsection we analyze the gaussian rounding step.
Let $t$ be a random function with the Gaussian distribution that matches the first two moments of a distribution $\cD$ over $L_2(\univ)$.

\begin{lemma}
  The $4^{th}$ moment of $\norm{\Pi t}_4$ satisfies
  \begin{displaymath}
    \E_t \norm{\Pi t}^4_4
    = 3 \E_{f,f'\sim \cD} \Iprod{{(\Pi f)^2},{(\Pi f')^2}}\mper
  \end{displaymath}
  Moreover, this holds even if $\cD$ is a level $\ell \geq 100$ pseudodistribution.
  (Note that even in this case  $t$ is still an actual distribution.)
\end{lemma}

\begin{proof}
    \begin{multline*}
    \E_t \norm{\Pi t}^4_4
    = \E_t \E_\elem \Pi t(\elem)^4
    = 3 \E_\elem \Paren{\E\nolimits_t \Pi t(\elem)^2}^{2}\\
    = 3 \pE_\elem \Paren{\E\nolimits_{f\sim \cD} \Pi f(\elem)^2}^{2}
    = 3 \pE_{f,f'\sim \cD} \Iprod{(\Pi f)^2,(\Pi f')^2}
    \mper\qedhere
  \end{multline*}

\end{proof}




\begin{fact}
  If $\{A,B,C,D\}$ have Gaussian distribution, then
  \begin{displaymath}
    \E ABCD = \E AB\cdot \E CD~ + ~\E AC\cdot \E BD ~+~ \E BC\cdot \E AD
    \mper
  \end{displaymath}
\end{fact}

\begin{lemma}
  The fourth moment of $\norm{\Pi t}_2$ satisfies
  \begin{displaymath}
    \E_t\norm{\Pi t}_2^4 \le 3 \Paren{\E_{f\sim \cD}\norm{\Pi f}_2^2}^2
    \mper
  \end{displaymath}
  Moreover, this holds even if $\cD$ is a level $\ell \geq 4$ pseudodistribution.
\end{lemma}

\begin{proof} By the previous fact,
  \begin{multline*}
    \E_t\norm{\Pi t}_2^4
    =\E_{\elem,\elem'} \E_t \Pi t(\elem)^2\cdot \Pi t(\elem')^2\\
    = \E_{\elem,\elem'} \pE_{f} \Pi f(\elem)^2\cdot \pE_f \Pi f(\elem')^2 + 2 \Paren{\pE_f \Pi
      f(\elem)\Pi f(\elem')}^2
    = 3 \Paren{\pE_{f}\norm{\Pi f}_2^2}^2
    \mper\qedhere
  \end{multline*}
\end{proof}




\subsection{Conditioning}
\label{sec:conditioning}

We now show the sense in which conditioning can make progress.
Let $\cD$ be a distribution over $L_2(\univ)$.
For $\elem\in \univ$, let $\cD_\elem$ be the distribution $\cD$ reweighed by $f(\elem)^2$ for $f\sim \cD$.
That is, $\Pr_{\cD_\elem}\{f\} \propto f(\elem)^2\cdot \Pr_{\cD}\{f\}$, or in other words, for every function $P(\cdot)$, $\E_{f\sim \cD_\elem} P(f) = (\E_{f sim \cD} f(\elem)^2 P(f) )/ ( \E_{f\sim \cD} f(\elem)^2 )$.
Similarly, we write $\cD_{\elem_1,\ldots, \elem_r}$ for the distribution $\cD$ reweighed by $f(\elem_1)^2 \cdots f(\elem_r)^2$.

\begin{lemma}[Conditioning] \label{lem:conditioning}
For every even $r\in\N$, there are points $\elem_1,\ldots,\elem_r\in \univ$ such that the reweighed distribution $\cD'=\cD_{\elem_1,\ldots,\elem_r}$ satisfies
  \begin{displaymath}
    \E_{f,g\sim \cD'} \Iprod{f^2,g^2}\ge \Paren{\E\nolimits_{f,g\sim \cD}\Iprod{f^2,g^2}^r}^{1/r}
  \end{displaymath}

  Moreover, this holds even if $\cD$ is a level $\ell \geq r+4$ pseudodistribution.
\end{lemma}

\begin{proof} We have that
  \[
    \max_{\elem_1,\ldots,\elem_r}\pE_{f,g\sim \cD_{\elem_1,\ldots, \elem_r}}\Iprod{f^2,g^2} = \max_{\elem_1,\ldots,\elem_r}\left( \frac{\pE f(\elem_1)^2\cdots f(\elem_r)^2 \cdot g(\elem_1)^2 \cdots g(\elem_r)^2 \Iprod{f^2,g^2}}{\left(\pE_f f(\elem_1)^2 \cdots f(\elem_r)^2 \right)\left( \pE_g g(\elem_1)^2 \cdots g(\elem_r)^2 \right) } \right)
  \]
  but using $\E (X/Y) \leq (\max X)/(\max Y)$ and $\E_{\elem_1,\ldots,\elem_r} f(\elem_1)^2\cdots f(\elem_r)^2 g(\elem_1)^2\cdots g(\elem_r)^2 = \Iprod{g^2,f^2}^r$, the RHS is lower bounded by
  \[
\frac{\E_{\elem_1,\ldots,\elem_r} \pE f(\elem_1)^2\cdots f(\elem_r)^2 \cdot g(\elem_1)^2 \cdots g(\elem_r)^2 \Iprod{f^2,g^2}}{\E_{\elem_1,\ldots,\elem_r}\left(\pE_f f(\elem_1)^2 \cdots f(\elem_r)^2 \right)\left( \pE_g g(\elem_1)^2 \cdots g(\elem_r)^2 \right)} = \frac{\E_{f,g \sim \cD} \Iprod{f^2,g^2}^{r+1}}{\pE_{f,g\sim \cD} \Iprod{f^2,g^2}^r}
  \]

  Now, if $\cD$ was an actual expectation, then we could use \Holder's inequality to lower bound the numerator of the RHS by  $\left( \E_{f,g \sim \cD} \Iprod{f^2,g^2}^r\right)^{(r+1)/r}$ which would lower bound the RHS by $\left( \E_{f,g \sim \cD} \Iprod{f^2,g^2}^r\right)^{1/r}$.
For pseudoexpectations this follows by appealing to Lemma~\ref{lem:pseudo-holder}.
\end{proof}

\subsection{Truncating functions} \label{sec:truncating}

The following observation would be useful for us for analyzing the case that the distribution is over functions that are not completely inside the subspace.
Note that if the function $f$ is inside the subspace, we can just take $\topf =f $ in Lemma~\ref{lem:trunc}, and so the reader may want to skip this section in a first reading and
just pretend that $\topf = f$ below.

\begin{lemma}\label{lem:trunc} Let $\e<1/400$, $\Pi$ be a projector on $\R^\univ$ and
suppose that $f\from \univ \to \{0,1\}$ satisfies that $\Pr[f(\elem)=1]=\mu$ and $\norm{\Pi f}_2^2 \geq (1-\e)\mu$.
Then there exists a function $\topf:\R^\univ\rightarrow\R$ such that:
\begin{enumerate}
\item $\norm{\Pi \topf}_4^4 \geq \Omega(\mu)$.

\item For every $\elem\in\univ$, $\Pi f(\elem)^2 \geq \Omega(|\topf(\elem)|)$.
\end{enumerate}
\end{lemma}
\begin{proof}
Fix $\tau>0$ to be some sufficiently small constant (e.g., $\tau=1/2$ will do).  Let $f'=\Pi f$.
We define $\topf = f'\cdot 1_{|f'|\geq \tau}$  (i.e., $\topf(\elem)=f'(\elem)$ if $|f'(\elem)|\geq \tau$ and $\topf(\elem)=0$ otherwise) and define $\botf = f'\cdot 1_{|f'|<\tau}$.
Clearly $f'(\elem)^2 \geq \tau |\topf(\elem)|$  for every $\elem\in\univ$.

Since $\botf(x) \neq 0$ if and only if $f'(x) \in (0,\tau)$, clearly $|\botf(x)| \leq |f(x)-f'(x)|$ and hence $\norm{\botf}_2^2 \leq \e \mu$.
Using $f' = \topf + \botf$, we see that $\Pi \topf = f + (f'-f) - \botf + (\Pi \topf - \topf)$.
Now since $f'$ is in the subspace, $\norm{\Pi \topf - \topf}_2 \leq \norm{f'-\topf}_2  = \norm{\botf}_2$ and hence for
$g= (f'-f) - \botf + (\Pi \topf - \topf)$, $\norm{g}_2 \leq 3\sqrt{\e\mu}$.
Therefore the probability that $g(\elem) \geq 10\sqrt{\e}$ is at most $\mu/2$.
This means that with probability at least $\mu/2$ it holds that $f(\elem)=1$ and $g(\elem) \leq 10\sqrt{\e}$, in which case $\topf(\elem) \geq 1 - 10\sqrt{\e} \geq 1/2$.
In particular, we get that $\E \topf(\elem)^4 \geq \Omega(\mu)$.
\end{proof}

\begin{remark}[Non-Boolean functions] \label{rem:non-boolean}
The proof of Lemma~\ref{lem:trunc} establishes much more than its statement.
In particular note that we did not make use of the fact that $f$ is nonnegative, and a function $f$ into $\{0,\pm 1\}$ with $\Pr[ f(\elem)\neq 0]=\mu$
would work just the same.
We also did not need the nonzero values to have magnitude exactly one, since the proof would easily extend to the case where they
are in $[1/c,c]$ for some constant $c$.
One can also allow some nonzero values of the function to be outside that range, as long as their total contribution to the 2-norm squared is much
smaller than $\mu$.
\end{remark}


\subsection{Putting things together}
\label{sec:together}

We now show how the above analysis yields a combining algorithm, and we then discuss the changes needed to extend this argument to pseudodistributions, and hence obtain a rounding algorithm.

Let $\cD$ be a distribution over Boolean functions $f\from \univ\to \{0,1\}$ with $\norm{f}_2^2=\mu$ and $\norm{\Pi f}_2^2\ge 0.99\norm f_2^2$.
The goal is to compute a function $t\from \univ\to \R$ with $\norm{\Pi t}_4 \gg \norm{t}_2^2$, given the low-degree moments of $\cD$.

\medskip

Suppose that random-function rounding and coordinate-projection rounding fail to produce a function $t$ with $\norm{\Pi t}^4_4 \geq \gamma \norm{t}^4_2$.
Then, $\E_\elem \norm{\Pi \delta_\elem}_2^4\le O(\gamma)\cdot (\rank \Pi)^2$ (from failure of random-function rounding and Lemmas~\ref{lem:randfunc-4norm} and~\ref{lem:randfunc-2norm}).
By the failure of coordinate-projection rounding (and using Lemma~\ref{lem:coor-proj} applied to the distribution over $\topf$) we get that
\begin{displaymath}
  \Paren{\E_{f\sim \cD} \norm{\Pi \topf}_4^4}^2
  \le O(\gamma)\cdot \E_{f,f'\sim \cD}\Iprod{\topf, \topf'}^4 \cdot \E_\elem \norm{\Pi \delta_\elem}_2^4.
\end{displaymath}
Combining the two bounds, we get
\begin{displaymath}
  \E_{f,f'\sim \cD}\Iprod{\topf, \topf'}^4
  \ge \Omega(1/(\gamma \rank \Pi)^2)   \Paren{\E_{f\sim \cD} \norm{\Pi \topf}_4^4}^2
\end{displaymath}
Since (by Lemma~\ref{lem:trunc}), $(\Pi f)(\elem)^2 \ge \Omega(|\topf(\elem)|)$ for every $\elem\in\univ$ and $f$ in the support of $\cD$,
we have $\iprod{(\Pi f)^2,(\Pi f')^2}\ge \Omega(\iprod{\topf, \topf'})$ for all $f,f'$ in the support.
Thus,
\begin{displaymath}
  \E_{f,f'\sim \cD}\Iprod{(\Pi f)^2, (\Pi f')^2}^4
  \ge \Omega(1/(\gamma \rank \Pi)^2)   \Paren{\E_{f\sim \cD} \norm{\Pi \topf}_4^4}^2
\end{displaymath}
By the reweighing lemma, there exists $\elem_1,\ldots,\elem_4\in \univ$ such that the reweighted distribution $\cD'=\cD_{\elem_1,\ldots,\elem_4}$ satisfies
\begin{displaymath}
  \E_{f,f'\sim \cD'}\Iprod{(\Pi f)^2, (\Pi f')^2}
  \ge \Paren{\E_{f,f'\sim \cD}\Iprod{(\Pi f)^2, (\Pi f')^2}^4}^{1/4}
  \ge \Omega(1/(\gamma \rank \Pi))^{1/2}\Paren{\E_{f\sim \cD} \norm{\Pi \topf]}_4^4}^{1/2}
\end{displaymath}
The failure of Gaussian rounding (applied to $\cD'$) implies
\begin{displaymath}
    \E_{f,f'\sim \cD'}\Iprod{(\Pi f)^2, (\Pi f')^2}
    \le O(\gamma) \Paren{\E_{f\sim \cD'} \norm{\Pi f}_2^2}^2
    \mper
\end{displaymath}
Combining these two bounds, we get
\begin{displaymath}
       \E_{f\sim \cD} \norm{\Pi \topf}_4^4
  \le O(\gamma^3 \rank \Pi)\cdot \Paren{\E_{f\sim \cD'} \norm{\Pi f}_2^2}^4
\end{displaymath}
By the properties of $\cD$ and Lemma~\ref{lem:trunc}, the  left-hand side is $\Omega(\mu)$ and the
right-hand side is $O(\gamma^3 \rank \Pi  \mu^4)$.
Therefore, we get
\[
\gamma \geq \Omega\left( \tfrac{1}{(\rank \Pi)^{1/3}\mu}\right)
\]
\qedhere

\paragraph{Extending to pseudodistributions}
We now consider the case that $\cD$ is a pseudodistribution of level $\ell \geq 10$.
Most of the statements above just go through as is, given that the analysis of all individual steps does extend
(as noted) for pseudoexpectations.
One issue is that the truncation operation used to obtain $\topf$ is not a low degree polynomial.
While it may be possible to approximate it with such a polynomial, we sidestep the issue by simply adding $\topf$ as additional
auxiliary variables to our program, and enforcing the conclusions of Lemma~\ref{lem:trunc} as constraints that the
pseudoexpectation operator must be consistent with.
This is another example of how we design our relaxation to fit the rounding/combining algorithm, rather than the other way around.
With this step, we can replace statements such as ``(*) holds for all functions in the support of $\cD$'' (where (*) is some equality or
inequality constraint in the variables $f,\topf$) with the statement ``$\cD$ is consistent with (*)'' and thus complete the proof. \qed

\section{Finding planted sparse vectors} \label{sec:planted}

\newcommand{\fz}{\ensuremath{f_0}}
\newcommand{\f}{\ensuremath{f}}
\newcommand{\fp}{\ensuremath{f'}}
\newcommand{\Vp}{\ensuremath{V'}}
\newcommand{\V}{\ensuremath{V}}
\newcommand{\nonzeros}{\ensuremath{S}}
\newcommand{\zeros}{\ensuremath{{\overline{S}}}}
\newcommand{\g}{\ensuremath{g}}
\newcommand{\muVp}{\ensuremath{\mu_{2,1}(\Vp)}}
\newcommand{\mufz}{\ensuremath{\mu_{2,1}(\fz)}}
\newcommand{\PR}{{{\sc PlantedRecovery$(\mu,d,|\univ|,\epsilon)$ } }}

As an application of our work, we show how we can find sparse (or analytically sparse) vectors inside a sufficiently generic subspace.
In particular, this improves upon a recent result of Demanet and Hand~\cite{DemanetH13}
who used the $L_{\infty}/L_1$ optimization procedure of Spielman et al.~\cite{SpielmanWW12} to show one can recover a $\mu$-sparse vector planted
in a random $d$-dimension subspace $\Vp\subseteq \R^n$ when $\mu \ll 1/\sqrt{d}$.
Our result, combined with the bound on the SDP value of the $2\to 4$ norm of a random subspace from \cite{BarakBHKSZ12},
implies that if $d = O(\sqrt{n})$ then we can in fact recover such a vector as long as $\mu \ll 1$.
\begin{minipage}{\textwidth}
  \begin{description}
    	\item[Problem:]\PR
  \item[Input:]
  An arbitrary basis for a linear subspace $V=\mathrm{span}\left(\Vp\cup \{\fz\}\right)$, where:
 \begin{itemize}
 	\item $\Vp\sse \R^{\univ}$ is a \textbf{random $d$-dimensional subspace}, chosen as the span of $d$  vectors drawn independently from the standard Gaussian distribution on $\R^{\univ}$, and
 	\item $\fz$ is an arbitrary \textbf{$\mu$-sparse vector}, i.e.,   $S=\supp(\fz)$ has $|S|\leq\mu |\univ|$.
 \end{itemize}
  \item[Goal:]
Find a vector $f\in V$ with $\iprod{f,\fz}^2\ge (1-\epsilon)\Norm{f}_2 \Norm{\fz}_2$.
 \end{description}
\vspace{0in}
 \end{minipage}

The goal here should be thought of as recovering $\fz$ to arbitrarily high precision (``exactly''), and thus the running time of an algorithm should be logarithmic in $1/\epsilon$.
We note that $\fz$ is not required to be random, and it may be chosen adversarially based on the choice of $\Vp$.  We will prove the following theorem, which is this section's main result:


\begin{theorem} (Theorem~\ref{thm:planted:intro}, restated)
	\label{thm:planted:main}
For some absolute constant $K>0$, there is an algorithm that solves \PR with high probability in time $\poly(|\univ|,\log(1/\epsilon))$ for any $\mu<K\mu_0(d)$,
where
\[
	\mu_0(d)=
	\begin{cases}
	1 & \text{if $d\leq \sqrt{\univ}$, and}\\
	n/d^2 & \text{if $d\geq \sqrt{\univ}$}.
	\end{cases}
\]
\end{theorem}

Our algorithm will work in two stages. It will first solve a constant-degree sum-of-squares relaxation to find a somewhat noisy approximate solution.
It will then solve an auxiliary linear program that converts any sufficiently good approximate solution into an exact one.

The first stage is based on the following theorem (proven in Section~\ref{sub:approx_recovery}), which shows that we can approximately recover a vector
when it is planted in a subspace consisting of vectors with substantially smaller $L_4/L_2$ ratio, provided that we can certify this property of the subspace using a low-degree sum-of-squares proof.
To avoid unnecessary notation, we will use a degree 4 certificate in the statement and proof of the theorem; the proof goes through in greater generality, but this suffices for our application.

\begin{theorem}\label{thm:approx-recovery}
Let $V=\mathrm{span}\left(\Vp\cup \{\fz\}\right)$, where
$\fz\in \R^\univ$ is a vector with $\norm{\fz}_4/\norm{\fz}_2\ge C$, and
 $V'\sse \R^\univ$ is a linear subspace with
\begin{equation}\label{eq:l4l2ratio}
\max_{0\neq f\in \Vp} \tfrac{\norm{f}_4}{\norm{f}_2} \leq c.
\end{equation}
Furthermore, assume that~\eqref{eq:l4l2ratio} has a degree 4 sum-of-squares proof, i.e., that
\begin{equation}
\Norm{\Pi_{\Vp} f}_4^4 = c^4\Norm{\Pi_{\Vp} f}_2^4 - S,	 \label{eq:l4l2ratio:SOS}
\end{equation}
where $\Pi_{\Vp}$ is the orthogonal projection onto $\Vp$, and $S$ is a degree 4 sum of squares.

There is a polynomial-time algorithm
 based on a constant-degree sum-of-squares relaxation
  that returns a vector $f\in V$ with
  $\iprod{f,f_0}^2\ge \left(1-(c/C)^{\Omega(1)}\right)\Norm{\fz}_2 \Norm{\f}_2$.
\end{theorem}

If $\Vp$ is a random subspace of dimension $d$,
 \cite[Theorem~7.1]{BarakBHKSZ12} showed that~\eqref{eq:l4l2ratio} has a degree 4 sum-of-squares proof with high probability for $c=O(1)$ when $d\leq \sqrt{|\univ|}$,
 and for $c=O\left(d^{1/2}/|\univ|^{1/4}\right)$ when $d\geq \sqrt{|\univ|}$.\footnote{Their proof actually directly shows that the polynomial $P$ in the RHS of (\ref{eq:l4l2ratio:SOS})
 has $\spectralnorm{P} \leq c^4$, which corresponds to a degree $4$ SOS proof via Lemma~\ref{lem:spectral norm}.}
We  can concisely write these two cases together in our present notation as $c=O\left(\mu_0(d)^{-1/4}\right)$.

Since $\fz$ is $\mu$-sparse, we know that $\Norm{\fz}_4\geq \mu^{-1/4} \Norm{\fz}_2$, so we can take $C=\mu^{-1/4}$.
We can thus solve a constant-degree sum-of-squares program to obtain a vector $\f$ with  $\iprod{\f,\fz}=(1-O(1))\Norm{\fz}_2 \Norm{\f}_2$
whenever
$c\ll O(\mu^{-1/4})$, i.e., when
\begin{equation}\label{eq:mu_bound}
		\mu\ll O\left(\frac{1}{c^4}\right)=O\left(\mu_0(d)\right).
\end{equation}


For the second stage, we will consider the following
 linear program, which can be thought of as searching for  a sparse vector in $V$ with  a large inner product with $\f$:
 \begin{equation}\label{eq:LP}
 \argmin_{y\in V } \|y\|_1 \text{ such that } \langle y,\f\rangle=1.
 \end{equation}

In Section~\ref{sub:exact-recovery}, we will prove the following theorem, which provides conditions under which the linear program will exactly recover $\fz$ from any $\f$ that is reasonably correlated to it:

\begin{theorem}\label{thm:exact-recovery}
	
Let $V=\mathrm{span}\left(\Vp\cup \{\fz\}\right)$, and suppose that the following conditions hold:
 \begin{itemize}
 		\item $\supp(\fz)=S$, $|S|=\mu n$  \hfill \emph{[$\fz$ is a $\mu$-sparse vector]}\quad \phantom{a}
 	\item  $\otratio{\Vp} \leq \alpha$ where $\otratio{\Vp} = \max \norm{\fp}_2/\norm{\fp}_1$ for all $0\neq \fp\in \Vp$)  \hfill \emph{[$\Vp$ doesn't contain any $1/\alpha^2$ $L_2/L_1$-sparse vectors]}\quad \phantom{a}
 	\item $\iprod{\fz,\f} \ge (1-\epsilon)\Norm{\fz}_2 \Norm{\f}_2$  \hfill \emph{[$\f$ is correlated with $\fz$]}\quad \phantom{a}
 	\item  $\iprod{\fp,\f} \leq \eta \Norm{\fp}_2 \Norm{\f}_2$ for all $\fp\in \Vp$ \hfill \emph{[$\f$ is not very correlated with anything in $\Vp$]}.\quad \phantom{a}
 \end{itemize}
 If
 $$\frac{\eta}{1-\epsilon} <
  \frac{1}{\alpha\sqrt{\mu}} -2,$$
  then $\fz/\iprod{\fz,\f}$  is the unique optimal solution to~\eqref{eq:LP}.
\end{theorem}

\begin{remark} Because we believe the result might be useful elsewhere, we state the theorem in much more generality than needed for our application.
In particular in our application we only need the trivial bound $\eta \leq 1$. Also, a bound on $\otratio{\Vp}$ can also be derived using the relations
between the $4$ norm and the $2$ norm on vectors in $\Vp$.
\end{remark}

To prove Theorem~\ref{thm:planted:main},
we take $\f$ to be the vector with $\iprod{\fp,\f}\geq \left(1-(c/C)^{\Omega(1)}\right)\Norm{\fz}_2 \Norm{\f}_2$ given by Theorem~\ref{thm:approx-recovery} and solve the linear program from Equation~\eqref{eq:LP}.
The theorem is vacuous for $d>\sqrt{K}n$, so we may assume that $d$ is less than any fixed constant times $n$.  In this case,
the following classic result on almost-spherical sections of the $\ell_1$ ball~(\cite{kashin1977diameters,figiel1977dimension}, as stated in~\cite{DemanetH13})
 guarantees that $\otratio{\Vp} \leq O(1)$:
\begin{lemma}\label{lem:Kashin}
Fix $\delta\in (0,1)$, let $d\leq \delta n$, and let $W\sse \R^\univ$ be a random $d$-dimensional subspace given by the span of $d$ independent standard Gaussians. There exists a constant $C_\delta>0$ and absolute constants $\gamma_1,\gamma_2>0$ such that
\[
C_\delta \Norm{\fp}_2^2\leq \Norm{\fp}_1^2\leq\Norm{\fp}_2^2
\]
for all $\fp \in W$ with probability $1-\gamma_1 e^{-\gamma_2 n}$.
\end{lemma}

By Cauchy-Schwarz, we have $\iprod{\fp,\f} \leq  \Norm{\fp}_2 \Norm{\f}_2$, so Theorem~\ref{thm:exact-recovery} implies
 that we recover $\fz$ exactly
as long as\footnote{
We note that the last two bounds were somewhat weak: Lemma~\ref{lem:Kashin} holds for subspaces of linear dimension, but we only applied it to a subspace with $d\leq\sqrt{|\univ|}$; and the application of Cauchy-Schwarz could have been tightened using a better analysis.  However, these were sufficient to prove Theorem~\ref{thm:planted:main}.
}
\begin{equation}\label{eq:constraint_for_recovery}
	\frac{1}{\left(1-O(c/C)\right)} <
	  \sqrt{\frac{C_\delta}{\mu}}-2.
\end{equation}

$C_\delta$ is a constant for any fixed $\delta$, so, by taking $K$ sufficiently small in the statement of the theorem, we may assume that $\mu<C_\delta/16$, and thus that the right-hand side of~\pref{eq:constraint_for_recovery} is at least 2.
In this case, we can recover $\fz$ as long as $c\ll O(C)$.
Combining this with Equation~\pref{eq:mu_bound} and choosing $K$ appropriately thus completes the proof of Theorem~\ref{thm:planted:main}.  \qed

It thus suffices to prove Theorems~\ref{thm:approx-recovery} and~\ref{thm:exact-recovery}, which we will do in sections~\ref{sub:approx_recovery} and~\ref{sub:exact-recovery}, respectively.
We note that Theorems~\ref{thm:approx-recovery} and~\ref{thm:exact-recovery}
hold for any $\Vp$ that meets certain norm requirements, and they do not
 require $\Vp$ to be a uniformly random subspace.
  As such, the results of this section hold in a broader context.  (For example, they immediately generalize to other distributions of subspaces that meet the norm bounds.)
We hope that the technical results of this section will find other
uses, so we have stated them in a somewhat general way to
facilitate their application in other settings.

\subsection{Recovering $\fz$ approximately (Proof of Theorem~\ref{thm:approx-recovery})} 
\label{sub:approx_recovery}
In this section, we prove~\pref{thm:approx-recovery}, which allows us to recover a vector that is reasonably well-correlated with $\fz$.  The basic idea is that $\fz$ has a much larger $L_4/L_2$ ratio than anything in $\Vp$, so maximizing the the $L_4/L_2$ ratio should give a vector near $\fz$.

 The key ingredient of the theorem is the following lemma about (pseudo-)distributions supported on $L_4/L_2$-sparse functions in $\Vp$.
 Note that this lemma does not need the space to be random, but only that it can be certified to have no  $L_4/L_2$ sparse vectors by the SOS SDP.

 \begin{lemma} \label{lem:rec-sparse}
 Let $\Vp \subseteq \R^{\univ}$ be a linear subspace such that
 \begin{equation}
 \max_{0\neq f\in \Vp} \tfrac{\norm{f}_4}{\norm{f}_2} \leq c \label{eq:norm-bound-subspace}
 \end{equation}
 Let $\fz$ be a unit function  in $\Vp^{\perp}$ with $\norm{f_0}_4 = C > 100c$,
 and let $\cX$ be a distribution over $\R^{\univ}$ over unit functions  $f\in \mathrm{Span} \Vp \cup \{ f_0\}$ satisfying
 $\norm{f}_4 \geq C$.  Then
 \[
 \E \iprod{x,f_0}^2 \geq 1 - O(c/C) \mper
 \]
 Moreover this holds even if $\cX$ is a pseudodistribution of level $\ell\geq 8$, as long as (\ref{eq:norm-bound-subspace}) has a degree 4 sum-of-squares proof.
 \end{lemma}


We can obtain a pseudodistribution $\cX$ meeting the requirements in Lemma~\ref{lem:rec-sparse} by solving a degree 8 sum-of-squares program
that maximizes $\Norm{f}_4^4$ over $f\in V$ with $\Norm{f}_2^2=1$.
 If we sample a random Gaussian consistent with the first two moments of $\cX$, then
 we will obtain a vector $g$ whose expected $2$-norm squared is $1$ and whose expected inner product with $f_0$ is $(1-o(1))\Norm{\fz}$, so
  \pref{lem:rec-sparse} therefore implies \pref{thm:approx-recovery}.
 \Bnote{add here note on how to get exact recovery if we have time}
 \Bnote{I didn't claim that we can also plug in our result of the previous section because there is a subtlety I don't have time to think about ---
 we used auxiliary variables, and I'm not sure that our proof immediately implies that any pseudoexpectation of large enough level must be
 consistent with the constraint that the 4 norm of vectors in the subspace must be small.
 I think the eigenvalue bound clearly does this}

\begin{proof}[Proof of~\pref{lem:rec-sparse}]
Write every vector $f$ in the support of $\cX$ in the form $f = \alpha f_0 + f'$ where $f'\in \Vp$ and $\alpha = \iprod{f,f_0}$.
We know that
\begin{equation}
C  = \norm{f}_4 \leq \alpha\norm{f_0}_4 + \norm{f'}_4  \leq \alpha C + c\norm{f'}_2  \leq \alpha C + c \label{eq:iprod-f0-bound},
\end{equation}
so
\[
\alpha \geq 1 - c/C  \mper
\]
This concludes the proof for actual expectations.
To argue about pseudoexpectations, we need to use only constraints involving polynomials, and therefore we use
\[
C^4 = \norm{f}_4^4 = \pE_f \E_\elem ( \alpha f_0(\elem)+ f'(\elem))^4
\]
which equals
\[
\pE_f \alpha^4 \norm{f_0}_4^4  + 4\pE \alpha^3 \iprod{f_0^3,f'} + 6\pE \alpha^2 \iprod{f_0^2,f'^2} +  4\pE \alpha \iprod{f_0,f'^3} + \pE \norm{f'}_4^4 \mper
\]
The existence of a degree 4 sum-of-squares proof of~\eqref{eq:iprod-f0-bound} implies that $\cX$ must be consistent with the constraint $\norm{f}_4^4 \leq c^4$.
We can thus use Cauchy--Schwarz and \Holder's inequality (Lemma~\ref{lem:pef-cauchy-Schwarz-inside} and Corollary~\ref{cor:pef-holder-inside}), to bound all of the terms except the first one by a constant
times $|\alpha|^3 C^3c$, and so we get
\[
C^4 \leq \pE \alpha^4C^4 + 15|\alpha|^3C^3c.
\]
Using the fact that the expectation is consistent with the constraint $|\alpha| \leq 1$, we obtain
\[
\pE \alpha^4 \geq 1 - 15c/C \mper
\]
Since we satisfy $|\alpha| \leq 1$, we know that $\pE \alpha^6 \leq 1$, so we can apply Cauchy-Schwarz to show that
\[
\pE \alpha^4 \leq \sqrt{\pE \alpha^2}{\sqrt{\pE \alpha^6}} \leq \sqrt{\pE \alpha^2} \mcom
\]
which allows us to conclude that
\[
\pE \alpha^2 \geq 1-30c/C. \qedhere
\]
\end{proof}

\subsection{Recovering $\fz$ exactly (Proof of Theorem~\ref{thm:exact-recovery})} 
\label{sub:exact-recovery}
In this section, we prove~\pref{thm:exact-recovery}, which allows us to use a vector near $\fz$ to recover $\fz$ exactly (up to the precision used when solving the linear program).
Intuitively, this is relies on the same
tendency towards sparsity of vectors with minimal $1$-norm
 that underlies the earlier works that are based on $L_\infty/L_1$-sparsity.
Minimizing the $L_\infty/L_1$-sparsity amounts to solving the
 linear program in~\eqref{eq:LP} with $y$ equal to each of the unit basis vectors, and then taking the best of the $|\univ|$ solutions.
When $\fz$ is sparse enough, it will have at least one fairly large coefficient, and $\fz$ will then be sufficiently correlated with the corresponding unit basis vector for the linear program to find it.  This breaks down when $\mu= \Omega(1/\sqrt{d})$, at which point
any one basis vector is expected to be more correlated with some vector in $\Vp$ than it is with $\fz$.  Here, instead of using the unit basis vectors, we use a vector $y$ that shares many coordinates with $\fz$, which then lets us handle a much broader range of $\mu$.

\begin{proof}[Proof of~\pref{thm:exact-recovery}]
	To analyze the optimum of~\eqref{eq:LP}, we  decompose  $y\in V$  as $y=t\fz+ \fp$
	 for $t\in \R$ and $\fp\in V'$.
   	We will show that
	\begin{equation}\label{eq:exact_rec_goal}
\frac{\norm{f_0}_1}{\iprod{f_0,\f}}
\leq	\frac{\norm{y}_1}{\iprod {y,\f}} =
\frac{\Norm{t\fz +\fp}_1}{t\iprod{\fz,\f}+\iprod{\fp,\f}}
	\end{equation}
	for all $y\in V$, with equality only if $\fp=0$, which immediately implies  Theorem~\ref{thm:exact-recovery}.

 Let $\fp_\nonzeros$ and $\fp_\zeros$ be the vectors obtained from $\fp$ by zeroing out the
 	 coordinates outside $\nonzeros$ and $\zeros$, respectively, so that $\fp=\fp_\nonzeros+\fp_\zeros.$
Since $\fz$ is zero outside of $S$, we have
\begin{equation}\label{eq:split_coords}
	\Norm{t\fz +\fp}_1
  = \Norm{t\fz- \fp_S}_1 + \Norm{\fp_{\Sbar}}_1
  \geq t \Norm{\fz}_1- \Norm{\fp_S}_1 + \Norm{\fp_{\Sbar}}_1.
\end{equation}
Equation~\eqref{eq:split_coords} and the inequality
\[
	\frac{A+B}{C+D}\geq \min\left\{\frac{A}{B},\frac{C}{D}\right\}
\]
give
	\begin{equation}\label{eq:frac_to_min}
\frac{\Norm{t\fz +\fp}_1}{t\iprod{\fz,\f}+\iprod{\fp,\f}}
\geq \frac{t \Norm{\fz}_1- \Norm{\fp_S}_1 + \Norm{\fp_{\Sbar}}_1}{t\iprod{\fz,\f}+\iprod{\fp,\f}}
\geq \min\left\{\frac{\Norm{\fz}_1}{\iprod{\fz,\f}},
	\frac{ \Norm{\fp_{\Sbar}}_1- \Norm{\fp_S}_1}{\iprod{\fp,\f}}
	\right\},
	\end{equation}
	 where the second inequality in~\eqref{eq:frac_to_min} is strict unless the two terms inside the $\min$ are equal.
To prove the inequality asserted in~\eqref{eq:exact_rec_goal}, and thus Theorem~\ref{thm:exact-recovery},
it therefore suffices to show that
\begin{equation}\label{eq:exact-rec-ineq}
\frac{\Norm{\fz}_1}{\iprod{\fz,\f}}
<
\frac{ \Norm{\fp_{\Sbar}}_1- \Norm{\fp_S}_1}{\iprod{\fp,\f}}
\end{equation}	
for all $0\neq \fp\in \Vp.$

We can bound the left-hand side of~\eqref{eq:exact-rec-ineq} using the assumptions that $\iprod{\fz,f}\geq 1-\epsilon$ and that
$\fz$ is $\mu$-sparse:
\[
	\frac{\Norm{\fz}_1}{\iprod{\fz,\f}}\leq \frac{\Norm{\fz}_1}{(1-\epsilon)\Norm{\fz}_2\Norm{\f}_2}
	\leq \frac{\sqrt{\mu}\Norm{\fz}_2}{(1-\epsilon)\Norm{\fz}_2\Norm{\f}_2}
	= \frac{\sqrt{\mu}}{(1-\epsilon)\Norm{\f}_2}.
\]

To bound the numerator of the right-hand side of~\eqref{eq:exact-rec-ineq}, we need to show
that $\fp$ cannot have too large a fraction of its 1-norm concentrated in the coordinates in
$S$.
We first note that, if this occurred, it would lead to a large contribution to the  $2$-norm:
\begin{align*}
 \Norm{\fp_\nonzeros}_1
	= \mu\E_{i\in \nonzeros}|\fp(i)|
     \leq \mu \sqrt{ \E_{i\in \nonzeros}\fp(i)^2 }
	 \leq \mu\sqrt{\frac{1}{\mu}\E_i \fp(i)^2}
	 =\sqrt{\mu}\norm{\fp}_2.
\end{align*}
Combining this with our assumption that $\otratio{\Vp} \leq \alpha$, 
gives
 \begin{align*}
		\Norm{\fp_\zeros}_1-\Norm{\fp_\nonzeros}_1&=\Norm{\fp}_1-2 \Norm{\fp_\nonzeros}_1
		\geq  \alpha^{-1} \Norm{\fp}_2-2 \sqrt{\mu}\Norm{\fp}_2
		=\left(\alpha^{-1}-2\sqrt{\mu}\right) \Norm{\fp}_2,
	\end{align*}
and thus
\begin{align*}
\frac{ \Norm{\fp_{\Sbar}}_1- \Norm{\fp_S}_1}{\iprod{\fp,\f}}
	\geq \frac{\left(\alpha^{-1}-2\sqrt{\mu}\right) \Norm{\fp}_2 }{\eta \Norm{\fp}_2 \Norm{\f}_2}
	= \frac{\left(\alpha^{-1}-2\sqrt{\mu}\right) }{\eta \Norm{\f}_2}.
\end{align*}
If
$\frac{\eta}{1-\epsilon}<
 (\alpha\sqrt{\mu})^{-1}-2,$
this implies that
 \[
 \frac{ \Norm{\fp_{\Sbar}}_1- \Norm{\fp_S}_1}{\iprod{\fp,\f}}
 	>  \frac{\left(\alpha^{-1}-2\sqrt{\mu}\right) }{(1-\epsilon)
\left  ((\alpha\sqrt{\mu})^{-1}-2\right)\Norm{\f}_2}
  = \frac{\sqrt{\mu} }
	{(1-\epsilon)\Norm{\f}_2}
\geq \frac{\Norm{\fz}_1}{\iprod{\fz,\f}},
 \]
from which our desired result follows.
\end{proof}	

\section{Results for Small Set Expansion}
\label{sec:SSE}

As stated in Corollaries~\ref{cor:cayley-sse} and~\ref{cor:sse-approx}, our results imply two consequences for the \textsf{Small Set Expansion} problem of~\cite{RaghavendraS10}.
This is the problem of deciding, given an input graph $G$ and parameters $\delta,\e$, whether there is a measure-$\delta$ subset $S$ of $G$'s vertices where all but an $\e$ fraction of $S$'s edges stay inside it,
or that $G$ is a \emph{small set expander} in the sense that every sufficiently small set has almost all its edges leaving it.
Beyond being a natural problem in its own right, \textsf{Small Set Expansion} is also closely related  to the \textsf{Unique Games} problem whose conjectured hardness is known as Khot's ``Unique Games Conjecture'' \cite{Khot02}.
\cite{RaghavendraS10} gave a reduction from \textsf{Small Set Expansion} to \textsf{Unique Games}.
While a reduction in the other direction is not known, all currently known algorithmic and integrality gap results apply to both problems equally well (e.g., \cite{AroraBS10,RaghavendraST10,BarakGHMRS12,BarakBHKSZ12}),
and thus they are likely to be computationally equivalent.

We give an algorithm to solve \textsf{Small Set Expansion} in quasipolynomial time on an interesting family of Cayley graphs,
and a new polynomial-time approximation algorithm for this problem on general graphs,
with the approximation guarantee depending on the dimension of the input graph's top eigenspace.

\subsection{Small-set expansion of Cayley graphs}






We consider the problem of solving the small set expansion problem on Cayley graphs over $\GF2^{\ell}$.
One reason to consider such graphs is that, until recently, the hardest looking instances for this problem were graphs of this type
(i.e., the noisy hypercube~\cite{KhotV05} and the ``short code'' graph~\cite{BarakGHMRS12}).
\cite{BarakBHKSZ12} showed that these instances can in fact be solved via constant rounds of the SOS hierarchy,
but we still do not have any other good candidate hard instances, and so it is natural to ask whether Cayley graphs can provide such candidates.
Also, since the SOS algorithm does not make use of the algebraic structure of Cayley graphs,
it is plausible that if this algorithm can efficiently solve the \textsf{Small-Set Expansion} on Cayley graphs, then it can in fact solve it on all graphs.

Let $G$ be a Cayley graph on $\GF2^\ell$ with $n=2^\ell$ vertices.
Let $V_{\ge \lambda}$ be the linear subspace spanned by the eigenfunctions of $G$ with eigenvalue at least $\lambda$.
(We identify $G$ here with its random-walk matrix.)
Let $P_\lambda$ be the degree-$4$ polynomial $P_\lambda(f)=\norm{\Pi_{\ge \lambda} f}_4^4$, where $\Pi_{\ge \lambda}$ is the projector into $V_{\ge \lambda}$.
We define $K_\lambda(G)=\norm{P_\lambda}_{\spectral}$.

In this section, we describe approximation algorithms with running times that depend on $K_\lambda(G)$.
The algorithms run in quasipolynomial time if $K_\lambda(G)$ is polylogarithmic.
We will show interesting families of graphs with $K_\lambda(G)=O(1)$.
(See \pref{thm:interesting-cayley}.)

The following theorem shows that low-degree sum-of-squares relaxations can detect $L_4/L_2$-sparse functions in the subspaces $V_{\ge \lambda}$ (in the case when $K_\lambda(G)$ is not too large).
This result follows from \pref{thm:nonneg} and the fact that the polynomial $P_\lambda$ has nonnegative coefficients in an appropriate basis.

\begin{theorem}
  \label{thm:cayley-norms}
  Sum-of-squares relaxations of degree $\e^{-O(1)}K_\lambda(G)^{O(1)}\log n$ provide an additive $\e$-approximation to the maximum of $\norm{f}_4/\norm{f}_2$ over all non-zero functions $f\in V_{\ge \lambda}$.
\end{theorem}

\begin{proof}
  The problem of maximizing $\norm{f}_4/\norm{f}_2$ over the subspace $V_{\ge \lambda}$ is equivalent to maximizing the polynomial $P_\lambda$ over functions with norm $1$.
  (Also notice that $\norm{f}_4\ge \norm{f}_2$ for every function $f$.)
  In order to apply \pref{thm:nonneg}, we need to verify that $P_\lambda$ has nonnegative coefficients in an appropriate basis.
  Since $G$ is a Cayley graph over $\GF2^\ell$, we can take the characters $\{\chi_{\alpha}\}_{\alpha \in \GF2^\ell}$ as an eigenbasis.
  (Here, $\chi_\alpha(x)=(-1)^{\sum_i \alpha_ix_i}$.)
  If we represent $f=\sum_\alpha \hat f_\alpha \chi_\alpha$ in this eigenbasis and let $S_{\ge \lambda} = \{\alpha \mid \lambda_\alpha \ge \lambda\}$ be the indices of the eigenfunctions with eigenvalue at least $\lambda$, then
  \begin{displaymath}
    P_\lambda(f) = \norm{\Pi_{\ge \lambda}f}_4^4 = \E \Bigparen{\sum_{\alpha \in S_{\ge \lambda}} \hat f_\alpha \chi _\alpha}^4
    =  \sum_{\alpha,\beta,\alpha',\beta' \in S_{\ge \lambda}}  \hat f_\alpha \hat f_\beta \hat f_{\alpha'} \hat f_{\beta'}  \E\chi _\alpha\chi_\beta \chi_{\alpha'}\chi_{\beta'}
    = \sum_{\substack{\alpha,\beta,\alpha',\beta' \in S_{\ge \lambda}\\\alpha+\beta=\alpha'+\beta'}} \hat f_\alpha \hat f_\beta \hat f_{\alpha'} \hat f_{\beta'} \mper
  \end{displaymath}
  It follows that $P_{\lambda}$ has nonnegative coefficients in the monomial basis corresponding to the eigenfunctions of~$G$.
  By \pref{thm:nonneg}, sum-of-squares relaxations of degree $\e^{-O(1)}\norm{P_\lambda}_{\spectral}^{O(1)}\log n$ provides an additive approximation to  the maximum of $P_\lambda$ over functions with $\norm{f}^2=\sum_{\alpha}\hat f_\alpha^2=1$.
\end{proof}

Using the characterization of small-set expansion in terms of $L_4/L_2$-sparse functions \cite{BarakBHKSZ12}, \pref{thm:cayley-norms} implies the following approximation algorithm for small-set expansion on Cayley graphs.
This theorem implies \pref{cor:cayley-sse}.

\begin{theorem}
  \label{thm:cayley-sse}
  For some absolute constant $C\ge 1$ and all $\mu, \e>0$ small enough, sum-of-squares relaxations of degree $K_\lambda(G)^{O(1)}\log n$ can distinguish between the following two cases with $\lambda=1-C\varepsilon$.
  \begin{description}
  \item[Yes:] The Cayley graph $G$ contains a vertex set of measure at most  $\mu$ and expansion at most $\e$.
  \item[No:] All vertex sets of measure at most $C/\sqrt \mu$ in $G$ have expansion at least $1-1/C$.
  \end{description}
\end{theorem}

\begin{proof}
  We will show that the maximum of $\norm{f}_4/\norm{f}_2$ over $f\in V_{\ge \lambda}$ distinguishes the two cases (by a constant margin).
  Therefore, \pref{thm:cayley-norms} implies that we can distinguish between the cases using sum-of-squares relaxations.

  \textbf{Yes}-case:
  Let $f$ be the indicator function of a set with measure at most $\mu$ and expansion at most $\e$.
  Then, $\snorm{\Pi_{\ge \lambda} f}\ge 0.99 \snorm{f}$.
  It follows that $\norm{\Pi_{\ge \lambda}f}_4^4\ge 0.9 \norm{f}_4^4$.
  (See \pref{lem:trunc}.)
  Therefore, $\norm{\Pi_{\ge \lambda} f}_4^4/\norm{\Pi_{\ge \lambda} f}_2^4\ge \Omega(1) \norm{f}_4^4/\norm{f}_2^4= \Omega(1)\cdot 1/\mu$.

  \textbf{No}-case:
  Let $\mu'=C/\sqrt \mu$.
  By \cite[Theorem 2.4]{BarakBHKSZ12}, graphs with this kind of small-set expansion satisfy $\norm{f}_4^4/\norm{f}_2^4\le O(1)/(\mu')^2 \ll 1/\mu$ for all functions $f\in V_{\ge \lambda}(G)$.
\end{proof}

The following theorem shows that there are interesting Cayley graphs that satisfy $K_\lambda(G)=O(1)$ for $\lambda=\Omega(1)$.
We consider constructions based on the long code and the short code \cite{BarakGHMRS12}.
These constructions are parameterized by the size of the graph and its eigenvalue gap.
In the context of the Unique Games Conjecture and the Small-Set Expansion Hypothesis, the most relevant case is that the eigenvalue gap is a constant.
(The eigenvalue gap corresponds to the gap to perfect completeness.)

\begin{theorem}
\label{thm:interesting-cayley}
  Long-code and short-code based graphs with constant eigenvalue gap satisfy $K_\lambda(G)=O(1)$ for all $\lambda=\Omega(1)$.
\end{theorem}

\begin{proof}
  By \cite[Lemma 5.1]{BarakBHKSZ12}, there exists a constant $C$ such that $P_\lambda(f) =  C\norm{f}_2^4 - S(f)$ where $S(\cdot)$ is a sum of squares (the same constant $C$ works for both graph constructions).
  Therefore, by Lemma~\ref{lem:spectral norm}, $\norm{P_\lambda}_{\spectral}\le C$.
\end{proof}



\subsection{Approximating small-set expansion using $\ASVP$}

The approximation algorithm for the analytical sparse vector problem (\pref{thm:ASVP}) implies the following approximation algorithm for small-set expansion.
An algorithm for the same problem with the factor $(\dim V_{\ge \lambda})^{1/3}$ replaced by a constant would refute the Small-Set Expansion Hypothesis \cite{RaghavendraS10, RaghavendraST12}.\footnote{
It's plausible that, under standard complexity assumptions such as $\mathbf{NP}\nsubseteq \mathbf{SUBEXP}$, even a smaller improvement to a $(\dim V_{\ge \lambda})^{o(1)}$ factor instead of $(\dim V_{\ge \lambda})^{1/3}$ would refute this hypothesis, though we have no proof of such an implication.}

\begin{theorem}
  For some absolute constant $C\ge 1$ and all $\mu,\e>0$ small enough, sum-of-squares relaxations with constant degree can solve the promise problem on regular graphs $G$:
  \begin{description}
  \item[Yes:] The graph contains a vertex with measure at most $\mu /(\dim V_{\ge \lambda})^{1/3}$ and expansion at most $\e$, where $\lambda=1-C\cdot \e$.
  \item[No:] All vertex sets of measure at most $C\sqrt \mu$ have expansion at least $1-1/C$.
  \end{description}
\end{theorem}

\begin{proof}
  Suppose $G$ satisfies the \textbf{Yes} property.
  Let $f$ be the indicator functions of a set with measure at most $\mu'=\mu/(\dim V_{\ge \lambda})^{1/3}$ and expansion at most $\e$.
  Then, $\norm{\Pi_{\ge \lambda} f}_2^2\ge (1-1/C')\norm{f}_2^2$, where we can make $C'$ as large as we like by making $C$ larger.
  By \pref{thm:ASVP}, constant-degree sum-of-squares relaxations allow us to find an $L_4/L_2$-sparse function $g\in V_{\ge \lambda}$, so that $\norm{g}_4^4\ge \Omega(1/\mu)\norm{g}_2^4$.
  By \cite[Theorem 2.4]{BarakBHKSZ12} (see also Appendix~\ref{app:sse-vs-norm}), such a function certifies that we are not in the \textbf{No} case.
\end{proof}




\section{Discussion and open questions}
\label{sec:discussion}

A general open question is to find other applications of our approach for rounding sum-of-squares relaxations.
Natural candidates would be problems where it seems that they do not display a ``dichotomy'' behavior, where beating some simple algorithm is likely to be
exponentially hard, but rather suggest more of a smooth tradeoff between time and performance.
As far as we are aware, all known ``robust''\footnote{For \textsc{knapsack}-like problems, there exist explicit lower bounds \cite{Grigoriev01b}, but here low-degree sum-of-squares proofs provide very good approximation (in this sense, the lower bound is not robust).}
lower-bound results for the sum-of-squares method are \emph{non-constructive}, i.e., they show that hard instances for the sos method exist but do not give an efficient way of constructing them.
More concretely, the results use the \emph{probabilistic method} and show that with high probability, random instances are hard for sum-of-squares relaxations~\cite{Grigoriev01,Schoenebeck08}.
Therefore, SOS seems promising for problems where random instances do not seem to be the most difficult, e.g., problems related to the Unique Games Conjecture.
A concrete problem of that type to look at is \textsf{Sparsest Cut}. In particular, can we obtain even a small improvement\footnote{An approach to obtain constant-factor approximations for \textsc{sparsest cut} in subexponential time is outlined in the dissertation~\cite[Chapter 9]{Steurer10d}. However, this approach also works with weaker hierarchies. An approach tailored to sum-of-squares would be interesting.} to \cite{AroraRV04}'s algorithm using more SOS levels?
In fact, we believe that even finding a natural reinterpretation of \cite{AroraRV04} result in our framework would be interesting.
That said, our result for finding a planted sparse vector shows that SOS can be useful for average-case problems as well, and in particular we believe SOS might be a strong tool
for solving  unsupervised learning problems, especially for nonlinear models.

A relaxation-based approximation algorithm can be thought of as having three components: the relaxation, the rounding algorithm, and its analysis.
In our approach  there is almost no creativity in choosing the relaxation, which is simply taken to be a sufficiently high level of the SOS hierarchy.
(Though there may be some flexibility in how we represent solutions.)
Can we similarly show a ``universal'' rounding algorithm, thus pushing all the creative choices into the \emph{analysis}?
A related question is whether one can formulate a theorem giving a translation from combining algorithms into rounding algorithms under sufficiently general conditions,
so that results like ours would follow as special cases, and as mentioned in Section~\ref{sec:relwork}, have already made some progress in this direction.

The notion of ``analytically sparse'' vectors seems potentially useful for more applications.
It would be interesting to explore the different choices for $L_q/L_p$ sparsity,
and what tradeoffs they yield in terms of computation time versus usefulness as a proxy for actual sparsity.
In particular, for the planted sparse vector question, it is natural to conjecture that there is an analytical relaxation that we can optimize over in $n^{O(\ell)}$ time,
and can detect sparse vectors in random subspaces of dimension $n^{1-1/\ell}$.

In the context of the Small-Set-Expansion Hypothesis / Unique Games Conjecture, the most important question is whether our results of Section~\ref{sec:asvp}
can be further improved. 
We do not know of any candidate hard instances for this problem (in the relevant range of parameters) and so conjecture that our algorithm (or at least our analysis of it) is not optimal and 
can be improved further. 

Related to the question of finding hard instances, our work suggests a different type of negative results for convex relaxations.
While integrality gaps are instances that are hard for a particular relaxation, regardless of the rounding algorithm,
one can consider the notion of ``combining gaps''.
These will be instances where there is a distribution of good solutions, but a particular combining algorithm $C$ fails to find one.
Hence, viewing $C$ as a rounding algorithm, such a result shows that $C$ will fail regardless of the relaxation used.
(Karloff's work~\cite{Karloff99} on hard instances for the \cite{GoemansW95} hyperplane cut rounding algorithms can be viewed as such an example.)
Studying such gaps can shed more light on our approach and computational difficulty in general.
In particular, it might be interesting to consider this question for random satisfiable instances of SAT or other constraint satisfaction problems.


\addreferencesection

\newcommand{\etalchar}[1]{$^{#1}$}
\providecommand{\bysame}{\leavevmode\hbox to3em{\hrulefill}\thinspace}
\providecommand{\MR}{\relax\ifhmode\unskip\space\fi MR }
\providecommand{\MRhref}[2]{%
  \href{http://www.ams.org/mathscinet-getitem?mr=#1}{#2}
}
\providecommand{\href}[2]{#2}

\appendix

\section{Pseudoexpectation toolkit} \label{app:toolkit}

We recall here the definition of pseudoexpectation from~\cite{BarakBHKSZ12} and prove some of its useful properties.
Some of these were already proven in~\cite{BarakBHKSZ12} but others are new.

\begin{definition}  Let $\pE$ be a functional that maps polynomial $P$ over $\R^n$ of degree at most $r$ into a real number which we denote by $\pE_x P(x)$ or $\pE P$ for short.  We say that $\pE$ is a \emph{level-$r$ pseudo-expectation functional} ($r$-\pef for short)  if it satisfies:
\begin{description}

\item[Linearity] For every polynomials $P,Q$ of degree at most $r$ and $\alpha,\beta \in \R$, $\pE( \alpha P + \beta Q) = \alpha \pE P + \beta \pE Q$.

\item[Positivity] For every polynomial $P$ of degree at most $r/2$, $\pE P^2 \geq 0$.

\item[Normalization] $\pE 1 = 1$ where on the RHS, $1$ denotes the degree-$0$ polynomial that is the constant $1$.
\end{description}
\end{definition}


The functional $\pE$  can be represented by a table of size $n^{O(r)}$ containing the pseudo-expectations of every monomial of degree at most $r$ (or some
other linear basis for polynomials of degree at most $r$). For a linear functional $\pE$, the map $P\mapsto \pE P^2$ is a quadratic form. Hence, $\pE$
satisfies the positivity condition if and only if the corresponding quadratic form is positive semidefinite. It follows that the convex set of level-$r$
pseudo-expectation functionals over $\R^n$ admits an $n^{O(r)}$-time separation oracle, and hence the $r$-round SoS relaxation can be solved up to accuracy
$\e$ in time $(mn\cdot \log(1/\e))^{O(r)}$.

For every random variable $X$ over $\R^n$, the functional $\pE P \seteq \E P(X)$ is a level-$r$ pseudo-expectation functional for every $r$.
As $r\ra\infty$, this hierarchy of pseudo-expectations will converge to the expectations of a true random variable~\cite{Lasserre01},
in general the convergence is not guaranteed to happen in a finite number of steps~\cite{KlerkL11}, although for most problems of interest
in TCS, $n$ levels would suffice for either exact convergence or sufficiently close approximation.




We now record various useful ways in which pseudoexpectations behave close to actual expectations.

For two polynomials $P$ and $Q$, we write $P \sle Q$ if $Q=P + \sum_{i=1}^m R_i^2$ for some polynomials $R_1,\ldots,R_m$.

If $P$ and $Q$ have degree at most $r$, then $P \sle Q$ implies that $\pE P \leq \pE Q$ every $r$-\pef $\pE$.  This follows using linearity and positivity, as
well as the (not too hard to verify) observation that if $Q-P = \sum_i R_i^2$ then it must hold that $\deg(R_i) \leq \max\{\deg(P),\deg(Q)\}/2$ for every $i$.

One of the most useful properties of pseudo-expectation is that it satisfies the Cauchy--Schwarz inequality:

\begin{lemma}[Pseudo Cauchy--Schwarz,\cite{BarakBHKSZ12}]
  \label{lem:pef-cauchy-schwarz}
  Let $P$ and $Q$ be two polynomials of degree at most $r$.
  Then, $\pE PQ\le \sqrt{\pE P^2}\cdot \sqrt{\pE Q^2}$ for any degree-$2r$
  pseudo-expectation functional $\pE$.
\end{lemma}

\begin{proof}
  We first consider the case $\pE P^2,\pE Q^2 >0$.
  Then, by linearity of $\pE$, we may assume that $\pE P^2=\pE Q^2=1$.
  Since $2 PQ\sle P^2 + Q^2$ (by expanding the square $(P-Q)^2$), it
  follows that $\pE PQ \le \tfrac12 \pE P^2 +\tfrac12\pE Q^2 =1$ as
  desired.
  It remains to consider the case $\pE P^2=0$.
  In this case, $2\alpha PQ \sle P^2+\alpha^2Q^2$ implies that $\pE PQ\le
  \alpha\cdot  \tfrac12\pE Q^2$ for all $\alpha>0$.
  Thus $\pE PQ=0$, as desired.
\end{proof}

In particular this implies the following corollary

\begin{corollary}[\cite{BarakBHKSZ12}]\label{cor:pef-constraint} If $P$ is a polynomial of degree $\leq r$, and $\pE_x$ is a $2r$-\pef such
that $\pE P(x)^2 = 0$, then $\pE P(x)Q(x) = 0$ for every $Q$ of degree $\leq r$.
\end{corollary}
\begin{proof} By Lemma~\ref{lem:pef-cauchy-schwarz},
\[
\pE PQ \leq \sqrt{\pE P^2}\sqrt{\pE Q^2} = 0
\]
\end{proof}

In this paper we also need the following variant of \Holder's inequality:

\begin{lemma}[Pseudoexpectation \Holder ] \label{lem:pseudo-holder} Let $d,c,k \in \N$,  $\cD$ be a level $\ell \geq 10dck$ pseudodistribution over $\R^n$, and  $P$ a sum of squares $n$-variate polynomial of degree $d$, then
\[
\pE_{X\sim \cD} P(X)^{r'} \geq \left( \pE_{X\sim \cD} P(X)^r\right)^{r'/r}
\]
where $r = ck$ and $r'=(c+1)k$.
\end{lemma}
\begin{proof}
We'll do the proof by induction on $r$. The base case is $r=c$ in which case this is simply the pseudoexpectation Cauchy Schwarz that $\pE P(X)^{2c} \geq (\pE P(X)^c )^2$.
Define $\cD'$ to be the pseudodistribution obtained by reweighing $\cD$ according to $P(X)^{r-c}$.
Using $\pE_{\cD'} P(X)^{2c} \geq \left( \pE_{\cD'} P(X)^c \right)^2$ we can write
\[
\frac{\pE_{\cD} P(X)^{r+x}}{\pE_{\cD} P(X)^{r-c}} = \frac{\pE_{\cD} P(X)^{r-c} P(X)^{2c}}{\pE_{\cD} P(X)^{r-c}} \geq \left( \frac{\pE_{\cD} P(X)^{r-c} P(X)^{c}}{\pE_{\cD} P(X)^{r-c}} \right)^2
\]
moving things around we get that
\[
\pE_{\cD} P(X)^{r+c} \geq \left( \pE_{\cD} P(X)^r \right)^2 / \pE_{\cD} P(X)^{r-c}
\]
which using our induction hypothesis on $r$ vs $r-c$, we can lower bound by
\[
\left( \pE_{\cD} P(X)^r \right)^2  / \left( \pE_{\cD} P(X)^{r}\right)^{(r-c)/r} = \left( \pE_{\cD} P(X)^r \right)^{(r+c)/r}
\]
\end{proof}

We sometime would need to extend a pseudoexpectation of one random variable to a pseudoexpectation of two independent copies of it.
The following lemma would be useful there

\begin{lemma}\label{lem:pef:independent} Suppose that $X$ and $Y$ are two pseudodistributions of level $\ell$.
Then we can define a level $\ell$ pseudoexpectation operator on $X,Y$ such that for every two polynomials $P$ $Q$ of degree at most $\ell/2$, $\pE P(X)Q(Y) = (\pE P(X) )(\pE Q(Y) )$.
\end{lemma}
\begin{proof} We define the pseudoexpectation operator in the obvious way---for every set of $\ell$ indices $i_1,\ldots,i_k,j_{k+1},\ldots,j_{\ell}$ we let $pE X_{i_1}\cdots X_{i_k}\cdot Y_{j_{k+1}} \cdots Y_{j_d} = ( \pE X_{i_1}\cdots X_{i_k} )\cdot ( \pE Y_{j_{k+1}} \cdots Y_{j_d} )$ and extend it linearly to all monomials.
Clearly $\pE 1 =1$ and so the only thing left to do is to prove that for every polynomial $P$ of degree $\leq \ell/2$ in the $X,Y$ variables $\pE P(X,Y)^2 \geq 0$.

Write $P(X,Y) = \sum M_i(X)N_i(Y)$ where $M_i,N_i$ are monomials, then $P(X,Y)^2 = \sum_{i,j} M_i(X)M_j(X)N_i(Y)N_j(Y)$ and so under our definition
\[
\pE P(X,Y)^2 = \sum_{i,j} (\pE M_i(X) M_j(X)) (\pE N_i(Y) N_j(Y) ) = \iprod{A,B}
\] where $A$ and $B$ are the matrices defined by $A_{i,j} = \pE M_i(X)M_j(X)$ and $B_{i,j} = \pE N_i(Y) N_j(Y)$.
But the pseudoexpectation conditions on $X,Y$ implies that both these matrices are p.s.d and so their dot product is nonnegative.
\end{proof}

We would like to understand how polynomials behave on linear subspaces of $\R^n$.
A map $P\from \R^n\to \R$ is \emph{polynomial} over a linear subspace $V\sse\R^n$ if $P$ restricted to $V$ agrees with a polynomial in the coefficients for
some basis of $V$.
Concretely, if $g_1,\ldots,g_m$ is an (orthonormal) basis of $V$, then $P$ is \emph{polynomial} over $V$ if $P(f)$ agrees with a polynomial in
$\iprod{f,g_1},\ldots,\iprod{f,g_m}$.
We say that $P\sle Q$ holds over a subspace $V$ if $P-Q$, as a polynomial over $V$, is a sum of squares. 
\begin{lemma}[\cite{BarakBHKSZ12}]
  Let $P$ and $Q$ be two polynomials over $\R^n$ of degree at most $r$,
  and let $B\from \R^n\to \R^k$ be a linear operator.
  Suppose that $P\sle Q$ holds over the kernel of $B$.
  Then, $\pE P \le \pE Q$ holds for any $r$-\pef $\pE$ over $\R^n$
  that satisfies $\pE_f \snorm{Bf}=0$.
\end{lemma}
\begin{proof}
  Since $P\sle Q$ over the kernel of $B$, we can write
  $Q(f)=P(f)+\sum_{i=1}^m R_i^2(f)+\sum_{j=1}^k(B f)_j S_j(f)$ for
  polynomials $R_1,\ldots,R_m$ and $S_1,\ldots,S_k$ over $\R^n$.
  By positivity, $\pE_f R_i^2(f)\ge 0$ for all $i\in[m]$.
  We claim that $\pE_f (Bf)_j S_j(f)=0$ for all $j\in[k]$ (which would
  finish the proof).
  This claim follows from the fact that $\pE_f (B f)_j^2 =0$ for all
  $j\in[k]$ and \pref{lem:pef-cauchy-schwarz}.
\end{proof}



\begin{lemma}[\cite{BarakBHKSZ12}]
  \label{lem:boundedness-relation}
  The relation $P^2\sle P$ holds if and only if $0\sle P\sle 1$.
  Furthermore, if $P^2 \sle P$ and $0 \sle Q \sle P$, then $Q^2\sle Q$.
\end{lemma}

\begin{proof}
  If $P\sge 0$, then $P\sle 1$ implies $P^2\sle P$.
  (Multiplying both sides with a sum of squares preserves the order.)
  On the other hand, suppose $P^2\sle P$.
  Since $P^2\sge 0$, we also have $P\sge 0$.
  Since $1-P=P-P^2+(1-P)^2$, the relation $P^2\sle P$ also implies $P\sle 1$.

  For the second part of the lemma, suppose $P^2\sle P$ and $0\sle Q\sle
  P$.
  Using the first part of the lemma, we have $P\sle 1$.
  It follows that $0\sle Q\sle 1$, which in turn implies $Q^2\sle Q$ (using
  the other direction of the first part of the lemma).
\end{proof}

\begin{fact}
  If $f$ is a $d$-f.r.v. over $\R^\cU$ and $\set{P_v}_{v\in
    \cU}$ are polynomials of degree at most $k$, then $g$ with
  $g(v)=P_v(f)$ is a level-$(d/k)$ pseudodistribution over $\R^{\cU}$.
  (For a polynomial $Q$ of degree at most $d/k$, the pseudo-expectation is
  defined as
  \begin{math}
    \pE_g Q(\set{g(v)}_{v\in \cU}) \seteq \pE_f Q(\set{P_v(f)}_{v\in \cU})\mper
  \end{math}
  )
\end{fact}


\begin{lemma}[\cite{BarakBHKSZ12}] \label{lem:linear-cs} For $f,g \in \L2(\cU)$,
\[
\iprod{f,g} \sle \tfrac{1}{2}\norm{f}^2 +\tfrac{1}{2}\norm{g}^2 \mper
\]
\end{lemma}
\begin{proof} The right-hand side minus the LHS equals the square polynomial $\tfrac{1}{2}\iprod{f-g,f-g}$
\end{proof}

Here is another form of the Cauchy--Schwarz inequality.  
\begin{lemma}[Function Cauchy--Schwarz inequality,\cite{BarakBHKSZ12}]\label{lem:pef-cauchy-Schwarz-inside}
  If $(f,g)$ is a level-$2$ \pd over $\R^\cU\times \R^{\cU}$,
  then
  \begin{displaymath}
    \pE_{f,g} \iprod{f,g}
    \le \sqrt{\pE_f \snorm{f}}\cdot \sqrt{\pE_g
      \snorm{g}}
    \mper
  \end{displaymath}
\end{lemma}

\begin{proof}
  Let $\bar f =f/\sqrt{\pE_f \snorm{f}}$ and $\bar g=g/\sqrt{\pE_g
    \snorm{g}}$.
  Note $\pE_{\bar f} \snorm {\bar f} = \pE_{\bar g}\snorm{\bar g}=1$.
  Since by Lemma~\ref{lem:linear-cs}, $\iprod{\bar f,\bar g}\sle \half \snorm{\bar f} + \half \snorm{\bar
    g}$, we can conclude the desired inequality,
  \begin{displaymath}
    \pE_{f,g} \iprod{f,g} = \sqrt{\pE_f \snorm{f}}\cdot \sqrt{\pE_g
      \snorm{g}} \pE_{\bar f,\bar g} \iprod{\bar f,\bar g}
    \le \sqrt{\pE_f \snorm{f}}\cdot \sqrt{\pE_g  \snorm{g}}
    \cdot \underbrace{\Paren{\tfrac12 \pE_{\bar f} \snorm{\bar f} + \tfrac
        12\E_{\bar g} \snorm{\bar g}}}_{=1}
    \mper\qedhere
  \end{displaymath}
\end{proof}

And it implies another form of \Holder 's inequality
\begin{corollary}[Function \Holder's inequality,\cite{BarakBHKSZ12}]\label{cor:pef-holder-inside}
If $(f,g)$ is a level $4$ \pd over $\R^\cU\times \R^{\cU}$, then
\begin{align*}
\pE_{f, g}  \E_{\elem \in \cU} f(\elem)g(\elem)^3 \leq \left(\pE_f \norm{f}_4^4\right)^{1/4} \left(\pE_g \norm{g}_4^4 \right)^{3/4} .
\end{align*}
\end{corollary}
\begin{proof}
Using \pref{lem:pef-cauchy-schwarz} twice, we have
\begin{align*}
\pE_{f, g}  \E_{\elem \in \cU} f(\elem)g(\elem)^3 \leq & \left(\pE_{f, g} \E_{\elem \in \cU} f(\elem)^2 g(\elem)^2\right)^{1/2} \left(\pE_{g} \norm{g}_4^4\right)^{1/2}
\leq \left(\pE_f \norm{f}_4^4\right)^{1/4} \left(\pE_g \norm{g}_4^4 \right)^{3/4} .
\end{align*}
\end{proof}


\subsection{Spectral norm and SOS proofs} \label{sec:spectralnorm}

Here we note the following alternative characterization of the spectral norm of a polynomial:
\begin{lemma} \label{lem:spectral norm}
Let $P$ be a degree-$4$ homogenous polynomial , then  $\spectralnorm{P} \leq c$ if and only if there is a sum of squares degree $4$ polynomial $S$ such that
$P(x) = c\norm{x}_2^4 - S(x)$.
\end{lemma}
\begin{proof}
Suppose that $\spectralnorm{P} \leq C$. Then there is an $n^2 \times n^2$ matrix $M$ such that $M\cdot x^{\ot 4} = P(x)$ for all $x$  $M  = cI - S$ where $I$ is the $n^2\times n^2$ identity and $S$ is
a positive semidefinite matrix. That is, $S = \sum \lambda_i Q_i^{\ot 2}$ for some $\lambda_i \geq 0$ and $Q_i \in \R^{n^2}$.
Now, if we consider $S$ as a degree $4$ polynomial $S(x) = S\cdot x^{\ot 4}$ then it equals $\sum \lambda_i (Q_i\cdot x^{\ot 2})^2$ and hence it is a sum of squares, and it satisfies
\[
P(x) = cI \cdot x^{\ot 4} - S(x) = c\sum_{i,j} x_i^2x_j^2  - S(x) = c\norm{x}_2^4 - S(x) \mper
\]
On the other hand, suppose that $P(x) = c\norm{x}_2^4 - \sum R_i(x)^2$ where the $R_i$'s are quadratic polynomials. We can let $r_i \in \R^{n^2}$ be such that $r_i \cdot x^{\ot 2} = R_i(x)$, and then let $M$ be the
quadratic operator on $R^{n^2}$ such that $M(y) = c\norm{y}_2^2 - \sum (r_i \cdot y)^2$ for every $y\in\R^{n^2}$. One can easily verify that the spectral norm of $M$ is at most $c$ and  $M\cdot x^{\ot 4} = P(x)$ for every $x\in\R^n$.
\end{proof}

\section{Low-Rank Tensor Optimization} \label{sec:lowrank}

For a vector $x\in \R^n$, let $\lVert x \rVert$ denote the Euclidean norm of $x$.
For a polynomial $P\in \R[X_1,\ldots,X_n]$, we define its \textbf{norm} as $\lVert P \rVert\defeq \max\{\lvert P(x)\rvert \mid \lVert x \rVert=1 \}$.

Consider an $n$-variate degree-$4$ polynomial $P$ of the form $P(x) = \sum_{i=1}^r Q_i(x)^2$ for quadratic polynomials $Q_1,\ldots,Q_r$ .

\begin{theorem} \label{thm:lowrank}
  There exists an algorithm that, given $P$ and $\e$, computes $\lVert P \rVert$ up to multiplicative error $\e$ in time $\exp(\poly(r,\e))$.
\end{theorem}

\begin{proof}
For $\lambda\in \R^n$, consider the polynomial $Q_\lambda(x) = \sum_{i=1}^r \lambda_i Q_i(x)$.

First, we claim that $\max_{ \lVert \lambda \rVert = 1 }\lVert Q_\lambda \rVert = \lVert P \rVert^{1/2}$.
On the one hand, $\lVert Q_\lambda \rVert \le \lVert \lambda \rVert \cdot \lVert P \rVert^{1/2}$ by Cauchy--Schwarz.
On the other hand, if $\lambda=\frac 1 {P(x^\ast)^{1/2}}( Q_1(x),\ldots,Q_r(x))$ for some vector $x^\ast\in\R^n$, then $Q_\lambda(x^\ast)=P(x^\ast)^{1/2}$.
Therefore, if we choose $x^\ast$ as a unit vector that maximizes $P$, then $\lVert Q_\lambda \rVert\ge P(x^\ast)^{1/2}=\lVert P \rVert^{1/2}$.

Next, we claim that we can compute $\max_{ \lVert \lambda \rVert = 1 }\lVert Q_\lambda \rVert $ up to error $\e$ in time $\exp(\poly(r,\e))$.
Since $Q_\lambda$ is quadratic, we can compute $\norm{Q_\lambda}$ in polynomial time.
(The norm of $Q_\lambda$ is equal to the largest singular value of the coefficient matrix of $Q_\lambda$.)
The idea is to compute $\lVert Q_\lambda \rVert$ for all vectors $\lambda\in N_\e$, where $N_\e$ is an $\e$-net of the unit ball in $\R^r$.
Let  $\lambda^*$  be the vector that achieves the maximum, $x^*$ be the corresponding input, and $u^*$ be the vector $(Q_1(x^*),\ldots,Q_r(x^*))$.
Thus $\max_{ \lVert \lambda \rVert = 1 }\lVert Q_\lambda \rVert^2  = \iprod{\lambda^*,u^*}^2$ and $\norm{u^*} = \norm{P}$.
Therefore, for every $\lambda$
\[
\norm{Q_\lambda}^2 \geq \norm{Q_\lambda(x^*)}^2 = \iprod{\lambda,u^*}^2 \mper
\]
But if $\norm{\lambda-\lambda^*} \leq \e$ then
\[
|\iprod{\lambda^*,u^*} - \iprod{\lambda,u^*}| \leq \norm{\lambda-\lambda^*}\norm{u^*}  = \norm{\lambda-\lambda^*}\norm{P} \mper
\]
Thus if $\norm{\lambda-\lambda'} \leq \e$ then we get a $1-O(\e)$ multiplicative approximation to $\norm{P}$.
\end{proof}

\begin{corollary}\label{cor:frobenius}
 If $M$ is a symmetric $n^2\times n^2$ PSD matrix with Frobenius norm at most $1$ then we can compute an $\e$ additive approximation to
\[
\max_{\norm{x}=1} \iprod{M,x^{\ot 4}}
\]
in $\poly(n)\exp(\poly(1/\e))$ time.
\end{corollary}
\begin{proof}
Write $M$ in its eigenbasis as $M = \sum \lambda_i Q_i^{\ot 2}$ for $n\times n$ matrices $\{Q_i\}$ with Frobenius norm at most $1$,
and let $M' = \sum_{\lambda_i \geq \e} \lambda_i Q_i^{\ot 2}$.
Since $\sum \lambda_i^2 = 1$ we know that the rank of $M'$ is at most $1/\e^2$,
and therefore we can compute a $1\pm \e$ multiplicative approximation to the maximum of
$\iprod{M',x^{\ot 4}}$ over unit $x$
 (which in particular implies an $\e$ additive approximation since this value is bounded by $1$).
But this implies an $\e$-additive approximation  for this maximum over $M$ since these quantities can differ by at most $\e$.
\end{proof}

\section{LOCC Polynomial Optimization}
\label{sec:locc}

Let $P\in \R[X]_4$ be a degree-$4$ homogeneous polynomial of the form $P(X)=A_1(X)\cdot B_1(X) + \cdots + A_m(X)\cdot B_m(X)$ for quadratic polynomials $A_1,\ldots,A_m\in \R[X]_2$ with $0\preceq A_i \preceq \lVert  X \rVert^2$ and quadratic polynomials $B_1,\ldots,B_m\in \R[X]_2$ with $B_i\succeq 0$ and $\sum_i B_i \preceq \lVert X \rVert^2$. 
Note that this corresponds to the tensor corresponding to $P$ having a one-way local operations and classical communication (LOCC) norm bounded by $1$~\cite{BrandaoCY11}. 
Without loss of generality, we may assume $\sum_i B_i = \lVert X \rVert^2$.
(We can choose $A_m=0$ and choose $B_m$ appropriately without changing $P$.)
Our goal is to compute the norm of $P$, defined as $\lVert  P \rVert=\max_{\lVert  x \rVert=1}\abs{P(x)}$. 
In the quantum setting, this corresponds to finding the maximum probability of acceptance by a separable state for the measurement operator $P$.

In this section, we will show that sum-of-squares relaxations provide good approximation for the norm of polynomials of the form above.
(For the case that the variables of $A_1,\ldots,A_m$ are disjoint from the variables for $B_1,\ldots,B_m$, the theorem is due to \Brandao, Christandl, and Yard~\cite{BrandaoCY11}.
Up to the Gaussian rounding step, the proof here is essentially the same as the proof by \Brandao and Harrow \cite{BrandaoH13}.)

\begin{theorem}
\label{thm:locc}
  Sum-of-squares relaxations with degree $d$ achieve the following approximations for the norm of degree-$4$ polynomials $P$ of the form above:
  \begin{itemize}
  \item the value of the relaxation is at most $3\lVert  P \rVert+\e$ for degree $d\ge O(1/\e^2)\cdot \log n$.
  \item in the case that the variables in $A_1,\ldots,A_m$ are disjoint from the variables in $B_1,\ldots,B_m$, the value of the relaxation is at most $\lVert  P \rVert+\e$ for degree $d\ge O(1/\e^2)\log n$.
  \end{itemize}
\end{theorem}

\paragraph{Direct Rounding}

As direct rounding for a distribution $\{X\}$, we choose a Gaussian variable with the same first two moments as $\{X\}$.
To analyze this rounding procedure, the following lemma is useful.

The lemma considers an arbitrary distribution over unit vectors in $\R^n$ (intended to maximize $P$).
We express the second moment $\rho=\E XX^\top$ of this distribution as convex combination $\rho=\sum_i \beta_i\rho_i$ for $\beta_i = \E_X B_i(X)$ and $\rho_i=\E_XXX^T B_i(X)/\beta_i$.
By the assumptions on the distribution $\{X\}$, the matrices $\rho$ and $\rho_1,\ldots,\rho_m$ are positive semidefinite and have trace $1$ (\emph{density matrices}).
The quantum entropy $H(\rho)=-\Tr \rho \log \rho$ is concave so that $H(\rho)\ge \sum_i \beta_i H(\rho_i)$.
The assumption of the lemma is that the inequality is approximately tight.
Roughly speaking, this condition means that the $\rho_i$ matrices are close $\rho$.
For the distribution $\{X\}$, this condition means that reweighing by the polynomials $B_i$ does not affect second moments of the distribution.
We say that the distribution has \emph{low global correlation} with respect to the polynomials $B_1,\ldots,B_m$.
(This notion is related to but distinct from the notion of global correlation in \cite{BarakRS11}).
The lemma asserts that if the distribution $\{X\}$ has low global correlation with respect to the polynomials $B_1,\ldots,B_m$, then sampling $X$ independently for the $A$-part and $B$-part of the polynomial $P$ gives roughly the same value as sampling $X$ in a correlated way.
(The next lemma explains why our direct rounding achieves at least the quantity corresponding to sampling $X$ independently for the two parts.)

\begin{lemma}
  \label{lem:locc-independent}
  Let $\{X\}$ be a distribution over $\R^n$ that satisfies the constraint $\lVert  X \rVert^2=1$.
  Suppose $\sum_i \beta_i H(\rho_i) \ge H(\sum_i \beta_i \rho_i) -\e^2 $ for $\rho_i=\E_X XX^\top B_i(X)/\beta_i$ and $\beta_i=\E_X B_i(X)$.
  Then,
  \begin{displaymath}
    \sum_i \E_X A_i(X)\cdot \E_X B_i(X)  \ge \sum_i \E_X  A_i(X)B_i(X)\quad -\e\mper
  \end{displaymath}
  Moreover, the statement holds if $\{X\}$ is a degree-$4$ pseudo-distribution.
\end{lemma}

\begin{proof}

\newcommand{\Diag}{\mathop{\mathrm{Diag}}}

Consider the block-diagonal density matrix
\begin{math}
  \rho
  = \sum_i \beta_i \rho_i \otimes e_ie_i^\top
  \mcom
\end{math}
and the block-diagonal measurement matrix
\begin{math}
  A
  = \sum_i A_i\otimes e_ie_i^\top
  \mper
\end{math}
(In this construction, we identify the quadratic polynomial $A$ with its representation as a symmetric square matrix.)
Furthermore, consider the partial traces $\rho_A=\Tr_B \rho=\sum_i \beta_i \rho_i$ and $\rho_B=\Tr_A \rho=\sum_i\beta_ie_ie_i^T$.
We can express the two sides of the conclusion of the lemma as follows,
\begin{align*}
  & \sum_i \E_{X} A_i(X)\cdot B_i(X) = \Tr A \rho\mcom
  \\
  &\sum_i  \E_{X} A_i(X)\cdot \E_X B_i(X) = \Tr A (\rho_A\otimes \rho_B)\mper
\end{align*}
Since $A$ has spectral norm at most $1$, we can bound the difference by$\lvert \Tr A (\rho - \rho_A\otimes\rho_B)\rvert\le \lVert  \rho-\rho_A\otimes\rho_B \rVert_*$.
(Here, $\lVert  \cdot  \rVert_*$ is the trace norm---the dual of the spectral norm.)
By Pinsker's inequality, $\lVert  \rho-\rho_A\otimes \rho_B \rVert^2\le H(\rho_A) + H(\rho_B)-H(\rho)$.
By the chain rule, $H(\rho)=H(\rho_B)+\sum_i \beta_i H(\rho_i)$.
The assumption of the lemma allows us to bound the trace norm by $\lVert  \rho-\rho_A\otimes \rho_B \rVert^2\le H(\sum_i \beta_i\rho_i)-\sum_i\beta_iH(\rho_i)\le \e^2$.
At this point, the conclusion of the lemma follows from the bound $\lvert \Tr A (\rho - \rho_A\otimes\rho_B)\rvert\le \lVert  \rho-\rho_A\otimes \rho_B \rVert_*\le \e$.
\end{proof}

The following lemma shows that  Gaussian rounding achieves a value at least as large as the value achieved by sampling $\{X\}$ independently for the $A$-part and $B$-part of the polynomial $P$.

\begin{lemma}
  Let $\{X\}$ be a distribution over $\R^n$ that satisfies the constraint $\lVert X \rVert^2=1$.
  Suppose $\{X'\}$ is a Gaussian distribution with the same first two moments as $\{X\}$.
  Then, $\{X'\}$ satisfies $\E_{X'}\lVert  X' \rVert^2=1$, $\E_{X'}\lVert  X' \rVert^4=3$, and
  \begin{displaymath}
    \E_{X'} \sum_i A_i(X')\cdot B_i(X') \ge \sum_i \E_X A_i(X) \cdot \E_X B_i(X)\mper
  \end{displaymath}
  Moreover, the statement holds if $\{X\}$ is a degree-$4$ pseudo-distribution.
\end{lemma}

\begin{proof}
  Using the assumption $A_i,B_i\succeq 0$, the lemma follows from the fact that Gaussian variables $P,Q$ satisfy $\E P^2Q^2\ge \E P^2 \E Q^2$.
\end{proof}

The previous two lemmas together yield the following corollary.

\begin{corollary}
  \label{cor:locc-direct}
  Let $\{X\}$ be a distribution over $\R^n$ that satisfies the constraint $\lVert X \rVert^2=1$.
  Suppose $\{X'\}$ and $\e>0$ are as in the previous two lemmas, that is,
  $\{X'\}$ is a Gaussian distribution with the same first two moments as $\{X\}$
  and $\sum_i \beta_i H(\rho_i) \ge H(\sum_i \beta_i \rho_i) -\e^2 $ for $\rho_i=\E_X XX^\top B_i(X)/\beta_i$ and $\beta_i=\E_X B_i(X)$.
  Then,
  \begin{displaymath}
    \E_{X'} P(X') \ge \E_{X} P(X) - \e \text{ and }\E_{X'}\lVert  X' \rVert^4 = 3\mper
  \end{displaymath}
  Moreover, the statement holds if $\{X\}$ is a degree-$4$ pseudo-distribution.
\end{corollary}

\paragraph{Making progress}

The following lemma shows that there exists a low-degree polynomial so that reweighing by the polynomial results in a distribution that has low global correlation with respect to the polynomials $B_1,\ldots,B_m$.

\begin{lemma}
\label{lem:locc-progress}
  Let $\{X\}$ be a distribution over $\R^n$ that satisfies the constraint $\lVert X \rVert^2=1$.
  Then, there exists a polynomial  $B\in\R[X]_{2d}$ of the form $B=B_{i(1)}\cdots B_{i(d)}$ with $d=O(1/\e^2)\log n$
  such that $\sum_i \beta_i H(\rho_i) \ge H(\sum_i \beta_i \rho_i) -\e^2 $ for $\rho_i=\E_X XX^\top B(X) B_i(X)/\beta_i$ and $\beta_i=\E_X B(X) B_i(X)$.
  Moreover, the statement holds if $\{X\}$ is a degree-$d+4$ pseudo-distribution.
\end{lemma}

\begin{proof}
  By contraposition, suppose that $\sum_i \beta_i H(\rho_i) < H(\sum_i \beta_i \rho_i) -\eta$ holds for all polynomials $B$ of the form $B=B_{i(1)}\cdots B_{i(d')}$ with $d'\le d= 10/\e^2 \cdot \log n$.
  Then, we can greedily construct a sequence of polynomial $B_{i^\ast(1)},\ldots,B_{i^\ast(d)}$ such that in each step the entropy decreases by at least $\eta$.
  In particular, $H(\rho^\ast)\le H(\rho) - \eta \cdot d$ for $\rho^\ast \propto \E_X XX^\top B_{i^*(1)}\cdots B_{i^\ast (d)}(X)$ and $\rho=\E_X XX^\top$.
  Since $H(\rho)\le \log n$ and $H(\rho^\ast)\ge 0$, we have $\eta\ge 1/d \cdot \log n = \e^2/10$.
  As desired it follows that there exists a polynomial $B$ of the desired form such that $\sum_i \beta_i H(\rho_i) \ge H(\sum_i \beta_i \rho_i) -\e^2$.
\end{proof}

\paragraph{Putting things together}

The following lemma combines the conclusion about direct-rounding and making-progress.

\begin{lemma}
  Let $\{X\}$ be a distribution over $\R^n$ that satisfies the constraints $\lVert  X \rVert^2=1$ and $P(X)\ge c$.
  Then, there exists a polynomial $B\in\R[X]_{2d}$ of the form $B=B_{i(1)}\cdots B_{i(d)}$ with $d=O(1/\e^2)\log n$ such that the Gaussian distribution $X'$ that matches the first two moments of $\{X\}$ reweighted by $B(X)$ satisfies
  \begin{displaymath}
    \E_{X'} P(X)\ge c-\e \text{ and } \E_{X'} \lVert  X' \rVert^4 = 3\mper
  \end{displaymath}
  (Concretely, $\set{X'}$ is the Gaussian distribution that satisfies $\E_{X'} Q(X')=\E_{X}Q(X)B(X)/\E_XB(X)$ for quadratic polynomial $Q$).
  Moreover, the statement holds for degree-$2d+4$ pseudo-distributions.
\end{lemma}

\begin{proof}
  Take the polynomial $B$ as in \pref{lem:locc-progress}.
  Reweigh the distribution $\{X\}$ by the polynomial $B$.
  Apply \pref{cor:locc-direct} to the resulting distribution.
\end{proof}

At this time, we have all ingredients for the proof of \pref{thm:locc}.

\begin{proof}[Proof of \pref{thm:locc}]
  Let $\{X\}$ be a degree-$d+4$ pseudo-distribution over $\R^n$ that satisfies the constraints $\lVert  X \rVert^2=1$ and $P(X)\ge c$ for $d=O(1/\e^2)\log n$.
  By the previous lemma, there exists a distribution $\{X'\}$ over $\R^n$ such that $\E_{X'} P(X')/\E_{X'}\lVert  X' \rVert^4\ge c/3-\e$.
  It follows that there exists a vector $x\in\R^n$ with $P(x)/\lVert x\rVert^4\ge c/3-\e$.
  (We can also find such a vector efficiently because we can sample from the distribution $\{X'\}$ efficiently and the random variables $P(X')$ and $\lVert X'\rVert$ are well-behaved.)
  By homogeneity, we get $\lVert  P \rVert\ge c/3-\e$.

  In the case that the variables $Y$ in $A_1,\ldots,A_m$ are disjoint from the variables $Z$ in $B_1,\ldots,B_m$, we can modify the direct-rounding distribution $\{X'\}=\{(Y',Z')\}$ slightly and sample the variables $Y'$ for the $A_i$ polynomials independently from the variables $Z'$ for the $B_i$ polynomials.
  By \pref{lem:locc-independent}, we still have $\E_{X'} P(X')\ge c-\e$.
  We can assume that $\E\lVert Y'\rVert ^2=\E\lVert  Z' \rVert^2=1/2$ (by adding the corresponding constraint to the sos relaxation).
  Therefore, $\E\lVert  Y' \rVert^2\lVert  X' \rVert^2=1/4$.
  It follows that there exists a vector $x=(y,z)$ in $\R^n$ with $P(y,z)/(\lVert  y \rVert^2\cdot \lVert  z \rVert^2)\ge 4(c-\e)$.
  By homogeneity, we can assume that $\lVert  y \rVert^2=1/2$ and $\lVert  z \rVert^2=1/2$.
  In this case, $\lVert x\rVert^2=1$ and $P(x)\g c-\e$ as desired.
\end{proof}

\section{The 2-to-q norm and small-set expansion} \label{app:sse-vs-norm}

This appendix  reproduces from \cite{BarakBHKSZ12} the proof that  a graph is a \emph{small-set expander} if and only if the projector to the subspace of its adjacency matrix's top eigenvalues has
a bounded $2\to q$ norm for even $q \geq 4$.
We also note that while \cite{BarakBHKSZ12} stated their result for the decision question, it does yield an efficient algorithm to transform a vector in the top eigenspace with large $4$ norm
into a small set that does not expand.

\paragraph{Notation} For a regular graph $G=(V,E)$ and a subset $S \subseteq V$, we define the \emph{measure} of $S$ to be $\mu(S)=|S|/|V|$
and we define $G(S)$ to be the distribution obtained by picking a random $x\in S$ and then outputting a random neighbor $y$ of $x$. We define the
\emph{expansion} of $S$, to be  $\bd_G(S)=\Pr_{y \in G(S)}[ y\not\in S]$, where $y$ is a random neighbor of $x$. For $\delta \in (0,1)$, we define
$\bd_G(\delta)=\min_{S\subseteq V: \mu(S)\leq \delta} \bd_G(S)$. We often drop the subscript $G$ from $\bd_G$ when it is clear from context. We identify $G$
with its normalized adjacency (i.e., random walk) matrix. For every $\lambda \in [-1,1]$, we denote by $V_{\ge \lambda}(G)$ the subspace spanned by the
eigenvectors of $G$ with eigenvalue at least $\lambda$.
The projector into this subspace is denoted $P_{\ge \lambda}(G)$.
For a distribution $D$, we let $\cp(D)$ denote the collision probability of $D$ (the probability that two independent samples from $D$ are identical).

Our main theorem of this section is the following:

\begin{theorem} \label{thm:sse-hyper} For every regular graph $G$,  $\lambda >0$ and even $q$,
\begin{enumerate}
\item \emph{(Norm bound implies expansion)} For all $\delta>0 ,\e>0$, $\Vert P_{\ge\lambda}(G) \Vert_{2 \to q} \leq \e/\delta^{(q-2)/2q}$ implies that
 $\bd_G(\delta) \geq 1-\lambda - \e^2$.

\item \emph{(Expansion implies norm bound)} There is a constant $c$ such that for all $\delta>0$, $\bd_G(\delta) > 1 - \lambda 2^{- cq}$ implies $\Vert
    P_{\ge\lambda}(G) \Vert_{2 \rightarrow q} \leq 2/\sqrt{\delta}$.
    Moreover there is an efficient algorithm such that given a function  $f\in V_{\geq \lambda}(G)$ such that
$\norm{f}_q > 2\norm{f}_2/\sqrt{\delta}$ finds a set $S$ of measure less than $\delta$ such that $\bd_G(S) \leq 1 - \lambda 2^{-cq}$.
\end{enumerate}
\end{theorem}



\begin{corollary} \label{thm:ssetohyper} If there is a polynomial-time computable relaxation $\cR$ yielding good approximation for the $2\to q$, 
then the \emph{Small-Set Expansion Hypothesis} of \cite{RaghavendraS10} is false.
\end{corollary}
\begin{proof}
Using~\cite{RaghavendraST12}, to refute the small-set expansion hypothesis it is enough to come up with an efficient algorithm that given an input graph $G$ and sufficiently small $\delta>0$,
can distinguish between the \emph{Yes} case: $\bd_G(\delta) < 0.1$ and the \emph{No} case $\bd_G(\delta') > 1-2^{-c\log(1/\delta')}$ for any $\delta'\geq \delta$ and some constant $c$.
In particular for all $\eta>0$ and constant $d$, if $\delta$ is small enough then in the \emph{No} case $\bd_G(\delta^{0.4}) > 1-\eta$.
Using Theorem~\ref{thm:sse-hyper}, in the \emph{Yes} case we know $\tfnorm{V_{1/2}(G)} \geq 1/(10\delta^{1/4})$,
while in the \emph{No} case, if we choose $\eta$ to be smaller then $\eta(1/2)$ in the Theorem,
then we know that $\tfnorm{V_{1/2}(G)} \leq 2/\sqrt{\delta^{0.2}}$.
Clearly, if we have a good approximation for the $2\to 4$ norm then, for sufficiently small $\delta$ we can distinguish between these two cases.
\end{proof}

The first (easier) part of Theorem~\ref{thm:sse-hyper} is proven in Section~\ref{app:hyper-imp-sse}. The second part will follow from the following lemma:

\begin{lemma}\label{lem:ssetonorm} Set $e = e(\lambda, q) := 2^{c q} /\lambda$, with a constant $c \leq 100$.
Then for every $\lambda>0$ and $1 \geq \delta \geq 0$, if $G$ is a graph that satisfies
\begin{equation}
\cp(G(S)) \leq 1/(e|S|) \label{eq:cp-expansion}
\end{equation}
 for all $S$ with $\mu(S)\leq \delta$,
then $\norm{f}_q \leq 2\norm{f}_2/\sqrt{\delta}$ for all $f\in V_{\geq \lambda}(G)$.
Moreover, there is an efficient algorithm that given a function $f\in V_{\geq \lambda}(G)$ such that
$\norm{f}_q > 2\norm{f}_2/\sqrt{\delta}$ finds a set $S$ that violates (\ref{eq:cp-expansion}).
\end{lemma}

\paragraph{Proving the second part of Theorem~\ref{thm:sse-hyper} from Lemma~\ref{lem:ssetonorm}}
We use the variant of the local Cheeger bound obtained in~\cite[Theorem 2.1]{Steurer10c},
stating that if $\bd_{G}(\delta) \geq 1 - \eta$ then for every $f\in \L2(V)$ satisfying $\norm{f}_1^2 \leq \delta \norm{f}_2^2$, $\norm{Gf}_2^2 \leq
c\sqrt{\eta}\norm{f}_2^2$.
The proof follows by noting that for every set $S$, if $f$ is the characteristic function of $S$ then $\norm{f}_1 = \norm{f}_2^2 =
\mu(S)$, and $\cp(G(S)) = \norm{Gf}_2^2 / (\mu (S) |S|)$.
Because this local Cheeger bound is algorithmic (and transforms a function with large $L_2/L_1$ ratio into a set by simply using a threshold cut),
this part is algorithmic as well.
\qed

\begin{proof}[Proof of Lemma~\ref{lem:ssetonorm}]  Fix $\lambda>0$.
We assume that the graph satisfies the condition of the Lemma with $e = 2^{c q} /\lambda$, for a constant $c$ that we'll set later.
Let $G=(V,E)$ be such a graph, and $f$ be function in $V_{\geq \lambda}(G)$ with $\norm{f}_2=1$ that maximizes $\norm{f}_q$.
We write $f = \sum_{i=1}^m \alpha_i \chi_i$ where $\chi_1,\ldots,\chi_m$ denote the eigenfunctions of $G$ with values $\lambda_1,\ldots,\lambda_m$ that are at least $\lambda$.
Assume towards a contradiction that $\norm{f}_q>2/\sqrt{\delta}$. We'll prove that  $g= \sum_{i=1}^m (\alpha_i/\lambda_i)\chi_i$ satisfies $\norm{g}_q \geq 10\norm{f}_q/\lambda$.
This is a contradiction since (using $\lambda_i \in [\lambda,1]$) $\norm{g}_2 \leq \norm{f}_2/\lambda$, and we assumed $f$ is a function in $V_{ \geq \lambda}(G)$ with a maximal ratio of $\norm{f}_q/\norm{f}_2$.
(To prove the ``moreover'' part, where we don't assume $f$ is the maximal function, we repeat this process with $g$ until we get stuck.)

Let $U \subseteq V$ be the set of vertices such that $|f(x)| \geq 1/\sqrt{\delta}$ for all $x\in U$. Using Markov and the fact that $\E_{x\in V}[ f(x)^2 ] =
1$, we know that $\mu(U)=|U|/|V| \leq \delta$, meaning that under our assumptions any subset $S\subseteq U$ satisfies $\cp(G(S))\leq 1/(e|S|)$. On the other
hand, because $\norm{f}_q^q \geq 2^q/\delta^{q/2}$, we know that $U$ contributes at least half of the term $\norm{f}_q^q = \E_{x\in V} f(x)^q$. That is, if we
define $\alpha$ to be $\mu(U)\E_{x\in U} f(x)^q$ then $\alpha \geq \norm{f}_q^q/2$. We'll prove the lemma by showing that $\norm{g}_q^q \geq 10\alpha/\lambda$.

Let $c$ be a sufficiently large constant ($c=100$ will do). We define $U_i$ to be the set $\{x \in U : f(x) \in [c^i/\sqrt{\delta},c^{i+1}/\sqrt{\delta}) \}$,
and let $I$ be the maximal $i$ such that $U_i$ is non-empty. Thus, the sets $U_0,\ldots,U_I$ form a partition of $U$ (where some of these sets may be empty).
We let $\alpha_i$ be the contribution of $U_i$ to $\alpha$. That is, $\alpha_i = \mu_i\E_{x\in U_i} f(x)^q$, where $\mu_i=\mu(U_i)$. Note that $\alpha =
\alpha_0 + \cdots + \alpha_I$. We'll show that there are some indices $i_1,\ldots,i_J$ such that:

\begin{description}

\item[(i)] $\alpha_{i_1} + \cdots + \alpha_{i_J} \geq \alpha/(2c^{10})$.

\item[(ii)] For all $j\in [J]$, there is a nonnegative function $g_j:V\to\R$ such that $\E_{x\in V}g_j(x)^q \geq e\alpha_{i_j}/(10c^2)^{q/2}$.

\item[(iii)] For every $x\in V$, $g_1(x)  + \cdots + g_J(x) \leq |g(x)|$.
\end{description}

Showing these will complete the proof, since it is easy to see that for two nonnegative functions and even $q$, $g',g''$, $\E(g'(x)+g''(x))^q \geq \E g'(x)^q
+ \E g''(x)^q$, and hence \textbf{(ii)} and \textbf{(iii)} imply that
\begin{equation}
\norm{g}_4^4= \E g(x)^4 \geq (e/(10c^2)^{q/2}) \sum_j \alpha_{i_j} \;. \label{eq:ssetonorm-finalbound}
\end{equation}
Using \textbf{(i)} we conclude that for $e \geq (10c)^{q}/\lambda$, the right-hand side of (\ref{eq:ssetonorm-finalbound}) will be larger than
$10\alpha/\lambda$.

We find the indices $i_1,\ldots,i_J$ iteratively. We let $\cI$ be initially the set $\{0..I\}$ of all indices. For $j=1,2,...$ we do the following as long as
$\cI$ is not empty:
\begin{enumerate}
\item Let $i_j$ be the largest index in $\cI$.

\item Remove from $\cI$ every index $i$ such that $\alpha_i \leq c^{10}\alpha_{i_j}/2^{i-i_j}$.

\end{enumerate}

We let $J$ denote the step when we stop. Note that our indices $i_1,\ldots,i_J$ are sorted in descending order. For every step $j$, the total of the
$\alpha_i$'s for all indices we removed is less than $c^{10}\alpha_{i_j}$ and hence we satisfy \textbf{(i)}. The crux of our argument will be to show
\textbf{(ii)} and \textbf{(iii)}. They will follow from the following claim:

\begin{claim}\label{clm:sse-norm}  Let $S\subseteq V$ and $\beta>0$ be such that $|S| \leq \delta$ and $|f(x)| \geq \beta$ for all $x\in S$. Then there is a set $T$ of size at least $e|S|$ such that $\E_{x\in T} g(x)^2 \geq \beta^2/4$.
\end{claim}
The claim will follow from the following lemma:

\begin{lemma} \label{lem:sse-norm} Let $D$ be a distribution with $\cp(D) \leq 1/N$ and $g$ be some function. Then there is a set $T$ of size $N$ such that $\E_{x\in T} g(x)^2 \geq (\E  g(D))^2/4$.
\end{lemma}
\begin{proof} Identify the support of $D$ with the set $[M]$ for some $M$, we let $p_i$ denote the probability that $D$ outputs $i$, and sort the $p_i$'s such that $p_1\geq p_2 \cdots p_M$. We let $\beta'$ denote $\E g(D)$; that is, $\beta'=\sum_{i=1}^M p_ig(i)$. We separate to two cases. If $\sum_{i>N} p_ig(i) \geq \beta'/2$, we define the distribution $D'$ as follows: we set $\Pr [ D' = i]$ to be $p_i$ for $i>N$, and we let all $i\leq N$ be equiprobable (that is be output with probability $(\sum_{i=1}^N p_i)/N$). Clearly, $\E  |g(D')|  \geq \sum_{i>N} p_ig(i) \geq \beta'/2$, but on the other hand, since the maximum probability of any element in $D'$ is at most $1/N$, it can be expressed as a convex combination of flat distributions over sets of size $N$, implying that one of these sets $T$ satisfies $\E_{x\in T} |g(x)|  \geq \beta'/2$, and hence $\E_{x\in T} g(x)^2  \geq \beta'^2/4$.

The other case is that $\sum_{i=1}^N p_ig(i) \geq \beta'/2$. In this case we use Cauchy--Schwarz and argue that
\begin{equation}
\beta'^2/4 \leq \left(\sum_{i=1}^N p_i^2\right)\left(\sum_{i=1}^N g(i)^2 \right) \;. \label{eq:cs-sse-norm}
\end{equation}
But using our bound on the collision probability, the right-hand side of (\ref{eq:cs-sse-norm}) is upper bounded by  $\tfrac{1}{N}\sum_{i=1}^N g(i)^2 =
\E_{x\in[N]} g(x)^2$.
\end{proof}
\begin{proof}[Proof of Claim~\ref{clm:sse-norm} from Lemma~\ref{lem:sse-norm}] By construction $f=Gg$, and hence we know that for every $x$, $f(x)=\E_{y\sim x} g(y)$. This means that if we let $D$ be the distribution $G(S)$ then
\[
\E |g(D)| = \E_{x\in S} \E_{y\sim x} |g(y)|  \geq  \E_{x\in S} |\E_{y\sim x}  g(y) | = \E_{x\in S} |f(x)| \geq \beta \;.
\]
By the expansion property of $G$, $\cp(D) \leq 1/(e|S|)$ and thus by Lemma~\ref{lem:sse-norm} there is a set $T$ of size $e|S|$ satisfying  $\E_{x\in T} g(x)^2
\geq \beta^2/4$.
\end{proof}

We will construct the functions $g_1,\ldots,g_J$ by applying iteratively Claim~\ref{clm:sse-norm}. We do the following for $j=1,\ldots,J$:

\begin{enumerate}

\item Let $T_j$ be the set of size $e|U_{i_j}|$ that is obtained by applying Claim~\ref{clm:sse-norm} to the function $f$ and the set $U_{i_j}$. Note that
    $\E_{x \in T_j} g(x)^2 \geq \beta_{i_j}^2/4$, where we let $\beta_i = c^i/\sqrt{\delta}$ (and hence for every $x\in U_i$, $\beta_i \leq |f(x)| \leq
    c\beta_i$).

\item Let $g'_j$ be the function on input $x$ that outputs $\gamma\cdot |g(x)|$ if $x\in T_j$ and $0$ otherwise, where $\gamma \leq 1$ is a scaling factor
    that ensures that $\E_{x \in T_j} g'(x)^2 $  equals exactly $\beta_{i_j}^2/4$.

\item We define $g_j(x) = \max \{ 0 , g'_j(x) - \sum_{k<j} g_k(x) \}$.

\end{enumerate}

Note that the second step ensures that $g'_j(x) \leq |g(x)|$, while the third step ensures that $g_1(x) + \cdots + g_j(x) \leq g'_j(x)$ for all $j$, and in
particular $g_1(x) + \cdots + g_J(x) \leq |g(x)|$.  Hence the only thing left to prove is the following:

\begin{claim} $\E_{x\in V} g_j(x)^q \geq e\alpha_{i_j}/(10c)^{q/2}$
\end{claim}

\begin{proof} Recall that for every $i$, $\alpha_i = \mu_i \E_{x\in U_i} f(x)^q$, and hence  (using $f(x) \in [\beta_i,c\beta_i)$ for $x\in U_i$):
\begin{equation}
 \mu_i\beta_i^q \leq \alpha_i \leq \mu_ic^q\beta_i^q \;.  \label{eq:alpha-i-bound}
 \end{equation}

Now fix $T=T_j$. Since $\E_{x\in V} g_j(x)^q$ is at least (in fact equal) $\mu(T)\E_{x\in T} g_j(x)^q$ and $\mu(T) = e\mu(U_{i_j})$, we can use
(\ref{eq:alpha-i-bound}) and $\E_{x\in T} g_j(x)^q \geq (E_{x\in T} g_j(x)^2)^{q/2}$, to reduce proving the claim to showing the following:
\begin{equation}
\E_{x\in T} g_j(x)^2  \geq (c\beta_{i_j})^2/(10c^2) = \beta_{i_j}^2/10 \;. \label{eq:lower-bound-gj}
\end{equation}

We know that  $\E_{x\in T} g'_j(x)^2  = \beta_{i_j}^2/4$. We claim  that (\ref{eq:lower-bound-gj}) will follow by showing that for every $k<j$,
\begin{equation}
\E_{x\in T} g'_k(x)^2  \leq 100^{-i'}\cdot \beta_{i_j}^2/4 \;, \label{eq:upper-bound-gk}
\end{equation}
where $i' = i_k - i_j$. (Note that $i'>0$ since in our construction the indices $i_1,\ldots,i_J$ are sorted in descending order.)

Indeed, (\ref{eq:upper-bound-gk}) means that if we let momentarily $\norm{g_j}$ denote $\sqrt{\E_{x\in T} g_j(x)^2}$ then
\begin{equation}
  \norm{g_j} \geq \norm{g'_j} - \norm{\tsum_{k<j} g_k} \geq
  \norm{g'_j}-\sum_{k<j}\norm{g_k}\ge  \norm{g'_j}(1-  \sum_{i'=1}^{\infty} 10^{-i'}) \geq  0.8\norm{g'_j} \;. \label{eq:norm-subtract}
\end{equation}
The first inequality holds because we can write $g_j$ as $g'_j - h_j$, where $h_j=\min\set{g'_j,\sum_{k<j}g_k}$.
Then, on the one hand, $\norm{g_j}\ge \norm{g'_j}-\norm{h_j}$, and on the other hand, $\norm{h_j}\le \norm{\sum_{k<j}g_k}$ since $g'_j\ge 0$.
The second inequality holds because $\norm{g_k}\le \norm{g'_k}$.
By squaring (\ref{eq:norm-subtract}) and plugging in the value of $\norm{g'_j}^2$ we get (\ref{eq:lower-bound-gj}).

\paragraph{Proof of (\ref{eq:upper-bound-gk})} By our construction, it must hold that
\begin{equation}
c^{10}\alpha_{i_k}/2^{i'} \leq \alpha_{i_j}  \;, \label{eq:sse-norm-relate-levelsets}
\end{equation}
since otherwise the index $i_j$ would have been removed from the $\cI$ at the $k^{th}$ step. Since $\beta_{i_k} = \beta_{i_j}c^{i'}$, we can plug
(\ref{eq:alpha-i-bound}) in (\ref{eq:sse-norm-relate-levelsets}) to get
\[
\mu_{i_k}c^{10+4i'}/2^{i'} \leq c^4\mu_{i_j}
\]
or
\[
\mu_{i_k} \leq \mu_{i_j}(2/c)^{4i'}c^{-6} \;.
\]

Since $|T_i| = e|U_i|$ for all $i$, it follows that $|T_{k}|/|T| \leq (2/c)^{4i'}c^{-6}$. On the other hand, we know that $\E_{x\in T_{k}} g'_k(x)^2 =
\beta_{i_k}^2/4 = c^{2i'}\beta^2_{i_j}/4$. Thus,
\[
\E_{x\in T} g'_k(x)^2  \leq 2^{4i'}c^{2i'-4i'-6}\beta_{i_j}^2/4 \leq (2^4/c^2)^{i'}\beta_{i_j}^2/4 \;,
\] and now we just choose $c$ sufficiently large so that $c^2/2^4 > 100$.
\end{proof}
\end{proof}

\subsection{Norm bound implies small-set expansion} \label{app:hyper-imp-sse}

In this section, we show that an upper bound on $2\to q$ norm of the projector to the top eigenspace of a graph implies that the graph is a small-set expander.
This proof appeared elsewhere implicitly~\cite{KhotV05,ODonnell07} or explicitly~\cite{BarakGHMRS12,BarakBHKSZ12} and is presented here only for completeness.
Fix a graph $G$ (identified with its normalized adjacency matrix), and $\lambda \in (0,1)$, letting $V_{ \geq
\lambda}$ denote the subspace spanned by eigenfunctions with eigenvalue at least $\lambda$.

If $p,q$ satisfy $1/p+1/q=1$ then $\norm{x}_p = \max_{y: \norm{y}_q \leq 1} |\iprod{x,y}|$. Indeed, $|\iprod{x,y}| \leq \norm{x}_p\norm{y}_q$ by \Holder's
inequality, and by choosing $y_i = \sign(x_i)|x_i|^{p-1}$ and normalizing one can see this equality is tight. In particular, for every $x \in L(\cU)$,
$\norm{x}_q = \max_{y:\norm{y}_{q/(q - 1)}\leq 1} |\iprod{x,y}|$ and $\norm{y}_{q/(q - 1)} = \max_{\norm{x}_q \leq 1} |\iprod{x,y}|$. As a consequence
\[
\norm{A}_{2\to q} = \max_{\norm{x}_2 \leq 1} \norm{Ax}_q = \max_{\norm{x}_2 \leq 1, \norm{y}_{q/(q - 1)} \leq 1} |\iprod{Ax,y}| = \max_{\norm{y}_{q/(q - 1)}\leq 1} |\iprod{A^Ty,x}| = \norm{A^T}_{q/(q- 1)\to 2}
\]

Note that if $A$ is a projection operator, $A=A^T$. Thus, part~1 of Theorem~\ref{thm:sse-hyper} follows from the following lemma:

\begin{lemma}\label{lem:hyper-to-sse} Let $G=(V,E)$ be  regular graph and $\lambda \in (0,1)$. Then, for every $S \subseteq V$,
\[
\bd(S) \geq 1-\lambda - \norm{V_{\lambda}}_{q/(q-1) \to 2}^2\mu(S)^{(q-2)/q}
\]
\end{lemma}
\begin{proof} Let $f$ be the characteristic function of $S$, and write $f = f' + f''$ where $f'\in V_{\lambda}$ and  $f''=f-f'$ is the projection to the eigenvectors with value less than $\lambda$. Let $\mu = \mu(S)$. We know that
\begin{equation}
\bd(S) = 1 - \iprod{f,Gf}/\normt{f}^2 = 1 - \iprod{f,Gf}/\mu \;, \label{eq:expansioneval}
\end{equation}
And $\norm{f}_{q/(q-1)} = \left( \E f(x)^{q/(q-1)} \right)^{(q-1)/q} = \mu^{(q-1)/q}$, meaning that $\norm{f'} \leq \norm{V_{\lambda}}_{q/(q-1)\to 2}
\mu^{(q-1)/q}$. We now write
\begin{eqnarray}
\iprod{f,Gf} =  \iprod{f',Gf'} + \iprod{f'',Gf''} \leq \normt{f'}^2 + \lambda\normt{f''}^2 &\leq&  \norm{\cV}_{q/(q-1)\to 2}^2\norm{f}_{q/(q-1)}^2 + \lambda\mu \nonumber \\ &\leq& \norm{\cV}_{2\to q}^2 \mu^{2(q-1)/q} + \lambda \mu\mper
\end{eqnarray}
Plugging this into (\ref{eq:expansioneval}) yields the result.
\end{proof}

\end{document}